\documentclass[a4paper, 11pt]{article}

\usepackage[utf8]{inputenc}
\usepackage[T1]{fontenc}
\usepackage[margin=1in]{geometry}

\usepackage[numbers, compress, sort]{natbib}
\usepackage{hyperref}
\usepackage{url}
\usepackage{booktabs}
\usepackage{siunitx}
\usepackage{amsfonts}
\usepackage{amsmath, amssymb, amsthm}
\usepackage{caption}
\usepackage{subcaption}
\usepackage{bbm}
\usepackage[ruled, linesnumbered]{algorithm2e}
\usepackage{thm-restate}
\usepackage{enumitem}
\usepackage{microtype}
\usepackage[most]{tcolorbox}
\usepackage{xcolor}
\usepackage{tikz}
\usetikzlibrary{positioning, arrows.meta}
\usepackage{graphicx}
\usepackage{rotating}
\usepackage{colortbl}

\newcommand{\meansd}[2]{#1\,{\scriptstyle \pm\,#2}}
\definecolor{gold}{RGB}{255, 185, 0}
\definecolor{silver}{RGB}{192, 192, 192}
\definecolor{bronze}{RGB}{205, 127, 50}

\newcommand{\nonl}{\renewcommand{\nl}{\let\nl\oldnl}}
\newcommand{\bluecode}[1]{\small\texttt{\color{blue}#1}}
\SetCommentSty{bluecode}

\SetKwInput{KwParam}{Parameters}
\let\oldnl\nl

\definecolor{datasetbody}{gray}{0.97}
\definecolor{datasettitle}{gray}{0.90}
\newtcolorbox{datasetbox}[1][]{
	colback=datasetbody,
	colframe=white,
	colbacktitle=datasettitle,
	coltitle=black,
	boxrule=0pt,
	arc=3mm,
	left=6pt,
	right=6pt,
	top=6pt,
	bottom=6pt,
	enhanced,
	breakable,
	fonttitle=\bfseries,
	title=#1,
	before skip=10pt,
	after skip=10pt
}
\newlist{datasetitems}{itemize}{1}
\setlist[datasetitems]{
	label={},
	leftmargin=0pt,
	itemsep=0.6em,
	topsep=0.6em,
	parsep=0pt
}

\newtheorem{definition}{Definition}
\newtheorem{remark}{Remark}

\newtheorem{proposition}{Proposition}
\newtheorem{lemma}{Lemma}

\newcommand{\eqn}{Eq.}

\newcommand{\rank}{r}
\newcommand{\order}{p}
\newcommand{\dimension}{d}

\newcommand{\dmdop}{\Theta}
\newcommand{\neff}{n_{\mathrm{eff}}}

\newcommand{\N}{\mathbb{N}}
\newcommand{\Z}{\mathbb{Z}}
\newcommand{\R}{\mathbb{R}}
\newcommand{\C}{\mathbb{C}}
\newcommand{\varprocess}[2]{\mathrm{VAR}_{#1}(#2)}
\renewcommand{\P}[1]{\mathbb{P}\left( #1 \right)}
\newcommand{\E}[1]{\mathbb{E}\left[ #1 \right]}
\newcommand{\Var}[1]{\mathrm{Var}\left( #1 \right)}
\newcommand{\Cov}[2]{\mathrm{Cov}\!\left( #1, #2 \right)}
\newcommand{\transp}[1]{#1^{\mathsf{T}}}
\newcommand{\herm}[1]{#1^{\mathsf{H}}}
\newcommand{\conj}[1]{#1^{\ast}}
\newcommand{\norm}[1]{\left \lVert #1 \right \rVert}
\newcommand{\abs}[1]{\left \lvert #1 \right \rvert}
\DeclareMathOperator*{\argmax}{arg\,max}
\DeclareMathOperator*{\argmin}{arg\,min}
\DeclareMathOperator{\Tr}{tr}
\DeclareMathOperator{\Rpart}{\operatorname{Re}}

\DeclareMathOperator{\diag}{diag}

\newcommand{\arlzero}{\mathtt{ARL}_0}
\newcommand{\arlone}{\mathtt{ARL}_1}
\newcommand{\precision}{\mathtt{P}}
\newcommand{\recall}{\mathtt{R}}
\newcommand{\fone}{\mathtt{F}_1}

\title{CHASM: Online Changepoint Detection in \\ Temporal and Cross-Variable Dependence}

\author{Victor K.~Khamesi  \and  Edward A. K.~Cohen \and Niall M.~Adams \and Dean A.~Bodenham}

\date{{\normalsize Department of Mathematics, Imperial College London,} \\
{\small \texttt{victor.khamesi21@imperial.ac.uk},
\quad \texttt{e.cohen@imperial.ac.uk},
\quad \texttt{n.adams@imperial.ac.uk},
\quad \texttt{dean.bodenham@imperial.ac.uk}}
}

\begin{document}

\maketitle

\begin{abstract}
	Changepoint detection identifies times when the generative process of a time series changes, with applications in healthcare, cybersecurity, and finance. In multivariate settings, changes in cross-variable and temporal dependence are particularly challenging to detect, as they are often less pronounced than shifts in marginal statistics such as the mean or variance. Existing methods detect changes using reconstruction error, which provides only an indirect measure of dynamical change, or rely on scalar functionals that may be too coarse to capture global structure. We introduce CHASM, an online nonparametric method that monitors the truncated eigenvalue sequence of the recursively estimated dynamic mode decomposition operator. Designing such an approach raises two challenges: the permutation invariance of eigendecompositions, resolved via optimal linear assignment, and the lack of online changepoint methods for multivariate complex-valued time series, addressed through a novel augmented monitoring scheme. We study the theoretical properties of the dynamics estimator under the canonical vector autoregressive model, which directly motivates our algorithmic design. The proposed method achieves competitive or superior performance to modern competitors across synthetic and real-world data sets, including challenging settings in video and text data. It is unsupervised, depends on a small number of interpretable parameters, and requires no distributional assumptions beyond finite moments, making it readily deployable across scientific domains.
\end{abstract}

\section{Introduction}
\label{sec:intro}

Modern data streams arise across domains as diverse as finance \citep{berens2015testing}, climate monitoring \citep{shi2022changepoint}, neuroscience \citep{stoehr2021detecting}, natural language processing \citep{wang2018real}, and video understanding \citep{zheng2025blast}. In each of these settings, the data-generating process is rarely stationary: market regimes shift, environmental conditions evolve, neural activity varies with cognitive state, text reflects changing topics, and video transitions between scenes. Detecting such changes promptly is critical for timely decision-making, such as rebalancing portfolios \citep{aue2009break} or segmenting recordings \citep{michel2012towards}. Changepoint detection seeks to identify times at which the underlying generative model of a time series changes.

Two settings are commonly considered: offline approaches retrospectively locate changepoints given the full sequence, while online methods process observations sequentially and detect changes as quickly as possible after they occur. This paper focuses on the online regime, where algorithms must process each observation in constant time with memory independent of the sequence length. These computational constraints make online changepoint detection fundamentally more challenging.

In many applications, observations are multivariate and the structure of interest lies in the joint evolution of variables over time. The vector autoregressive (VAR) process is a standard model for such data, used across econometrics \citep{stock2001vector}, neuroscience \citep{gorrostieta2012investigating}, biology \citep{michailidis2013autoregressive}, and climate science \citep{jiang2021ultra}. The parameter of interest in a VAR process is the transition matrix, which encodes both temporal and cross-variable dependence. Changes in this matrix can be diverse or subtle, and need not necessarily produce clear shifts in the marginal distribution of individual variables.

In practice, the noise distribution is typically unknown and varies across applications, from heavy-tailed in finance \citep{pickands1975statistical} to Poisson in imaging \citep{foi2008practical} and correlated Gaussian in physical systems \citep{langel2024frequency}. This motivates nonparametric, unsupervised methods that make minimal distributional assumptions and generalise across domains without relying on labelled data, which are rarely available.

\subsection{Related work}
\label{sec:intro-related-work}

Online changepoint detection has received considerable attention since the cumulative sum (CUSUM) method \citep{page1954continuous}. Extensions cover multivariate and high-dimensional settings via nonparametric methods \citep{mei2010efficient}, Gaussian models \citep{chan2017optimal, chen2022high}, and probabilistic approaches \citep{adams2007bayesian, knoblauch2018spatio}. While these provide strong baselines, they target distributional rather than dynamical changes. More recently, a related method \citep{alanqary2021change} constructs a subspace-based CUSUM statistic to detect changes in dynamics.

Changepoint detection in VAR processes focuses on the transition matrix. Offline methods locate changepoints via low-rank and sparse decompositions \citep{bai2020multiple, bai2023multiple} or likelihood-ratio statistics over candidate segments \citep{enikeeva2025change}. In the online setting, Tian \& Safikhani \citep{tian2024sequential} estimate the transition matrix via $\ell_1$-penalisation \citep{basu2015regularized} and construct a detection statistic from one-step-ahead reconstruction errors. While the method comes with theoretical guarantees and shows strong empirical performance, it assumes sub-Gaussian noise with diagonal covariance and detects changes indirectly, potentially making it difficult to distinguish dynamical from noise changes.

Dynamic mode decomposition (DMD) has also been applied to changepoint detection. In \citep{kawahara2016dynamic}, an online kernel DMD method is proposed with a CUSUM statistic based on prediction error, while in \citep{takeishi2017learning} an offline autoencoder-based variant is developed; in both works, changepoint detection is a secondary application rather than the primary contribution. Recently, in \citep{khamesi2024online}, DMD is applied on time-delay embedded data and detects changes via reconstruction error in seasonal time series. Our work addresses the more general setting of arbitrary dynamical structure, and monitors the operator spectrum directly rather than through reconstruction error.

\subsection{Contributions}

We introduce the changepoint detection method CHASM. Our main contributions are as follows:

\begin{enumerate}
	\item We propose monitoring the dominant eigenvalues of a recursively updated dynamics operator, applying rank truncation as a post-estimation step. To the best of our knowledge, this is the first online changepoint method to directly track the operator spectrum. This naturally raises two challenges of independent interest which we address below.	
	\item We introduce a method to first align consecutive vectors of complex eigenvalues via a transport argument and second to detect changes in this multivariate complex-valued sequence.
	\item We establish the consistency of the DMD estimator and characterise its bias under the canonical VAR model with no parametric noise assumptions. Consequently, we reveal a trade-off between estimation accuracy and detection sensitivity.
	\item We evaluate CHASM on synthetic time series covering diverse noise distributions, change magnitudes, and dimensions, and on real-world data sets spanning motion, image, and text modalities, showing competitive to superior performance against modern competitors. While our theoretical analysis is focused on VAR processes, it is shown that in practice CHASM has good detection performance beyond this framework.
\end{enumerate}

\subsection{Organisation}

The rest of this paper is organised as follows. Section~\ref{sec:background} defines the changepoint detection problem under the VAR framework and reviews key elements of DMD. Section~\ref{sec:methodology} presents CHASM in three steps: online dynamics estimation, optimal eigenvalue alignment, and complex-valued changepoint detection. Section~\ref{sec:theory} establishes consistency and bias results for both the unweighted and weighted estimators. Section~\ref{sec:experiments} evaluates CHASM on synthetic and real-world data sets beyond the VAR model and compares it to recent competitors.

\paragraph{Notation.} 
Throughout this paper, for a matrix $M$, we denote by $\transp{M}$ its transpose, $\herm{M}$ its conjugate transpose, $\Tr(M)$ its trace, and $\zeta(M)$ its spectral radius; $I_\dimension$ is the identity matrix. We write $M \succ 0$ to denote that $M$ is positive definite. The notation $X \perp Y$ indicates independence between random variables $X$ and  $Y$. Unless stated, $\norm{\cdot}$ is the $\ell_2$ or spectral norm, $\norm{\cdot}_F$ the Frobenius norm. 

\section{Background}
\label{sec:background}

We study changepoint detection in temporal and cross-variable dependence, and present a general method for detecting such changes in Section~\ref{sec:methodology}. For analytical tractability, we formulate the detection problem using VAR processes. The VAR framework is highly flexible; Wold’s decomposition theorem (see, e.g. Section~4.8 of \citep{hamilton1994time}) shows that any stationary process can be approximated arbitrarily well by a VAR model. This framework enables our theoretical analysis in Section~\ref{sec:theory}, while our experiments demonstrate strong performance beyond settings typically modelled by VAR in Section~\ref{sec:experiments}.

\subsection{Vector autoregressive processes}
\label{sec:background-var}

A $\dimension$-dimensional first-order vector autoregressive process, $\varprocess{\dimension}{1}$, follows the recursion
\begin{equation}
	\label{eq:var1}
	x_t = \Theta \, x_{t-1} + \varepsilon_t, \qquad t \in \N,
\end{equation}
where $\Theta \in \R^{\dimension \times \dimension}$ is the unknown transition matrix and $\{\varepsilon_t\}$ is a white noise process with $\E{\varepsilon_t} = 0$ and $\Cov{\varepsilon_t}{ \varepsilon_s} = \mathbbm{1}_{t=s} \, \Sigma$, for some positive semi-definite matrix $\Sigma \in \R^{\dimension \times \dimension}$. The process $\{x_t\}$ is weakly stationary if and only if $\zeta(\Theta)<1$. Moreover, any $\varprocess{\dimension}{\order}$ process of order $\order > 1$ admits a standard companion-form representation as in \eqn~\eqref{eq:var1}. Therefore, we restrict attention to $\varprocess{\dimension}{1}$ throughout, without loss of generality, following \cite{basu2019low}; see Appendix~\ref{app:background-var} for details.

\paragraph{Problem definition.}
We study online changepoint detection in the transition dynamics of \eqn~\eqref{eq:var1}. Given a stream $x_1, x_2, \dots$, where $x_t \in \R^\dimension$, we test $\mathcal{H}_0$ against the single change alternative $\mathcal{H}_1$
\begin{equation}
	\label{eq:cpd-problem}
	\begin{aligned}
		\mathcal{H}_0 &: x_t = \Theta_0 \, x_{t-1} + \varepsilon_t \text{ for all } t \in \N, \\
		\mathcal{H}_1 &: \exists \, \tau \in \N \text{ such that } x_t = \begin{cases}
			\Theta_0 \, x_{t-1} + \varepsilon_t, & t \leq \tau \\
			\Theta_1 \, x_{t-1} + \varepsilon_t, & t > \tau
		\end{cases},
	\end{aligned}
\end{equation}
with $\Theta_0 \neq \Theta_1$ unknown. Changepoint detection is concerned with estimating $\tau$ as closely as possible. Evaluation metrics are defined in Section~\ref{sec:experiments-setup}. Although \eqn~\eqref{eq:cpd-problem} formalises a single changepoint, our method extends to multiple changepoints by restarting after each detection; see Section~\ref{sec:experiments-real}.

Crucially, a change in $\Theta$ need not manifest in the marginal distribution of $x_t$. Two $\mathrm{VAR}_d(1)$ processes with diagonal transition $\Theta_0 = \bigl(\begin{smallmatrix} \theta & 0 \\ 0 & \theta \end{smallmatrix}\bigr)$ and rotation transition $\Theta_1 = \bigl(\begin{smallmatrix} 0 & -\theta \\ \theta & 0 \end{smallmatrix}\bigr)$ share identical marginal covariances $(1-\theta^2)^{-1} I_2$ for any $\abs{\theta}<1$, yet differ fundamentally in their dynamics. This motivates the development of statistics that directly track the operator rather than marginal properties.

\paragraph{Spectrum as a surrogate statistic.}
Testing matrix equality is fundamentally different from vector comparisons since the rows and columns are structured via properties such as rank and eigenstructure, so treating a matrix as an unordered vector ignores those dependencies. This motivates the use of surrogate statistics. Reconstruction error \citep{khamesi2024online,tian2024sequential} introduces an additional layer of indirection between the parameter of interest and the detection statistic. Scalar functionals \citep{cai2013two,killick2012optimal,li2012two} collapse the operator to a single number, potentially discarding directional and oscillatory structure. We propose to monitor the dominant eigenvalues of $\Theta$ as an interpretable and direct operator-level summary that captures stability and oscillatory behaviour while reducing monitoring dimension; this trade-off constitutes a central methodological idea of this work.

\subsection{Dynamic mode decomposition}
\label{sec:background-dmd}

Dynamic mode decomposition (DMD) is a data-driven method for modelling dynamical systems \citep{schmid2010dynamic} that extracts a collection of modes each associated with a growth rate and an oscillation frequency.

\paragraph{Modes and dynamics.}
Given observations $x_1, \dots, x_{T}$, DMD seeks an operator $A$ such that $Y = AX$, where $X = \left( x_1, \dots, x_{T-1} \right)$ and $Y=\left( x_2, \dots, x_T \right)$. In the original formulation of \citep{schmid2010dynamic}, developed for high-dimensional fluid systems, the matrix $A$ is assumed to be low-rank. This motivates an approximation via the rank-truncated singular value decomposition (SVD) of $X = U_\rank S_\rank \herm{V_\rank}$, yielding the operator $A_\rank = \herm{U_\rank} Y V_\rank S_\rank^{-1} \in \C^{\rank \times \rank}$, whose eigendecomposition defines the modes and dynamics; a full derivation of the modal decomposition is provided in Appendix~\ref{app:background-dmd}.

\paragraph{Online DMD.}
The operator of \citep{schmid2010dynamic} cannot be updated exactly without storing all past observations and recomputing the SVD at each time step on the full sequence. To address this, \citep{zhang2019online} drop the low-rank assumption and maintain sufficient statistics to update the least-squares estimate recursively. The resulting online operator is mathematically equivalent to the batch solution of \citep{schmid2010dynamic} in the absence of rank truncation. In this work, we propose to combine the recursive estimation of \citep{zhang2019online} with the truncation principle of \citep{schmid2010dynamic}, applying regularisation as a post-estimation step to the online operator rather than during batch decomposition. This decouples regularisation from estimation, yielding a compact spectral representation updated in a single sequential pass.

\section{CHASM: Complex Hungarian-Aligned Spectrum Monitoring}
\label{sec:methodology}

\subsection{Online dynamics learning}
\label{sec:methodology-learning}

Consider a stream $x_1, x_2, \dots$, where $x_n \in \R^\dimension$ for some dimension $\dimension \in \N$. We estimate a sequence of linear operators $\dmdop_n \in \R^{\dimension \times \dimension}$ using recursive least squares, corresponding to online DMD; see Section~\ref{sec:background-dmd} and \citep{zhang2019online}. For a fixed forgetting factor $\rho \in (0, 1]$, the update equations for $\Theta_n$ are
\begin{equation}
	\label{eq:online-dmd-updates}
	\begin{aligned}
		\gamma_n &= \left(1 + \frac{1}{n-1} \transp{x_{n-1}} \Gamma_{n-1}^{-1} x_{n-1}\right)^{-1} \in \R, \\
		\Gamma_n^{-1} &= \frac{n}{\rho (n-1)} \left( \Gamma_{n-1}^{-1} - \frac{\gamma_n}{n-1} \Gamma_{n-1}^{-1} x_{n-1} \transp{x_{n-1}} \Gamma_{n-1}^{-1} \right) \in \R^{\dimension \times \dimension}, \\
		\Theta_n &= \Theta_{n-1} + \frac{\gamma_n}{n-1} \left( x_n - \Theta_{n-1} x_{n-1} \right) \transp{x_{n-1}} \Gamma_{n-1}^{-1} \in \R^{\dimension \times \dimension}.
	\end{aligned}
\end{equation}
Smaller values of $\rho$ emphasise recent observations for faster adaptation, while values closer to one yield more stable estimates and retain longer-term information. The recursion of \eqn~\eqref{eq:online-dmd-updates} is initialised by setting $\dmdop_0$ arbitrarily and $\Gamma_0^{-1} = \epsilon I_\dimension$ for small $\epsilon > 0$, or warm-started from a batch of at least $\dimension$ observations. Higher-order dynamics are handled by concatenating the last $\order \in \N$ samples into an augmented state vector, following Lemma~\ref{thm:varp-var1}. We restrict to order $\order = 1$ throughout, which provides sufficiently strong empirical performance across all settings considered.

\paragraph{Post-hoc regularisation.}
At each step, we compute a truncated eigendecomposition of $\dmdop_n$, retaining the $\rank \leq \dimension$ dominant eigenvalues ordered by decreasing magnitude, $\lambda_n = \transp{(\lambda_n^{(1)}, \dots, \lambda_n^{(\rank)})} \in \C^\rank$. The truncation level $r$ acts as a regularisation parameter controlling the complexity of the spectral representation; unlike \citep{tian2024sequential}, which regularise during estimation via $\ell_1$-penalisation without DMD, we leave estimation unconstrained and regularise post-hoc via spectrum truncation.

The stream $\{x_t\} \subset \R^\dimension$ is thus transformed into spectral summaries $\{\lambda_t\} \subset \C^\rank$, on which changepoint detection is performed. Two challenges arise: the eigendecomposition returns eigenvalues in no canonical order, so in practice numerical estimates may be inconsistently ordered across time steps producing artificially discontinuous trajectories; and changepoint detection for multivariate complex-valued sequences remains largely underexplored. Sections~\ref{sec:methodology-alignment} and~\ref{sec:methodology-detection} address these challenges.

\subsection{Optimal eigenvalue alignment}
\label{sec:methodology-alignment}

Since numerical routines do not guarantee consistent ordering of eigenvalues across time, the resulting trajectories may be artificially discontinuous. This is particularly pronounced for complex conjugate pairs, where estimation noise, especially early in the recursion, can induce arbitrary switching. We formalise the alignment problem as optimal transport between consecutive spectra.

\begin{restatable}{lemma}{hungarianotpermutation}
	\label{thm:methodology-hungarian-ot-permutation}
	Let $\nu_n$ be the uniform empirical measure over $\lambda_n \in \C^\rank$, and define $\nu_{n-1}$ analogously. The squared $2$-Wasserstein distance between $\nu_{n-1}$ and $\nu_n$ satisfies
	\begin{equation*}
		W_2^2(\nu_{n-1}, \nu_n) = \min_{\pi \in S_\rank} \frac{1}{\rank} \sum_{i=1}^\rank \abs{\lambda_{n-1}^{(i)} - \lambda_n^{\pi(i)}}^2, \qquad \nu_n = \frac{1}{\rank}\sum_{i=1}^\rank \delta_{\lambda_n^{(i)}},
	\end{equation*}
	where $\delta$ denotes the Dirac measure and $S_\rank$ the set of permutations of $\{1,\dots,\rank\}$.
\end{restatable}

Lemma~\ref{thm:methodology-hungarian-ot-permutation} shows that the optimal coupling is a permutation map. Therefore, alignment reduces to a search over $S_\rank$. However, a brute force approach is expensive at $\rank!$ evaluations, even for moderate $\rank$.

\begin{restatable}{proposition}{hungarianassignment}
	\label{thm:methodology-hungarian-assignment}
	Let $\Lambda_n = \diag(\lambda_n) \in \C^{\rank\times\rank}$, $\mathcal{S}_\rank \subset \R^{\rank \times \rank}$ the set of permutation matrices in \eqn~\eqref{eq:permutation-matrices}, and $\mathcal{B}_\rank \subset \R^{\rank \times \rank}$ the set of doubly stochastic matrices in \eqn~\eqref{eq:doubly-stochastic-matrices}. Then
	\begin{equation}
		\label{eq:w2-permutation}
		\Pi_n = \argmin_{\Pi \in \mathcal{S}_\rank} \norm{\transp{\Pi} \Lambda_n \Pi - \Lambda_{n-1}}_F^2 = \argmax_{\Pi \in \mathcal{B}_\rank} \Tr{\!\left(\Pi \Rpart{\!\left(\conj{\lambda_n} \transp{\lambda_{n-1}}\right)}\right)}.
	\end{equation}
\end{restatable}

Proposition~\ref{thm:methodology-hungarian-assignment} reformulates alignment as a linear assignment task, solvable in polynomial time by relaxing $\mathcal{S}_\rank$ to its convex hull $\mathcal{B}_\rank$. The earliest method to solve the assignment problem in polynomial time is the Kuhn-Munkres or Hungarian algorithm \citep{munkres1957algorithms}, which inspires the name of our approach and is $\mathcal{O}(\rank^4)$. In our implementation, we use \citep{jonker1988shortest} which is $\mathcal{O}(\rank^3)$. This result can independently be viewed as an extension of the orthogonal Procrustes problem to permutations over complex-valued points. Proofs of Lemma~\ref{thm:methodology-hungarian-ot-permutation} and Proposition~\ref{thm:methodology-hungarian-assignment} are provided in Appendix~\ref{app:proofs-alignment}.

\subsection{Complex multivariate changepoint detection}
\label{sec:methodology-detection}

Following Sections~\ref{sec:methodology-learning} and~\ref{sec:methodology-alignment}, the sequence of eigenvalue vectors $\{\lambda_t\}$ is aligned and suitable for changepoint detection. We propose to extend the multivariate exponentially weighted moving average (MEWMA) scheme of \citep{lowry1992multivariate} to the complex domain.

Rather than monitoring $\{\lambda_t\}$ directly, we track the first-order differences $v_n = \lambda_n - \lambda_{n-1} \in \C^\rank$, which can be interpreted as the instantaneous spectral velocity. Under stationary dynamics, $v_n$ fluctuates around zero with small magnitude and no preferred direction; when a change occurs, it induces a coherent directional shift detectable as a sustained deviation from zero. This behaviour is illustrated in Figure~\ref{fig:spectral-velocity}.

The second-order statistics of the eigenvalue velocity $v_n$ require both its Hermitian covariance $\Sigma_{v_n} = \E{(v_n - \mu_{v_n}) \herm{(v_n - \mu_{v_n})}}$ and pseudo-covariance $\tilde\Sigma_{v_n} = \E{(v_n - \mu_{v_n})\transp{(v_n - \mu_{v_n})}}$, where $\mu_{v_n} = \E{v_n} \in \C^\rank$ and $\Sigma_{v_n}, \tilde \Sigma_{v_n} \in \C^{\rank \times \rank}$ are all estimated sequentially from the stream. We form the augmented vector $\underline v_n = (v_n, \conj{v_n}) \in \C^{2\rank}$ and apply the recursion
\begin{equation}
	\label{eq:ewma-recursion}
	\underline z_n = (1-\alpha)\,\underline z_{n-1} + \alpha\,\underline v_n, \qquad \alpha \in (0, 1),
\end{equation}
with augmented mean $\mu_{\underline z_n} = (\mu_{v_n}, \conj{\mu_{v_n}}) \in \C^{2\rank}$ and covariance 
\begin{equation*}
	\Sigma_{\underline z_n}  = \beta_n \begin{pmatrix}
		\Sigma_{v_n} & \tilde \Sigma_{v_n} \\
		\conj{\tilde \Sigma_{v_n}} & \conj{\Sigma_{v_n}}
	\end{pmatrix} \in \C^{2\rank \times 2\rank}, \qquad \beta_n = \frac{\alpha(1 - (1-\alpha)^{2n})}{2-\alpha} \to \frac{\alpha}{2-\alpha} \text{ as } n \to \infty.
\end{equation*}
The changepoint statistic is defined as the squared Mahalanobis distance
\begin{equation}
	\label{eq:change-statistic}
	\mathcal{D}_n^2 = \herm{(\underline z_n - \mu_{\underline z_n})} \Sigma_{\underline z_n}^{-1} (\underline z_n - \mu_{\underline z_n}),
\end{equation}
with stopping rule $\hat\tau = \inf_n\{\mathcal{D}_n^2 > h\}$ for some threshold $h > 0$. Since $\Sigma_{\underline z_n}$ is positive semi-definite, $\mathcal{D}_n^2$ is guaranteed to be non-negative real despite the complex-valued inputs.

\begin{restatable}{lemma}{ewmadecomposition}
	\label{thm:methodology-ewma-decomposition}
	Let $z_n \in \C^\rank$ such that $z_n = (1 - \alpha) \, z_{n-1} + \alpha \, v_n$, similarly to \eqn~\eqref{eq:ewma-recursion}. Then,
	\begin{equation*}
		\mathcal D_n^2 = \frac{2}{\beta_n} \Rpart{\left( \herm{(z_n - \mu_{v_n})} \mathcal{C}_n^{-1} (z_n - \mu_{v_n}) + \herm{(z_n - \mu_{v_n})} \mathcal{Q}_n \conj{(z_n - \mu_{v_n})} \right)},
	\end{equation*}
	where $\mathcal{C}_n = \Sigma_{v_n} - \tilde \Sigma_{v_n} (\conj{\Sigma_{v_n}})^{-1} \conj{\tilde \Sigma_{v_n}} \in \R^{\rank \times \rank}$ and $\mathcal{Q}_n = - \mathcal{C}_n^{-1} \tilde \Sigma_{v_n} (\conj{\Sigma_{v_n}})^{-1} \in \R^{\rank \times \rank}$.
\end{restatable}

Lemma~\ref{thm:methodology-ewma-decomposition}, proved in Appendix~\ref{app:proofs-detection}, reduces the $2\rank$-dimensional inversion to two $\rank$-dimensional operations, slightly reducing the computational cost. When $v_n$ is proper complex, i.e. $\tilde\Sigma_{v_n} = 0$, $\mathcal{Q}_n = 0$ and $\mathcal{D}_n^2$ reduces to the standard Mahalanobis distance on $z_n$, recovering MEWMA \citep{lowry1992multivariate} as a special case. The eigenvalue velocities are generically non-circular \citep{schreier2010statistical} since complex conjugate pairs impose linear constraints between $v_n$ and $\conj{v_n}$, making augmentation necessary in general.

\paragraph{Algorithm summary.}
Algorithm~\ref{alg:chasm} summarises CHASM with minimal initialisation; see Section~\ref{sec:methodology-learning} for further discussion. The role of the forgetting factor $\rho$ is discussed in Section~\ref{sec:theory}. The rank $\rank$ is a regularisation parameter; we demonstrate robustness to its choice in Section~\ref{sec:experiments} and refer to the literature for principled rank selection \citep{brunton2022data,gavish2014optimal}. Parameters $(\alpha, h)$ are set jointly to control overall sensitivity, based on values used for real-valued MEWMA \citep{lowry1992multivariate,rigdon1995integral,runger1996markov}.

\begin{algorithm}[t]
	\caption{$\hat{\tau} \gets \texttt{CHASM}(\{x_t\} \mid \rho, \rank, \alpha, h)$}
	\label{alg:chasm}
	\nonl \bluecode{/* Complex Hungarian-Aligned Spectrum Monitoring */} \\
	\KwIn{Stream of observations $\{x_t\} \subset \R^{\dimension}$}
	\KwOut{Detected changepoint $\hat{\tau}$}
	\KwParam{Forgetting factor $\rho \in (0,1]$, rank $\rank \leq \dimension$, learning rate $\alpha \in (0,1)$, threshold $h > 0$}
	$\Theta_0 \gets I_\dimension$, $\Gamma_0^{-1} \gets \epsilon I_\dimension$ for small $\epsilon > 0$ \\
	\For{$n = 1, 2, \dots$}{
		$\gamma_n \gets \left(1 + \frac{1}{n-1} \transp{x_{n-1}} \Gamma_{n-1}^{-1} x_{n-1}\right)^{-1}$ \tcp*[f]{Dynamics estimation, Section~\ref{sec:methodology-learning}} \\
		$\Gamma_n^{-1} \gets \frac{n}{\rho (n-1)} \left( \Gamma_{n-1}^{-1} - \frac{\gamma_n}{n-1} \Gamma_{n-1}^{-1} x_{n-1} \transp{x_{n-1}} \Gamma_{n-1}^{-1} \right)$ \\
		$\Theta_n \gets \Theta_{n-1} + \frac{\gamma_n}{n-1} \left( x_n - \Theta_{n-1} x_{n-1} \right) \transp{x_{n-1}} \Gamma_{n-1}^{-1} $ \\
		$\lambda_n \gets \sigma(\dmdop_n \mid \rank)$ \\
		$\Pi_n \gets \argmax_{\Pi \in \mathcal{B}_\rank} \Tr{\!\left(\Pi \Rpart{\!\left(\conj{\lambda_n} \transp{\lambda_{n-1}}\right)}\right)}$ \tcp*[f]{Optimal alignment, Section~\ref{sec:methodology-alignment}} \\
		$\lambda_n \gets \Pi_n \lambda_n$ \\
		$v_n \gets \lambda_n - \lambda_{n-1}$ \tcp*[f]{Changepoint detection, Section~\ref{sec:methodology-detection}} \\
		$z_n \gets (1 - \alpha) z_{n-1} + \alpha v_n$ \\
		$\mathcal{D}_n^2 \gets \frac{2(2-\alpha)}{\alpha} \Rpart{\left( \herm{(z_n - \mu_{v_n})} \mathcal{C}_n^{-1} (z_n - \mu_{v_n}) + \herm{(z_n - \mu_{v_n})} \mathcal{Q}_n \conj{(z_n - \mu_{v_n})} \right)}$ \\
		\If{$\mathcal{D}_n^2 > h$}{
			\Return{$\hat{\tau} = n$}
		}
	}
\end{algorithm}

\begin{figure}[t]
	\centering
	\begin{minipage}[t]{0.48\linewidth}
		\centering
		\resizebox{\linewidth}{!}{%
			\begin{tikzpicture}[
				>=Stealth,
				axis/.style={->, line width=0.1pt, black!55},
				unitcirc/.style={black!58, line width=0.9pt, densely dashed},
				velstat/.style={-{Stealth[length=3pt, width=2pt]},
					line width=0.9pt, gray!72, line cap=round},
				velcp/.style={-{Stealth[length=3pt, width=1.8pt]},
					line width=0.9pt, red!65!gray, line cap=round},
				eigA/.style={circle, inner sep=0pt, minimum size=7pt,
					fill=blue!60!gray, draw=none},
				eigB/.style={circle, inner sep=0pt, minimum size=7pt,
					fill=teal!55!gray, draw=none},
				eigAghost/.style={circle, inner sep=0pt, minimum size=7pt,
					fill=blue!60!gray, draw=none, opacity=0.38},
				eigBghost/.style={circle, inner sep=0pt, minimum size=7pt,
					fill=teal!55!gray, draw=none, opacity=0.38},
				axlbl/.style={font=\footnotesize, gray!65},
				]
				
				\def\R{2.0}    
				\def\sep{5.2}  
				
				\begin{scope}[xshift=0cm]
					\draw[unitcirc] (0,0) circle (\R);
					\draw[axis] (-\R-0.32, 0) -- (\R+0.42, 0)
					node[axlbl, above=1pt, color=black!70] {$\mathrm{Re}$};
					\draw[axis] (0, -\R-0.32) -- (0, \R+0.52)
					node[axlbl, above=1pt, color=black!70] {$\mathrm{Im}$};
					
					\coordinate (A1u) at ( 1.31,  1.03);
					\coordinate (A1l) at ( 1.31, -1.03);
					\coordinate (A2u) at (-0.83,  1.17);
					\coordinate (A2l) at (-0.83, -1.17);
					
					\draw[velstat] (A1u) -- +( 0.31, -0.25); 
					\draw[velstat] (A1u) -- +(-0.31, +0.25);
					\draw[velstat] (A1u) -- +(-0.06, -0.39);
					
					\draw[velstat] (A1l) -- +( 0.29,  0.06);
					\draw[velstat] (A1l) -- +(-0.29, -0.26);
					\draw[velstat] (A1l) -- +(-0.39,  0.04);
					
					\draw[velstat] (A2u) -- +( 0.31,  0.25);
					\draw[velstat] (A2u) -- +(-0.38,  0.13);
					\draw[velstat] (A2u) -- +( 0.10, -0.39);
					
					\draw[velstat] (A2l) -- +( 0.38, -0.13);
					\draw[velstat] (A2l) -- +(-0.28,  0.28);
					\draw[velstat] (A2l) -- +(-0.05, -0.40); 
					
					\node[eigA] at (A1u) {};  \node[eigA] at (A1l) {};
					\node[eigB] at (A2u) {};  \node[eigB] at (A2l) {};
					
				\end{scope}
				
				\node[font=\small, anchor=north west, color=black] at (-\R-0.32, \R+0.52) {(a)};
				
				\begin{scope}[xshift=\sep cm]
					\node[font=\small, anchor=north west, color=black] at (-\R-0.32, \R+0.52) {(b)};
					\draw[unitcirc] (0,0) circle (\R);
					\draw[axis] (-\R-0.32, 0) -- (\R+0.42, 0)
					node[axlbl, above=1pt, color=black!70] {$\mathrm{Re}$};
					\draw[axis] (0, -\R-0.32) -- (0, \R+0.52)
					node[axlbl, above=1pt, color=black!70] {$\mathrm{Im}$};
					
					\coordinate (G1u) at ( 1.31,  1.03);
					\coordinate (G1l) at ( 1.31, -1.03);
					\coordinate (G2u) at (-0.83,  1.17);
					\coordinate (G2l) at (-0.83, -1.17);
					
					\coordinate (N1u) at ( 0.76,  1.49);
					\coordinate (N1l) at ( 0.76, -1.49);
					\coordinate (N2u) at (-1.24,  0.72);
					\coordinate (N2l) at (-1.24, -0.72);
					
					\draw[velcp] (1.310,  1.030) to[bend left=12]  (1.127,  1.183);
					\draw[velcp] (1.127,  1.183) to[bend right=9]  (0.944,  1.337);
					\draw[velcp] (0.944,  1.337) to[bend left=11]  (N1u);
					
					\draw[velcp] (1.310, -1.030) to[bend right=12] (1.127, -1.183);
					\draw[velcp] (1.127, -1.183) to[bend left=9]   (0.944, -1.337);
					\draw[velcp] (0.944, -1.337) to[bend right=11] (N1l);
					
					\draw[velcp] (-0.830,  1.170) to[bend right=13] (-0.967,  1.020);
					\draw[velcp] (-0.967,  1.020) to[bend left=10]  (-1.104,  0.870);
					\draw[velcp] (-1.104,  0.870) to[bend right=12] (N2u);
					
					\draw[velcp] (-0.830, -1.170) to[bend left=13]  (-0.967, -1.020);
					\draw[velcp] (-0.967, -1.020) to[bend right=10] (-1.104, -0.870);
					\draw[velcp] (-1.104, -0.870) to[bend left=12]  (N2l);
					
					\node[eigAghost] at (G1u) {};  \node[eigAghost] at (G1l) {};
					\node[eigBghost] at (G2u) {};  \node[eigBghost] at (G2l) {};
					\node[eigA] at (N1u) {};  \node[eigA] at (N1l) {};
					\node[eigB] at (N2u) {};  \node[eigB] at (N2l) {};
					
				\end{scope}
			\end{tikzpicture}
		}
		\caption{\small Spectral velocity under stationary dynamics (a) and at a changepoint (b). Stationarity yields directionally incoherent velocities; a changepoint induces a coherent shift across consecutive eigenvalues.}
		\label{fig:spectral-velocity}
	\end{minipage}
	\hfill
	\begin{minipage}[t]{0.48\linewidth}
		\centering
		\resizebox{\linewidth}{!}{%
			\includegraphics{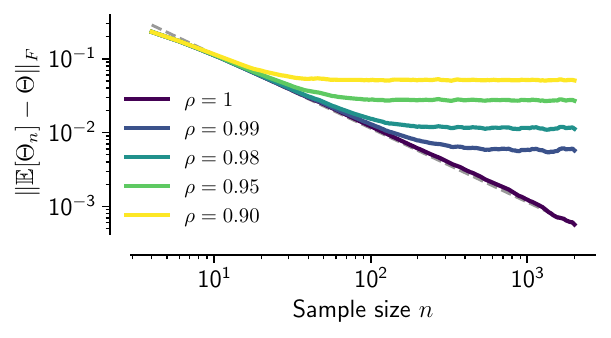}
		}
		\caption{\small Asymptotic bias of the online DMD operator under a $\varprocess{\dimension}{1}$ model ($d=2$, $10^5$ Monte Carlo samples). For $\rho=1$, bias decays as $n^{-1}$ (Theorem~\ref{thm:unweighted-bias}); for $\rho < 1$, it saturates at order $(1-\rho)$ (Theorem~\ref{thm:weighted-bias}).}
		\label{fig:bias}
	\end{minipage}
\end{figure}

\paragraph{Computational complexity.}
The per-step time complexity of CHASM is $\mathcal{O}(\dimension^2 \rank)$ and its space complexity is $\mathcal{O}(\dimension^2)$. Rank-one operator updates cost $\mathcal{O}(\dimension^2)$; extracting $\rank$ dominant eigenvalues costs $\mathcal{O}(\dimension^2 \rank)$ \citep{golub2013matrix,lehoucq1998arpack}, avoiding the cubic cost of a full eigendecomposition; optimal alignment costs $\mathcal{O}(\rank^3)$ \citep{jonker1988shortest}; MEWMA updates and statistic evaluation are likewise $\mathcal{O}(\rank^3)$.

\section{Theoretical analysis}
\label{sec:theory}

In this section, we consider a stationary $\varprocess{\dimension}{1}$ process $x_t = \Theta\,x_{t-1} + \varepsilon_t$ where $\zeta(\Theta) < 1$, $\{\varepsilon_t\}$ is a martingale difference sequence with respect to the natural filtration $\mathcal{F}_t = \sigma(x_0,\dots,x_t)$, satisfying $\E{\varepsilon_t \mid \mathcal{F}_{t-1}} = 0$ almost surely and $\E{\varepsilon_t\transp{\varepsilon_t}} = \Sigma \succ 0$. The stationary covariance $\Gamma = \E{x_t\transp{x_t}}$ satisfies $\Gamma \succ 0$ by Lemma~\ref{thm:gamma-pd} and solves the Lyapunov equation $\Gamma = \Theta\Gamma\transp{\Theta} + \Sigma$; see Appendix~\ref{app:proofs-bias}.

\subsection{Unweighted estimator}
\label{sec:theory-unweighted}

We study the online unweighted DMD operator $\dmdop_n$, obtained from \eqn~\eqref{eq:online-dmd-updates} with $\rho = 1$.

\begin{restatable}{theorem}{unweightedbias}
	\label{thm:unweighted-bias}
	Suppose $\varepsilon_t \perp \varepsilon_s$ for $t \neq s$ and $\E{\norm{\varepsilon_t}^8} < \infty$. Then $\dmdop_n$ is consistent, i.e. $\dmdop_n \xrightarrow{\mathbb{P}} \Theta$, and
	\begin{equation*}
		\E{\dmdop_n} = \Theta + \mathcal{O}(n^{-1}) \qquad \text{as } n \to \infty.
	\end{equation*}
\end{restatable}

The proof is provided in Appendix~\ref{app:proofs-bias-unweighted}. The eighth-moment condition controls higher-order terms and is standard in dependent-data asymptotics \citep{bartz2014covariance, zhang2013communication}; consistency of $\dmdop_n$ requires only finite fourth moments. This condition is nonparametric, strictly weaker than sub-Gaussian noise, and imposes no restriction on the covariance structure, in contrast to \citep{tian2024sequential}. No structural assumption is imposed on $\Theta$ beyond stationarity, which accommodates full-rank, low-rank, oscillatory, and sparse dynamics. 

Theorem~\ref{thm:unweighted-bias} shows that, under stationarity of the data-generating process, the estimator consistently learns the true in-control dynamics, with bias decaying at rate $\mathcal{O}(n^{-1})$. This guarantees that the estimated matrix provides a reliable baseline representation of the in-control process.

\subsection{Weighted estimator}
\label{sec:theory-weighted}

We now consider the online weighted DMD operator $\Theta_n$, given by \eqn~\eqref{eq:online-dmd-updates} with a fixed forgetting factor $\rho \in (0, 1)$. We define the effective sample size $\neff = \sum_{t=1}^n \rho^{n-t} = \frac{1-\rho^n}{1-\rho} \to (1-\rho)^{-1}$ as $n \to \infty$, which quantifies the contribution of recent observations relative to the distant past.

\begin{restatable}{theorem}{weightedbias}
	\label{thm:weighted-bias}
	Suppose $\varepsilon_t \perp \varepsilon_s$ for $t \neq s$ and $\E{\norm{\varepsilon_t}^8} < \infty$. Then for any $\rho \in (0,1)$, there exist $c, c_e > 0$ depending only on $\Theta$, $\Sigma$, $\dimension$, such that, for the event $\mathcal{E}_n = \{\lambda_{\min}(\Gamma_n) \geq \lambda_{\min}(\Gamma)/2\}$,
	\begin{equation*}
		\norm{\E{\left(\dmdop_n - \Theta\right)\mathbbm{1}_{\mathcal{E}_n}}} \leq c \,\neff^{-1}, \qquad \P{\mathcal{E}_n} \geq 1 - c_e\,\neff^{-1},
	\end{equation*}
	and consequently $\E{\dmdop_n} = \Theta + \mathcal{O}(\neff^{-1})$ as $n \to \infty$.
\end{restatable}

Theorem~\ref{thm:weighted-bias} mirrors Theorem~\ref{thm:unweighted-bias} with the sample size $n$ replaced by the effective sample size $\neff$. As $\rho \to 1^-$, one recovers $\neff \sim n$, and the unweighted regime is obtained as a limiting case. In contrast to the unweighted estimator, $\neff$ remains bounded for fixed $\rho < 1$, and the bias bound does not vanish but stabilises at order $\neff^{-1} = (1-\rho)$. These results are illustrated in Figure~\ref{fig:bias}.

\paragraph{Effective sample size and detection.}
Without forgetting, both bias and variance vanish asymptotically, yielding a concentrated estimator that may become insensitive to changepoints. With forgetting, the estimator can retain a non-degenerate variability driven by recent observations, thereby maintaining sensitivity to changes. The forgetting factor $\rho$ controls this trade-off: larger values favour stability and accurate estimation of the in-control regime, while smaller values increase responsiveness to changes at the cost of higher variability. The effective sample size $\neff$ can be interpreted as the number of observations in the estimate. Since the objective of this work is changepoint detection rather than optimal estimation of $\Theta$, maintaining a non-vanishing error provides a principled mechanism to balance stability and adaptivity. This trade-off justifies the role of $\rho$ in Section~\ref{sec:methodology} and is confirmed empirically by the $\arlzero$ results of Section~\ref{sec:experiments}, where smaller $\rho$ produces shorter in-control run lengths. 

\section{Experiments}
\label{sec:experiments}

\subsection{Experimental setup}
\label{sec:experiments-setup}

We evaluate detection speed via the average run length in control ($\arlzero$), estimating the expected time to false alarm, and out of control ($\arlone$), corresponding to the expected detection delay after a change. Accuracy is assessed via precision ($\precision$), recall ($\recall$), and $\fone$-score, which characterise how reliably changes are detected. These metrics are complementary; see Appendix~\ref{app:experiments-metrics} for more details.

We compare CHASM against six competitors: subspace-based methods mSSA and mSSA-MW \citep{alanqary2021change}; our closest VAR-based competitor the method of \citet{tian2024sequential}; a strong high-dimensional baseline method OCD \citep{chen2022high}; a DMD-based approach CPDMD \citep{khamesi2024online}; and a Bayesian method with a VAR prior BOCPDMS \citep{knoblauch2018spatio}. We follow the authors' recommendations for parameter selection, with details provided in Appendix~\ref{app:experiments-parameters}. Runtime experiments are shown in Appendix~\ref{app:experiments-synthetic-complexity}.

\subsection{Synthetic data simulations}
\label{sec:experiments-synthetic}

\paragraph{Data sets.}
We consider six synthetic data sets based on $\varprocess{\order}{1}$ processes as a canonical and flexible model for multivariate time series, enabling controlled variation of dimension, noise, and dynamics. Four are bivariate ($\dimension=2$) with a common transition parametrisation and differ only by noise: Gaussian, Laplace, Student’s $t_\nu$ ($\nu \in \{3, \dots, 30\}$), and Huber-$\epsilon$ ($\epsilon \in \{0, \dots, 0.40\}$). The remaining two are higher-dimensional ($\dimension \in \{2, \dots, 40\}$): a sparse setting embedding a bivariate system in noise, and a dense setting with full-rank transitions obtained via a similarity transform of a random real canonical form. Full details are in Appendix~\ref{app:experiments-synthetic}. 

\begin{table}[t]
	\centering
	\small
	\caption{\small Best $\fone$-score per method and synthetic data set across each method's parameter grid. Values are mean $\pm$ std from $10^3$ bootstrap replicates. For each data set, the best method is in \textbf{bold} and the second best is \underline{underlined}. CHASM achieves the best performance on 5 out of 6 data sets and ranks second on the sixth.}
	\label{tab:accuracy-synthetic}
	\setlength{\tabcolsep}{4pt}
	\renewcommand{\arraystretch}{1.1}
	\begin{tabular}{l c c c c c c c}
		\toprule
		& {CHASM} & {mSSA} & {mSSA-MW} & {Tian \& S.} & {OCD} & {CPDMD} & {BOCPDMS} \\
		\midrule
		Gaussian   & $\meansd{\mathbf{.75}}{.01}$ & $\meansd{.51}{.01}$ & $\meansd{.49}{.01}$ & $\meansd{\underline{.74}}{.01}$ & $\meansd{.40}{.01}$ & $\meansd{.40}{.01}$ & $\meansd{.58}{.01}$ \\
		Student's $t$   & $\meansd{\mathbf{.73}}{.01}$ & $\meansd{.46}{.01}$ & $\meansd{.48}{.01}$ & $\meansd{\underline{.62}}{.01}$ & $\meansd{.32}{.01}$ & $\meansd{.36}{.01}$ & $\meansd{.53}{.01}$ \\
		Laplace   & $\meansd{\mathbf{.74}}{.01}$ & $\meansd{.48}{.01}$ & $\meansd{.49}{.01}$ & $\meansd{\underline{.64}}{.01}$ & $\meansd{.33}{.01}$ & $\meansd{.36}{.01}$ & $\meansd{.49}{.01}$ \\
		Huber-$\epsilon$   & $\meansd{\mathbf{.73}}{.01}$ & $\meansd{.46}{.01}$ & $\meansd{.47}{.01}$ & $\meansd{\underline{.62}}{.01}$ & $\meansd{.31}{.01}$ & $\meansd{.36}{.01}$ & $\meansd{.55}{.01}$ \\
		High-dim. sparse   & $\meansd{\underline{.84}}{.01}$ & $\meansd{.31}{.01}$ & $\meansd{.28}{.01}$ & $\meansd{\mathbf{.87}}{.01}$ & $\meansd{.32}{.01}$ & $\meansd{.42}{.01}$ & $\meansd{.31}{.01}$ \\
		High-dim. full   & $\meansd{\mathbf{.65}}{.01}$ & $\meansd{.51}{.01}$ & $\meansd{.48}{.01}$ & $\meansd{.34}{.01}$ & $\meansd{.31}{.01}$ & $\meansd{\underline{.53}}{.01}$ & $\meansd{.31}{.01}$ \\
		\bottomrule
	\end{tabular}
\end{table}

Table~\ref{tab:accuracy-synthetic} reports $\fone$-scores on synthetic data. Across the four bivariate data sets, CHASM consistently leads while its performance varies by at most $0.02$ across noise distributions, confirming the nonparametric moment-only assumptions of Section~\ref{sec:theory}. In contrast, Tian \& S(afikhani) is the closest competitor under Gaussian noise but degrades substantially under heavier tails. On high-dimensional sparse data, Tian \& S. leads, as it is specifically constructed for low-rank sparse dynamics; CHASM matches closely this performance without any sparsity assumption. In the high-dimensional full-rank setting, Tian \& S. performs poorly while CHASM leads, even with rank fixed at $\rank=2$ throughout, demonstrating robustness to rank misspecification and to varied transition dynamics. BOCPDMS performs moderately in the bivariate setting but degrades in both high-dimensional cases, consistent with the scalability limitations of Bayesian inference. CPDMD, despite sharing DMD as a building block, underperforms throughout, suggesting that monitoring reconstruction error rather than the spectrum is insufficient. mSSA and mSSA-MW perform consistently below CHASM across all data sets, and OCD performs worst overall, as it ignores temporal dynamics entirely. Detailed $\arlzero/\arlone$ results, where CHASM is also preferred, are provided in Appendix~\ref{app:experiments-results}.

\begin{figure}[t]
	\centering
	\includegraphics[width=\linewidth]{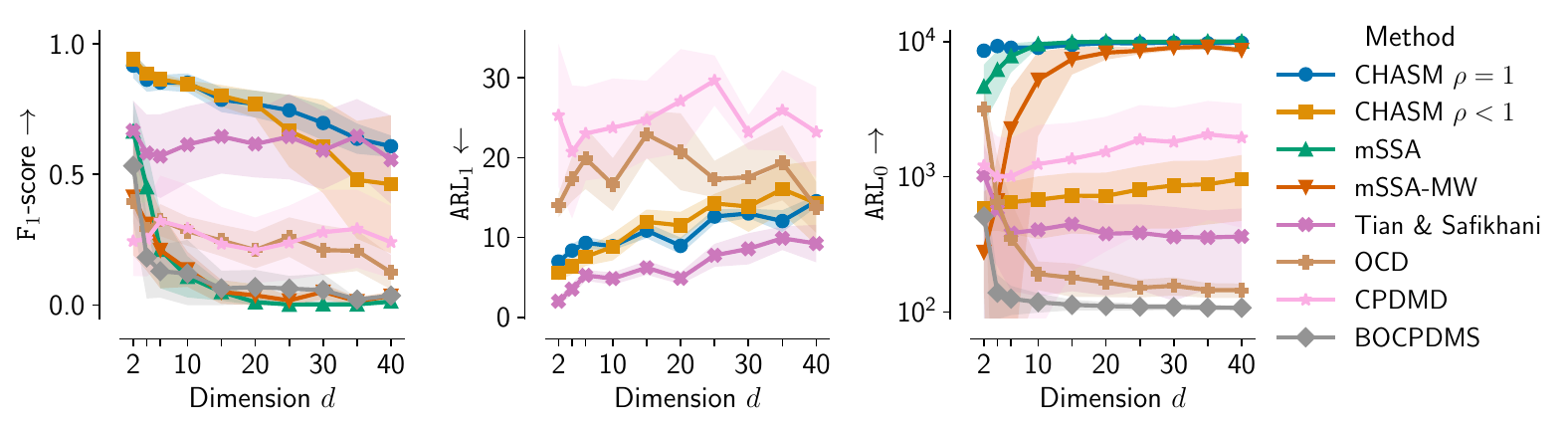}
	\caption{\small $\fone$, $\arlone$, and $\arlzero$ vs. dimension $\dimension$ for the sparse VAR data set. Higher is better for $\fone$ and $\arlzero$, and lower is better for $\arlone$. Shaded bands show mean $\pm$ std over the parameter grid. Methods with near-zero $\fone$ (mSSA, mSSA-MW, BOCPDMS) are omitted from $\arlone$ for readability; the full figure is in Appendix~\ref{app:experiments-results}.}
	\label{fig:dimension}
\end{figure}

Figure~\ref{fig:dimension} shows performance as dimension increases. CHASM with $\rho=1$ yields the best $\fone$-score and $\arlzero$ for all dimensions, with an $\arlone$ that is only slightly larger than Tian \& S. CHASM with $\rho < 1$ has similar performance, except for lower $\arlzero$, consistent with the discussion in Section~\ref{sec:theory}. However, Section~\ref{sec:experiments-real} shows that using $\rho < 1$ yields improved performance for real-world data sets.

\begin{table}[t]
	\centering
	\small
	\caption{\small Best $\fone$-score per method and real-world data set across each method's parameter grid. Values are mean $\pm$ std from each time series. For each data set, the best method is in \textbf{bold} and the second best is \underline{underlined}. CHASM achieves the best performance on 4 out of 5 data sets and ranks second on the fifth.}
	\label{tab:accuracy-real}
	\setlength{\tabcolsep}{4pt}
	\renewcommand{\arraystretch}{1.1}
	\begin{tabular}{l c c c c c c c}
		\toprule
		& {CHASM} & {mSSA} & {mSSA-MW} & {Tian \& S.} & {OCD} & {CPDMD} & {BOCPDMS} \\
		\midrule
		HASC \citep{kawaguchi2011hasc}         & $\meansd{\underline{.35}}{.10}$ & $\meansd{.30}{.10}$ & $\meansd{.33}{.04}$ & $\meansd{.30}{.07}$ & $\meansd{.28}{.10}$ & $\meansd{\mathbf{.47}}{.08}$ & $\meansd{.18}{.05}$ \\
		CIFAR-100 \citep{krizhevsky2009learning}    & $\meansd{\mathbf{.87}}{.17}$ & $\meansd{.56}{.19}$ & $\meansd{.64}{.20}$ & $\meansd{.62}{.19}$ & $\meansd{\underline{.67}}{.15}$ & $\meansd{.61}{.18}$ & $\meansd{.18}{.02}$ \\
		20 Newsgroups \citep{lang1995newsweeder} & $\meansd{\mathbf{.73}}{.20}$ & $\meansd{.50}{.16}$ & $\meansd{.56}{.18}$ & $\meansd{\underline{.67}}{.18}$ & $\meansd{.57}{.16}$ & $\meansd{.50}{.14}$ & $\meansd{.17}{.02}$ \\
		UCF-Crime \citep{sultani2018real}    & $\meansd{\mathbf{.62}}{.20}$ & $\meansd{\underline{.61}}{.19}$ & $\meansd{.50}{.23}$ & $\meansd{.54}{.17}$ & $\meansd{.34}{.19}$ & $\meansd{.55}{.20}$ & $\meansd{.32}{.18}$ \\
		WikiSection \citep{arnold2019sector}  & $\meansd{\mathbf{.50}}{.12}$ & $\meansd{.37}{.14}$ & $\meansd{\underline{.49}}{.10}$ & $\meansd{.36}{.11}$ & $\meansd{.48}{.10}$ & $\meansd{.44}{.11}$ & $\meansd{.16}{.14}$ \\
		\bottomrule
	\end{tabular}
\end{table}

\begin{figure}[t]
	\centering
	\begin{minipage}[t]{0.48\linewidth}
		\centering
		\resizebox{\linewidth}{!}{%
			\includegraphics{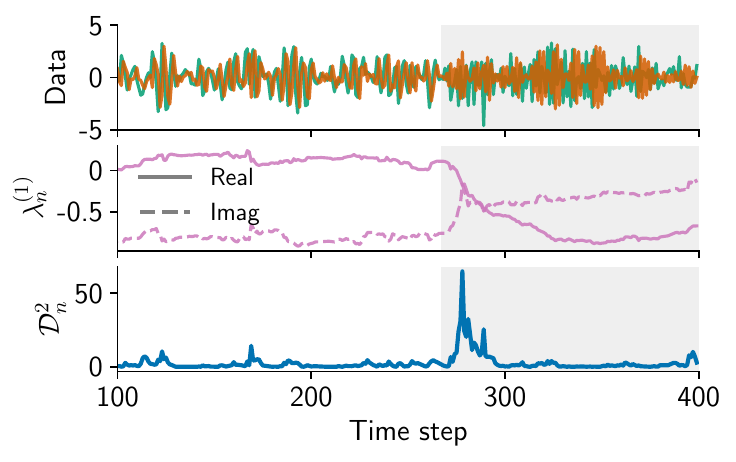}
		}
		\caption{\small Synthetic $\varprocess{2}{1}$ example with a single changepoint (shaded post-change region), along with one of the eigenvalues and CHASM's $\mathcal{D}_n^2$ statistic.
		}
		\label{fig:synthetic-statistic}
	\end{minipage}
	\hfill
	\begin{minipage}[t]{0.48\linewidth}
		\centering
		\resizebox{\linewidth}{!}{%
			\includegraphics{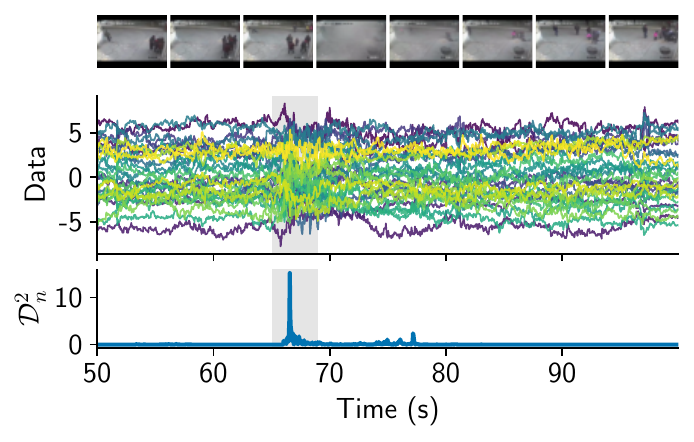}
		}
		\caption{\small Real-world UCF-Crime video example showing sample frames, 32-dim. CLIP-based signal (shaded change region), and CHASM's $\mathcal{D}_n^2$ statistic.
		}
		\label{fig:real-statistic}
	\end{minipage}
\end{figure}

\subsection{Real-world experiments}
\label{sec:experiments-real}

\paragraph{Data sets.}
We evaluate on five real-world data sets across four modalities. HASC provides accelerometry signals; CIFAR-100 and 20 Newsgroups form piecewise i.i.d. sequences of CLIP embeddings \citep{radford2021learning} for images and text; UCF-Crime represents videos as sequences of per-frame CLIP embeddings; and WikiSection embeds articles at the token level using Qwen3-Embedding \citep{zhang2025qwen3}, inducing sequential structure via the autoregressive language model. All data (except HASC) are converted to vector sequences via pretrained models in a fully online, real-time manner, without fine-tuning. As all samples contain annotated changes, $\arlzero/\arlone$ are not meaningful, and we report detection accuracy only. Full details are provided in Appendix~\ref{app:experiments-real}.

\paragraph{Scope and novelty.}
These data sets extend evaluation beyond VAR, covering i.i.d. embedding shifts (CIFAR-100, 20 Newsgroups) and nonlinear temporal dynamics (UCF-Crime, WikiSection, HASC), up to $\dimension=64$. None are naturally VAR-modelled, making this a novel and challenging setting for changepoint detection. All methods are applied fully unsupervised, without task-specific training.

Table~\ref{tab:accuracy-real} reports $\fone$-scores on real data; see Appendix~\ref{app:experiments-results} for results per data set. CHASM is the only method consistently strong across modalities. It reaches clear superior performance on CIFAR-100 and 20 Newsgroups, showing generalisation beyond VAR. On UCF-Crime and WikiSection, with nonlinear temporal dependencies, CHASM again yields superior performance, indicating robustness beyond its assumptions. WikiSection is hardest, with all methods performing relatively poorly. HASC is the only exception, where CPDMD leads, likely due to periodic signals with clear changes in marginal distribution favouring reconstruction-based statistics. No competing method achieves strong performance across all modalities, whereas CHASM remains consistent throughout. Figures~\ref{fig:synthetic-statistic} and~\ref{fig:real-statistic} show how the CHASM statistic $\mathcal{D}_n^2$ reacts to changes in a synthetic and a real-world time series.

\section{Conclusion}
\label{sec:conclusion}

We introduce CHASM, an online method for detecting changes in the transition dynamics of multivariate time series. By monitoring the eigenvalue sequence of a dynamics operator, it yields a direct and computationally efficient operator-level detection statistic. The method requires no parametric noise assumptions or structural constraints on the dynamics. Experiments confirm robustness across diverse noise and regimes, with superior performance on real-world tasks beyond VAR settings.

\paragraph{Limitations.}
While efficient, CHASM uses repeated eigendecompositions; first-order perturbation approximations may reduce cost without degrading performance, pending empirical validation. The optimal choice of $\rho$ varies by setting, so an automatic tuning step would be a natural improvement.
Beyond dynamics operators, the framework of Sections~\ref{sec:methodology-alignment} and~\ref{sec:methodology-detection} naturally extends to general matrix sequences, with applications to covariance or adjacency matrices as a promising direction.

\section*{Acknowledgements}
Victor Khamesi is funded by a Roth Scholarship from the Department of Mathematics, Imperial College London. Ed Cohen acknowledges funding from the EPSRC, grant number EP/X002195/1.

\cleardoublepage

{\small
\bibliographystyle{plainnat}
\bibliography{references}
\nocite{}}

\cleardoublepage

\appendix

\section{Extended background}
\label{app:background}

\subsection{Vector autoregressive processes}
\label{app:background-var}

\begin{definition}
	\label{def:var-process}
	A $\dimension$-dimensional vector autoregressive process of order $\order$, denoted $\varprocess{\dimension}{\order}$, $\dimension, \order \in \N$, is a discrete-time stochastic process $\{ y_t \}$, where $y_t \in \R^{\dimension}$, following the recursion
	\begin{equation}
		\label{eq:var-process}
		y_t = \sum_{k=1}^{\order} \Theta_k \, y_{t-k} + \varepsilon_t \qquad \text{for all } t \in \N,
	\end{equation}
	where $\Theta_1, \dots, \Theta_{\order} \in \R^{\dimension \times \dimension}$ are fixed. The sequence $\{ \varepsilon_t \}$ is a $\dimension$-dimensional white noise process with zero mean, i.e. for all $t, s \in \N$, $\E{\varepsilon_t} = 0$ and $\E{\varepsilon_t \transp{\varepsilon_s}} = \mathbbm{1}_{t=s} \, \Sigma$, with $\Sigma \in \R^{\dimension \times \dimension}$ symmetric positive semi-definite.
\end{definition}

\begin{definition}
	\label{def:stationarity}
	A stochastic process $\{ y_t \}$ is weakly stationary if, for all $t, s \in \N$, $\E{y_t}$ is finite and constant, and $\Cov{y_t}{y_s}$ is finite and depends only on the lag $(t-s)$.
\end{definition}

In the context of VAR processes, the general stationarity conditions of Definition~\ref{def:stationarity} translates into a simple algebraic identity on the transition matrices $\Theta_1, \dots, \Theta_{\order}$.

\begin{lemma}
	\label{thm:var-stationarity}
	A $\varprocess{\dimension}{\order}$ process with transition matrices $\Theta_1, \dots, \Theta_{\order}$ is weakly stationary if and only if
	\begin{equation}
		\label{eq:var-stationarity}
		\det\!\left(I_{\dimension} - \sum_{k=1}^{\order} \Theta_k \, z^k \right) \neq 0
		\qquad \text{for all } z \in \C \text{ such that } \abs{z} \leq 1.
	\end{equation}
\end{lemma}

\begin{proof}
	See \citep{hamilton1994time}, Section 10.1.
\end{proof}

The parameters of interest in a $\varprocess{\dimension}{\order}$ process are the transition matrices ${\Theta_1, \dots, \Theta_{\order}}$, which fully characterise the linear dynamics of the system. We note that it is sufficient to restrict our attention to $\varprocess{\dimension}{1}$ processes by means of a standard reduction.

\begin{lemma}
	\label{thm:varp-var1}
	Let $\{ y_t \}$ be a $\varprocess{\dimension}{\order}$ process as in Definition~\ref{def:var-process}. Then,
	\begin{enumerate}
		\item[(i)] $\{ y_t \}$ admits an equivalent $\varprocess{\dimension \order}{1}$ representation $\{ x_t \}$  via the companion form
		\begin{equation*}
			x_t = \Theta \, x_{t-1} + \tilde \varepsilon_t \qquad \text{for all } t \in \N,
		\end{equation*}
		where $x_t \in \R^{\dimension \order}$ stacks $\order$ consecutive observations of $y_t$, $\Theta \in \R^{\dimension \order \times \dimension \order}$ is the companion matrix constructed from the original autoregressive parameters $\Theta_1, \dots, \Theta_\order$, and $\{ \tilde \varepsilon_t \}$ is a $\dimension \order$-dimensional white noise process.
		\item[(ii)] The process $\{ x_t \}$, and hence $\{ y_t \}$, is weakly stationary if and only if $\zeta(\Theta) < 1$, where $\zeta(\Theta)$ denotes the spectral radius of the companion matrix $\Theta$.
	\end{enumerate}
\end{lemma}

\begin{proof}
	See \citep{hamilton1994time}, Section 10.1.
\end{proof}

\begin{remark}
	\label{rem:preliminaries-var-seasonality}
	A seasonal time series, where the current state is a lagged copy of its value $s$ time steps in the past, can be seen as a special case of a VAR process. In particular, $y_t = y_{t-s} + \eta_t$ corresponds to a sparse $\varprocess{\dimension}{s}$ model with $\Theta_s = I_{\dimension}$ and $\Theta_k = 0$ for $k \neq s$.
\end{remark}

\subsection{Dynamic mode decomposition}
\label{app:background-dmd}

Dynamic mode decomposition (DMD) \citep{schmid2010dynamic} is a data-driven spectral method for identifying the dominant spatiotemporal structures of a dynamical system from measurement data. Originally developed for high-dimensional fluid systems, it has since found applications in dimensionality reduction, forecasting, and control \citep{schmid2022dynamic}, and admits an interpretation as a finite-dimensional approximation of the Koopman operator \citep{brunton2022modern}. Unlike methods such as principal or independent component analysis, DMD enforces a joint representation of spatial structure and temporal evolution.

\subsubsection{Modal decomposition of linear dynamics}

We assume the system evolves according to an ordinary differential equation 
\begin{equation*}
	\frac{\mathrm{d}x}{\mathrm{d}t} = \mathcal A x,
\end{equation*}
where $x(t) \in \R^\dimension$ is the state and $\mathcal{A} \in \R^{\dimension \times \dimension}$ is a linear operator. For a given initial condition $x(0)$, the solution is given by $x(t) = e^{\mathcal{A} t} x(0)$. In particular, if $\mathcal{A} = \Phi \, \Omega \, \Phi^{-1}$ is diagonalisable, then 
\begin{equation*}
	x(t) = \sum_{i=1}^\dimension \phi_i e^{\omega_i t} b_i,
\end{equation*}
where $\Phi = \left( \phi_1, \dots, \phi_\dimension \right) \in \C^{\dimension \times \dimension}$ with eigenvectors $\phi_i \in \C^\dimension$, $\Omega = \diag{\!(\omega_1, \dots, \omega_\dimension)} \in \C^{\dimension \times \dimension}$ with eigenvalues $\omega_i \in \C$, and $b = \Phi^{-1} x(0) \in \C^\dimension$. In practice, data are observed at discrete times $t_n = n \Delta t$, yielding snapshots $x_n = x(t_n)$. The continuous-time dynamics are approximated by discrete-time linear map 
\begin{equation*}
	x_{n+1} = A \, x_n, \qquad A = e^{\mathcal{A} \Delta t},
\end{equation*}
with solution $x_n = A^{n-1} x_1$ for an initial condition $x_1 \in \R^\dimension$. If $A = \Phi \, \Lambda \, \Phi^{-1}$, then
\begin{equation*}
	x_n = \sum_{i=1}^\dimension \phi_i \lambda_i^{n-1} b_i,
\end{equation*}
where $\lambda_i = e^{\omega_i \Delta t}$ and $b = \Phi^{-1} x_1$. Each eigenpair $(\phi_i, \lambda_i)$ corresponds to a spatiotemporal mode: $\phi_i$ encodes spatial structure and $\lambda_i$ encodes temporal dynamics through its modulus (growth or decay rate) and argument (oscillation frequency).

\section{Proofs of main results}
\label{app:proofs}

\subsection{Auxiliary results}
\label{app:proofs-aux}

\begin{definition}
	\label{def:asymptotic-boundedness}
	Let $\{a_n\}$ be a sequence of positive real values.
	\begin{enumerate}
		\item[(i)] For a sequence of \emph{deterministic} matrices $\{A_n\} \subset \R^{k \times k}$, we write $A_n = \mathcal{O}(a_n)$ as $n \to \infty$ if there exist $c > 0$ and $n_0 \geq 1$ such that
		\begin{equation*}
			\norm{A_n} \leq c \, a_n, \qquad \text{for all } n > n_0.
		\end{equation*}
		\item[(ii)] For a sequence of \emph{random} matrices $\{X_n\} \subset \R^{k \times k}$, we write $X_n = \mathcal{O}_p(a_n)$ as $n \to \infty$ if, for every $\epsilon > 0$, there exist $m_\epsilon > 0$ and $n_\epsilon \geq 1$ such that
		\begin{equation*}
			\P{\norm{X_n} > m_\epsilon \, a_n} < \epsilon, \qquad \text{for all } n > n_\epsilon.
		\end{equation*}
	\end{enumerate}
\end{definition}

\begin{lemma}[Weyl's inequality]
	\label{thm:weyl-inequality}
	Let $A, B \in \C^{k \times k}$ be Hermitian and let the eigenvalues $\lambda_i(A)$, $\lambda_i(B)$, and $\lambda_i(A+B)$ be arranged in increasing order
	\begin{equation*}
		\lambda_{\min} = \lambda_1 \leq \lambda_2 \leq \dots \leq \lambda_{k-1} \leq \lambda_k = \lambda_{\max}.
	\end{equation*}
	For each $i \in \{ 1, \dots, k \}$, we have
	\begin{equation*}
		\lambda_i(A) + \lambda_1(B) \leq \lambda_i(A+B) \leq \lambda_i(A) + \lambda_n(B).
	\end{equation*}
\end{lemma}

\begin{proof}
	See \citep{horn2012matrix}, Theorem~4.3.1.
\end{proof}

\subsection{Optimal eigenvalue alignment}
\label{app:proofs-alignment}

\hungarianotpermutation*

\begin{proof}
	By definition,
	\begin{equation*}
		W_2^2(\nu_{n-1}, \nu_n) = \inf_{\gamma \in \Gamma(\nu_{n-1}, \nu_n)} \int_{\C \times \C} \abs{z_1 - z_2}^2 \, \mathrm{d}\gamma(z_1, z_2),
	\end{equation*}
	where $\Gamma(\nu_{n-1}, \nu_n)$ denotes the set of couplings of $\nu_{n-1}$ and $\nu_n$, i.e. the set of joint probability measures on $\C \times \C$ with marginals $\nu_{n-1}$ and $\nu_n$. 
	
	Since both measures are discrete and uniform over $\rank$ atoms, any coupling $\gamma$ is characterised by a matrix $\Pi \in \R^{\rank \times \rank}$ with $\Pi_{ij} = \gamma\left(\{\hat\lambda_{n-1}^{(i)}\}, \{\hat\lambda_n^{(j)}\}\right)$. The integral thus becomes a discrete sum
	\begin{equation*}
		\int_{\C \times \C} \abs{z_1 - z_2}^2 \, \mathrm{d}\gamma(z_1, z_2) = \sum_{i=1}^{\rank} \sum_{j=1}^{\rank} \abs{\hat \lambda_{n-1}^{(i)} - \hat \lambda_n^{(j)}}^2 \Pi_{ij}.
	\end{equation*}
	The marginal constraints require
	\begin{equation*}
		\sum_{j=1}^\rank \Pi_{ij} = \frac{1}{\rank} \quad \text{for all } i, \qquad \sum_{i=1}^\rank \Pi_{ij} = \frac{1}{\rank} \quad \text{for all } j,
	\end{equation*}
	so $P = \rank \,\Pi \in \mathcal B_\rank$ is a doubly stochastic matrix, where
	\begin{equation}
		\label{eq:doubly-stochastic-matrices}
		\mathcal{B}_{\rank} = \left\{ M \in \R^{\rank \times \rank} ~ : ~ M_{ij} \geq 0, ~ \sum_{j=1}^{\rank} M_{ij} = \sum_{i=1}^{\rank} M_{ij} = 1 \right\}
	\end{equation}
	The transport cost becomes
	\begin{equation*}
		W_2^2(\nu_{n-1}, \nu_n) = \min_{P \in \mathcal{B}_\rank} \frac{1}{\rank} \sum_{i,j=1}^\rank \abs{\hat\lambda_{n-1}^{(i)} - \hat\lambda_n^{(j)}}^2 P_{ij}.
	\end{equation*}
	This is a linear program in $P$, so its minimum is attained at a vertex of $\mathcal{B}_\rank$. By Birkhoff-von Neumann, the vertices of $\mathcal{B}_\rank$ is precisely the set of permutation matrices $\mathcal S_\rank$, where
	\begin{equation}
		\label{eq:permutation-matrices}
		\mathcal{S}_{\rank} = \left\{ M \in \{0, 1\}^{\rank \times \rank} ~ : ~ \sum_{j=1}^{\rank} M_{ij} = \sum_{i=1}^{\rank} M_{ij} = 1 \right\}
	\end{equation}
	Identifying each matrix of $\mathcal{S}_\rank$ with $\pi \in S_r$ via $P_{ij} = \mathbbm{1}_{j = \pi(i)}$ yields
	\begin{equation*}
		W_2^2(\nu_{n-1}, \nu_n) = \min_{\pi \in S_\rank} \frac{1}{\rank} \sum_{i=1}^\rank \abs{\lambda_{n-1}^{(i)} - \lambda_n^{\pi(i)}}^2.
	\end{equation*}
\end{proof}

\hungarianassignment*

\begin{proof}
	For any permutation matrix $\Pi \in \mathcal{S}_{\rank}$,
	\begin{align*}
		\norm{\transp{\Pi} \Lambda_n \Pi - \Lambda_{n-1}}_F^2 &= \Tr{\!\left(\herm{\left(\transp{\Pi} \Lambda_n \Pi - \Lambda_{n-1}\right)} \left(\transp{\Pi} \Lambda_n \Pi - \Lambda_{n-1}\right)\right)} \\
		&= \Tr{\!\left(\transp{\Pi} \conj{\Lambda_n} \Lambda_n \Pi\right)} - \Tr{\!\left(\transp{\Pi} \conj{\Lambda_n} \Pi \Lambda_{n-1} + \conj{\Lambda_{n-1}} \transp{\Pi} \Lambda_n \Pi\right)} + \Tr{\!\left(\conj{\Lambda_{n-1}} \Lambda_{n-1}\right)}.
	\end{align*}
	For the first term, we note that
	\begin{equation*}
		\Tr{\!\left( \transp{\Pi} \conj{ \Lambda_n}  \Lambda_n \Pi \right)} = \Tr{\!\left(  \Lambda_n \Pi \transp{\Pi} \conj{\Lambda_n} \right)} = \Tr{\!\left(\Lambda_n \conj{\Lambda_n} \right)} = \norm{ \Lambda_n}_F^2,
	\end{equation*}
	which is independent of $\Pi$. For the cross term, we note that $\conj{\Lambda_{n-1}} \transp{\Pi} \Lambda_n \Pi = \conj{\left(\transp{\Pi} \conj{\Lambda_n} \Pi \Lambda_{n-1}\right)}$, so using $\Tr{\!\left(M + \conj{M}\right)} = 2\Tr{\!\left(\Rpart{\!\left(M\right)}\right)}$,
	\begin{equation*}
		\Tr{\!\left(\transp{\Pi} \conj{\Lambda_n} \Pi \Lambda_{n-1} + \conj{\Lambda_{n-1}} \transp{\Pi} \Lambda_n \Pi\right)} = 2\Tr{\!\left(\Rpart{\!\left(\transp{\Pi} \conj{\Lambda_n} \Pi \Lambda_{n-1}\right)}\right)}.
	\end{equation*}
	Since $\Lambda_n$ and $\Lambda_{n-1}$ are diagonal, $\transp{\Pi}\conj{\Lambda_n}\Pi$ is diagonal with $i$-th entry $\conj{\big(\lambda_n^{(\pi(i))}\big)}$, and the product $\transp{\Pi}\conj{\Lambda_n}\Pi\Lambda_{n-1}$ is diagonal with $i$-th entry $\conj{\big(\lambda_n^{(\pi(i))}\big)}\lambda_{n-1}^{(i)}$. Therefore,
	\begin{equation*}
		\Tr{\!\left(\Rpart{\!\left(\transp{\Pi}\conj{\Lambda_n}\Pi\, \Lambda_{n-1}\right)}\right)} = \sum_{i=1}^\rank \Rpart{\!\left(\conj{\big(\lambda_n^{(\pi(i))}\big)} \lambda_{n-1}^{(i)}\right)} = \Tr{\!\left(\Pi \Rpart{\!\left(\conj{\lambda_n} \transp{\lambda_{n-1}}\right)}\right)},
	\end{equation*}
	where $\conj{\lambda_n}\transp{\lambda_{n-1}} \in \C^{\rank\times\rank}$ is the outer product of the eigenvalue vectors, with $(i,j)$-th entry $\conj{\big(\lambda_n^{(i)}\big)}\lambda_{n-1}^{(j)}$. Combining all three terms,
	\begin{equation*}
		\norm{\transp{\Pi}\Lambda_n\Pi - \Lambda_{n-1}}_F^2 = \norm{\Lambda_n}_F^2 + \norm{\Lambda_{n-1}}_F^2 - 2\Tr{\!\left(\Pi\Rpart{\!\left(\conj{\lambda_n} \transp{\lambda_{n-1}}\right)}\right)}.
	\end{equation*}
	Since $\norm{\Lambda_n}_F^2 + \norm{\Lambda_{n-1}}_F^2$ is constant with respect to $\Pi$,
	\begin{equation*}
		\Pi_n = \argmin_{\Pi \in \mathcal{S}_{\rank}} \norm{\transp{\Pi}\Lambda_n\Pi - \Lambda_{n-1}}_F^2 = \argmax_{\Pi \in \mathcal{S}_{\rank}} \Tr{\!\left(\Pi\Rpart{\!\left(\conj{\lambda_n} \transp{\lambda_{n-1}}\right)}\right)}.
	\end{equation*}
	The new maximisation problem is addressed by relaxing the feasible set $\mathcal{S}_{\rank}$ into its convex hull. Since the objective is linear in $\Pi$, the optimal solution of this relaxed problem is guaranteed to lie on the boundary of the feasible region \citep{bertsimas1997introduction}. Therefore,
	\begin{equation*}
		\Pi_n = \argmax_{\Pi \in \mathcal{B}_{\rank}} \left\{\Tr{\!\left(\Pi\Rpart{\!\left(\conj{\lambda_n} \transp{\lambda_{n-1}}\right)}\right)}\right\},
	\end{equation*}
	where the maximum is now taken over the larger set $\mathcal{B}_{\rank} \supset \mathcal{S}_{\rank}$.
\end{proof}

\subsection{Complex multivariate changepoint detection}
\label{app:proofs-detection}

\ewmadecomposition*

\begin{proof}
	Note that $\underline{z}_n - \mu_{\underline{z}_n} = (z_n - \mu_{v_n},\, \conj{(z_n - \mu_{v_n})}) \in \C^{2\rank}$. The augmented covariance writes
	\begin{equation*}
		\Sigma_{\underline{z}_n} = \beta_n \begin{pmatrix} 
			\Sigma_{v_n} & \tilde\Sigma_{v_n} \\ 
			\conj{\tilde\Sigma_{v_n}} & \conj{\Sigma_{v_n}} 
		\end{pmatrix}.
	\end{equation*}
	Applying the block matrix inversion formula with Schur complement $\mathcal{C}_n = \Sigma_{v_n} - \tilde\Sigma_{v_n} (\conj{\Sigma_{v_n}})^{-1}\conj{\tilde\Sigma_{v_n}}$ gives
	\begin{equation*}
		\Sigma_{\underline{z}_n}^{-1} = \frac{1}{\beta_n}\begin{pmatrix}
			\mathcal{C}_n^{-1} & \mathcal{Q}_n \\
			\conj{\mathcal{Q}_n} & (\conj{\mathcal{C}_n})^{-1}
		\end{pmatrix},
	\end{equation*}
	where $\mathcal{Q}_n = -\mathcal{C}_n^{-1}\tilde\Sigma_{v_n} (\conj{\Sigma_{v_n}})^{-1}$. Substituting into \eqn~\eqref{eq:change-statistic} and expanding the quadratic form,
	\begin{align*}
		\mathcal{D}_n^2 &= \frac{1}{\beta_n}
		\begin{pmatrix} \herm{(z_n-\mu_{v_n})} & \transp{(z_n-\mu_{v_n})}
		\end{pmatrix}
		\begin{pmatrix} \mathcal{C}_n^{-1} & \mathcal{Q}_n \\
			\conj{\mathcal{Q}_n} & (\conj{\mathcal{C}_n})^{-1} \end{pmatrix}
		\begin{pmatrix} z_n-\mu_{v_n} \\ \conj{(z_n-\mu_{v_n})}
		\end{pmatrix} \\
		&= \frac{1}{\beta_n}\Big(
		\herm{(z_n-\mu_{v_n})}\mathcal{C}_n^{-1}(z_n-\mu_{v_n})
		+ \herm{(z_n-\mu_{v_n})}\mathcal{Q}_n\conj{(z_n-\mu_{v_n})} \\
		&\qquad\quad + \transp{(z_n-\mu_{v_n})}\conj{\mathcal{Q}_n}(z_n-\mu_{v_n})
		+ \transp{(z_n-\mu_{v_n})}(\conj{\mathcal{C}_n})^{-1}
		\conj{(z_n-\mu_{v_n})}\Big).
	\end{align*}
	The first and fourth terms are complex conjugates, as are the second and third. Since $w + \conj{w} = 2\Rpart{w}$ for any $w \in \C$,
	\begin{equation*}
		\mathcal{D}_n^2 = \frac{2}{\beta_n}\Rpart{\!\left(
			\herm{(z_n-\mu_{v_n})}\mathcal{C}_n^{-1}(z_n-\mu_{v_n}) +
			\herm{(z_n-\mu_{v_n})}\mathcal{Q}_n\conj{(z_n-\mu_{v_n})}
			\right)}. \qedhere
	\end{equation*}
\end{proof}

\subsection{Asymptotic analysis}
\label{app:proofs-bias}

\begin{lemma}
	\label{thm:gamma-pd}
	Under the model assumptions of Section~\ref{sec:theory}, the stationary covariance matrix $\Gamma = \E{x_t \transp{x_t}}$ is positive definite, i.e. $\Gamma \succ 0$.
\end{lemma}

\begin{proof}
	Recall the $\text{MA}(\infty)$ representation of the stationary process $\{ x_t \}$,
	\begin{equation}
		\label{eq:ma-infty}
		x_t = \sum_{j \geq 0} \Theta^j \varepsilon_{t-j}.
	\end{equation}
	Define the $\ell_p$ norm of a random vector $y$ as $\norm{y}_p = \E{\norm{y}^p}^{1/p}$. Since $\zeta(\Theta) < 1$, the series $\sum_{j \geq 0} \Theta^j \varepsilon_{t-j}$ converges absolutely in $\ell_2$, and
	\begin{equation*}
		\Gamma = \E{x_t \transp{x_t}} = \sum_{j \geq 0} \Theta^j \Sigma \transp{\left( \Theta^j \right)} = \Sigma + \sum_{j \geq 1} \Theta^j \Sigma \transp{\left( \Theta^j \right)}.
	\end{equation*}
	Therefore, $\Gamma \succeq \Sigma \succ 0$ since all other terms are positive semi-definite.
\end{proof}

\begin{lemma}
	\label{thm:finite-p-moments}
	Let $p \geq 1$. If the noise $\{ \varepsilon_t \}$ has finite $p$-th moments, i.e., $\E{\norm{\varepsilon_t}^p} < \infty$, then the state process $\{ x_t \}$ also has finite $p$-th moments, i.e., $\E{\norm{x_t}^p} < \infty$.
\end{lemma}

\begin{proof}
	By Minkowski's inequality for the $\ell_p$ norm and the $\text{MA}(\infty)$ representation of the stationary process $\{ x_t \}$ in \eqn~\eqref{eq:ma-infty},
	\begin{align*}
		\norm{x_t}_p &\leq \sum_{j \geq 0} \norm{\Theta^j \varepsilon_{t-j}}_p \\
		&\leq \sum_{j \geq 0} \norm{\Theta^j} \norm{\varepsilon_{t-j}}_p \\
		&= \norm{\varepsilon_0}_p \sum_{j \geq 0} \norm{\Theta^j},
	\end{align*}
	where we used that $\norm{\varepsilon_{t-j}}_p = \norm{\varepsilon_0}_p$ by stationarity of $\{ \varepsilon_t \}$. Since $\zeta(\Theta) < 1$, there exist constants $\alpha > 0$ and $\beta \in \left(\zeta(\Theta), 1\right)$ such that $\norm{\Theta^j} \leq \alpha \beta^j$ for all $j \geq 0$. Therefore,
	\begin{equation*}
		\norm{x_t}_p \leq \norm{\varepsilon_0}_p \frac{\alpha}{1-\beta} < \infty,
	\end{equation*}
	which implies $\E{\norm{x_t}^p} < \infty$.
\end{proof}

\subsubsection{Unweighted estimator}
\label{app:proofs-bias-unweighted}

Throughout this section, we consider the updates of \eqn~\eqref{eq:online-dmd-updates} with fixed $\rho = 1$.

\begin{lemma}
	\label{thm:dmd-estimator-conditional-bias}
	For any $n$, the conditional bias of $\dmdop_n$ given $\mathcal{F}_{n-1}$ is non-zero almost surely. Specifically,
	\begin{equation*}
		\E{\dmdop_n \mid \mathcal{F}_{n-1}} - \Theta = \Xi_n \Gamma_n^{-1} \neq 0 \qquad \text{almost surely},
	\end{equation*}
	where
	\begin{equation*}
		\label{eq:dmd-estimator-conditional-factorisation}
		\Xi_n = \frac{1}{n} \sum_{t=1}^{n-1} \varepsilon_t \transp{x_{t-1}}, \qquad \Gamma_n = \frac{1}{n} \sum_{t=1}^n x_{t-1} \transp{x_{t-1}}.
	\end{equation*}
\end{lemma}

\begin{proof}
	For any time step $n$, the updates of \eqn~\eqref{eq:online-dmd-updates} compute the DMD matrix
	\begin{equation*}
		\dmdop_n = \left( \sum_{t=1}^n x_t \transp{x_{t-1}} \right) \left( \sum_{t=1}^n x_{t-1} \transp{x_{t-1}} \right)^{-1}
	\end{equation*}
	exactly. Substituting the model dynamics $x_t = \Theta \, x_{t-1} + \varepsilon_t$ from the $\varprocess{\dimension}{1}$ model,
	\begin{equation*}
		\sum_{t=1}^n x_t \transp{x_{t-1}} = \Theta \sum_{t=1}^n x_{t-1} \transp{x_{t-1}} + \sum_{t=1}^n \varepsilon_t \transp{x_{t-1}}.
	\end{equation*}
	Right-multiplying by $\left( \sum_{t=1}^n x_{t-1} \transp{x_{t-1}} \right)^{-1}$, we obtain
	\begin{equation}
		\label{eq:dmd-estimator-decomposition}
		\dmdop_n = \Theta + \left( \sum_{t=1}^n \varepsilon_t \transp{x_{t-1}} \right) \left( \sum_{t=1}^n x_{t-1} \transp{x_{t-1}} \right)^{-1}.
	\end{equation}
	Conditioned on $\mathcal{F}_{n-1}$, the denominator $\sum_{t=1}^n x_{t-1} \transp{x_{t-1}}$ is deterministic since $x_0, \dots, x_{n-1}$ are all $\mathcal{F}_{n-1}$-measurable. The conditional expectation of the numerator term can be written as
	\begin{equation*}
		\E{\sum_{t=1}^n \varepsilon_t \transp{x_{t-1}} \mid \mathcal{F}_{n-1}} = \sum_{t=1}^n \E{\varepsilon_t \mid \mathcal{F}_{n-1}} \transp{x_{t-1}},
	\end{equation*}
	since $x_{t-1}$ is $\mathcal{F}_{n-1}$-measurable for $t \leq n$. We evaluate the term $\E{\varepsilon_t \mid \mathcal{F}_{n-1}}$ for different $t$.
	\begin{itemize}
		\item For $t < n$, the filtration $\mathcal{F}_{n-1}$ contains $x_t$ and $x_{t-1}$. Since $\varepsilon_t = x_t - \Theta \, x_{t-1}$, the noise realisation $\varepsilon_t$ is fully determined by the information in $\mathcal{F}_{n-1}$, hence
		\begin{equation*}
			\E{\varepsilon_t \mid \mathcal{F}_{n-1}} = \varepsilon_t \neq 0 \qquad \text{almost surely}.
		\end{equation*}
		\item For $t = n$, the noise $\varepsilon_n$ occurs at time $n$ which is not contained in $\mathcal{F}_{n-1}$. By the martingale difference property, 
		\begin{equation*}
			\E{\varepsilon_n \mid \mathcal{F}_{n-1}} = 0.
		\end{equation*}
	\end{itemize}
	Substituting these back,
	\begin{equation*}
		\E{\sum_{t=1}^n \varepsilon_t \transp{x_{t-1}} \mid \mathcal{F}_{n-1}} = \sum_{t=1}^{n-1} \varepsilon_t \transp{x_{t-1}} \neq 0 \qquad \text{almost surely}.
	\end{equation*}
	Therefore, the conditional bias
	\begin{equation*}
		\E{\dmdop_n \mid \mathcal{F}_{n-1}} - \Theta = \left( \sum_{t=1}^{n-1} \varepsilon_t \transp{x_{t-1}} \right) \left( \sum_{t=1}^n x_{t-1} \transp{x_{t-1}} \right)^{-1}
	\end{equation*}
	depends on the realised history of the noise and is non-zero almost surely.
\end{proof}

\begin{proposition}
	\label{thm:gamma-consistency}
	If $\E{\norm{\varepsilon_t}^4} < \infty$, then $\Gamma_n$ is a consistent estimator of $\Gamma$. In particular, $\Gamma_n \xrightarrow{\mathbb{P}} \Gamma$, and
	\begin{equation*}
		\E{\norm{\Gamma_n - \Gamma}^2} = \mathcal{O}(n^{-1}) \qquad \text{as } n \to \infty.
	\end{equation*}
\end{proposition}

\begin{proof}
	Let $\Delta_n = \Gamma_n - \Gamma$, and consider the $(i, j)$-th entry of $\Delta_n$,
	\begin{equation*}
		(\Delta_n)_{ij} = \frac{1}{n} \sum_{t=0}^{n-1} \left( x_{t, i} x_{t, j} - \Gamma_{ij} \right).
	\end{equation*}
	Define $(Y_t)_{ij} = x_{t, i} x_{t, j} - \Gamma_{ij}$, so that $\E{(Y_t)_{ij}} = 0$ and $(\Delta_n)_{ij} = \frac{1}{n} \sum_{t=0}^{n-1} (Y_t)_{ij}$. Then,
	\begin{equation*}
		\Var{(\Delta_n)_{ij}} = \frac{1}{n^2} \sum_{s,t=0}^{n-1} \Cov{(Y_s)_{ij}}{(Y_t)_{ij}}.
	\end{equation*}
	By stationarity, $\Cov{(Y_s)_{ij}}{(Y_t)_{ij}} = \gamma_{ij}(t-s)$ depends only on the lag. Hence,
	\begin{equation}
		\label{eq:variance-delta-entry}
		\Var{(\Delta_n)_{ij}} = \frac{1}{n} \sum_{h=-(n-1)}^{n-1} \left( 1 - \frac{\abs{h}}{n} \right) \gamma_{ij}(h),
	\end{equation}
	where $\gamma_{ij}(h) = \Cov{x_{0, i} x_{0, j}}{x_{h, i} x_{h, j}}$. Using \eqref{eq:ma-infty}, decompose $x_h$ for $h \geq 1$ as $x_h = u_h + v_h$, where
	\begin{equation*}
		u_h = \sum_{k=0}^{h-1} \Theta^k \varepsilon_{h-k}, \qquad v_h = \sum_{k \geq h} \Theta^k \varepsilon_{h-k}.
	\end{equation*}
	Here, $u_h$ depends only on $\{ \varepsilon_1, \dots, \varepsilon_h \}$ and is independent of $\mathcal{F}_0$, while $v_h$ depends only on $\{ \varepsilon_s \}_{s \leq 0}$. Since $x_{h,i} x_{h,j} = u_{h,i} u_{h,j} + v_{h,i} v_{h,j} + u_{h,i} v_{h,j} + v_{h,i} u_{h,j}$ and $u_h$ is independent of $(\mathcal{F}_0, v_h)$, each cross-covariance $\Cov{x_{0,i} x_{0,j}}{\cdot}$ reduces to one involving only $v_h$ and $x_0$. By Cauchy-Schwarz and Lemma~\ref{thm:finite-p-moments},
	\begin{equation*}
		\abs{\gamma_{ij}(h)} \leq c_1 \norm{v_h}_{L^4}^2 + c_2 \norm{v_h}_{L^4}
	\end{equation*}
	for constants $c_1, c_2 > 0$ depending on $\E{\norm{x_0}^4}$ and $\E{\norm{u_h}^4}$. By Minkowski's inequality and $\norm{\Theta^k} \leq \alpha \beta^k$ for some $\alpha > 0$ and $\beta \in (\zeta(\Theta), 1)$,
	\begin{equation*}
		\norm{v_h}_{L^4} \leq \norm{\varepsilon_0}_{L^4} \sum_{k \geq h} \alpha \beta^k = \frac{\alpha \norm{\varepsilon_0}_{L^4}}{1 - \beta} \, \beta^h.
	\end{equation*}
	Therefore, $\abs{\gamma_{ij}(h)} \leq \omega \beta^{\abs{h}}$ for some $\omega > 0$ using $\gamma_{ij}(-h) = \gamma_{ij}(h)$ by stationarity, and
	\begin{equation}
		\label{eq:summable-autocovariance}
		\sum_{h \in \Z} \abs{\gamma_{ij}(h)} \leq \frac{2\omega}{1 - \beta} < \infty.
	\end{equation}
	Since $\left(1 - \frac{\abs{h}}{n}\right) \leq 1$ in \eqn~\eqref{eq:variance-delta-entry}, the absolute summability \eqn~\eqref{eq:summable-autocovariance} gives
	\begin{equation*}
		\Var{(\Delta_n)_{ij}} \leq \frac{1}{n} \sum_{h \in \Z} \abs{\gamma_{ij}(h)} = \mathcal{O}(n^{-1}) \qquad \text{as } n \to \infty.
	\end{equation*}
	Furthermore, as $\norm{\Delta_n}^2 \leq \norm{\Delta_n}_F^2 = \sum_{i,j=1}^\dimension \left( \Delta_n \right)^2_{ij}$,
	\begin{equation}
		\label{eq:bound-expectation-delta}
		\E{\norm{\Delta_n}^2} \leq \sum_{i,j=1}^\dimension \Var{(\Delta_n)_{ij}} = \mathcal{O}(n^{-1}) \qquad \text{as } n \to \infty.
	\end{equation}
	Finally, by Markov's inequality applied to the non-negative random variable $\norm{\Delta_n}^2$, for any $\epsilon > 0$,
	\begin{equation*}
		\P{\norm{\Gamma_n - \Gamma} > \epsilon} = \P{\norm{\Delta_n}^2 > \epsilon^2} \leq \frac{\E{\norm{\Delta_n}^2}}{\epsilon^2},
	\end{equation*}
	which shows $\P{\norm{\Gamma_n - \Gamma} > \epsilon} \to 0$ and therefore $\Gamma_n \xrightarrow{\mathbb{P}} \Gamma$ as $n \to \infty$.
\end{proof}

\begin{lemma}
	\label{thm:gamma-n-pd}
	If $\E{\norm{\varepsilon_t}^4} < \infty$, then $\P{\Gamma_n \succ 0} \to 1$ as $n \to \infty$. In particular, $\Gamma_n$ is eventually positive definite.
\end{lemma}

\begin{proof}
	By Weyl's inequality (Lemma~\ref{thm:weyl-inequality}) with $A = \Gamma$ and $B = \Delta_n = \Gamma_n - \Gamma$,
	\begin{equation*}
		\lambda_{\min}(\Gamma_n) \geq \lambda_{\min}(\Gamma) + \lambda_{\min}(\Delta_n) \geq \lambda_{\min}(\Gamma) - \norm{\Delta_n},
	\end{equation*}
	since $\lambda_{\min}(B) \geq -\norm{B}$ for any Hermitian matrix $B$. By Proposition~\ref{thm:gamma-consistency}, for any $\delta < \lambda_{\min}(\Gamma)$,
	\begin{equation*}
		\mathbb{P}\left( \norm{\Gamma_n - \Gamma} < \delta \right) \to 1 \qquad \text{as } n \to \infty.
	\end{equation*}
	Taking $\delta = \lambda_{\min}(\Gamma)/2$, and since $\lambda_{\min}(\Gamma) > 0$ by Lemma~\ref{thm:gamma-pd},
	\begin{equation*}
		\P{\lambda_{\min}(\Gamma_n) > \frac{\lambda_{\min}(\Gamma)}{2}} \to 1 \qquad \text{as } n \to \infty.
	\end{equation*}
	The result is finally obtained by the inclusion $\left\{ \lambda_{\min}(\Gamma_n) > \lambda_{\min}(\Gamma)/2 \right\} \subseteq \left\{ \Gamma_n \succ 0 \right\}$.
\end{proof}

\begin{lemma}
	\label{thm:neumann-expansion}
	On the event $\left\{ \Gamma_n \succ 0 \right\}$,
	\begin{equation*}
		\Gamma_n^{-1} = \Gamma^{-1} - \Gamma^{-1} \Delta_n \Gamma^{-1} + \Upsilon_n,
	\end{equation*}
	where $\Upsilon_n = \Gamma^{-1} \Delta_n \Gamma_n^{-1} \Delta_n \Gamma^{-1}$ with $\Delta_n = \Gamma_n - \Gamma$.
\end{lemma}

\begin{proof}
	The resolvent identity gives
	\begin{equation}
		\label{eq:resolvent-identity-gamma-n}
		\Gamma_n^{-1} - \Gamma^{-1} = -\Gamma_n^{-1} \left( \Gamma_n - \Gamma \right) \Gamma^{-1},
	\end{equation}
	or equivalently, $\Gamma_n^{-1} = \Gamma^{-1} - \Gamma_n^{-1} \Delta_n \Gamma^{-1}$. Substituting this expression into itself,
	\begin{align}
		\label{eq:expansion-r-non-symmetric}
		\Gamma_n^{-1} &= \Gamma^{-1} - \left( \Gamma^{-1} - \Gamma_n^{-1} \Delta_n \Gamma^{-1} \right) \Delta_n \Gamma^{-1} \nonumber \\
		&= \Gamma^{-1} - \Gamma^{-1} \Delta_n \Gamma^{-1} + \Gamma_n^{-1} \Delta_n \Gamma^{-1} \Delta_n \Gamma^{-1}.
	\end{align}
	Then, by \eqn~\eqref{eq:resolvent-identity-gamma-n} and \eqn~\eqref{eq:expansion-r-non-symmetric}
	\begin{align*}
		\Gamma_n^{-1} \Delta_n \Gamma^{-1} \Delta_n \Gamma^{-1} &= \Gamma_n^{-1} - \Gamma^{-1} + \Gamma^{-1} \Delta_n \Gamma^{-1} \\
		&= -\Gamma_n^{-1} \Delta_n \Gamma^{-1} + \Gamma^{-1} \Delta_n \Gamma^{-1} \\
		&= \left( \Gamma^{-1} - \Gamma_n^{-1} \right) \Delta_n \Gamma^{-1} \\
		&= \Gamma^{-1} \Delta_n \Gamma_n^{-1} \Delta_n \Gamma^{-1} \\
		&= \Upsilon_n.
	\end{align*}
\end{proof}

\begin{lemma}
	\label{thm:fourth-moment-delta}
	If $\E{\norm{\varepsilon_t}^8} < \infty$, then $\E{\norm{\Delta_n}^4} = \mathcal{O}(n^{-2})$ as $n \to \infty$.
\end{lemma}

\begin{proof}
	Since $\norm{\Delta_n}^4 \leq \norm{\Delta_n}_F^4$ and
	\begin{equation*}
		\norm{\Delta_n}_F^4 = \left( \sum_{i,j=1}^\dimension (\Delta_n)_{ij}^2 \right)^2 = \sum_{i,j=1}^\dimension (\Delta_n)_{ij}^4 + 2 \sum_{(i,j) < (k,\ell)} (\Delta_n)_{ij}^2 \, (\Delta_n)_{k\ell}^2,
	\end{equation*}
	it suffices to show $\E{(\Delta_n)_{ij}^4} = \mathcal{O}(n^{-2})$ for each $(i,j)$, noting that the cross terms satisfy
	\begin{equation*}
		\E{(\Delta_n)_{ij}^2 \, (\Delta_n)_{k\ell}^2} \leq \sqrt{\E{(\Delta_n)_{ij}^4}} \sqrt{\E{(\Delta_n)_{k\ell}^4}}
	\end{equation*}
	by Cauchy-Schwarz. Let
	\begin{equation*}
		S_n = \sum_{t=0}^{n-1} (Y_t)_{ij}, \qquad (\Delta_n)_{ij} = \frac{S_n}{n}
	\end{equation*}
	By the moment-cumulant relation for zero-mean random variables,
	\begin{equation}
		\label{eq:moment-cumulant}
		\E{S_n^4} = 3 \, \E{S_n^2}^2 + \sum_{t_1, t_2, t_3, t_4 = 0}^{n-1} \kappa(Y_{t_1}, Y_{t_2}, Y_{t_3}, Y_{t_4}),
	\end{equation}
	where
	\begin{align*}
		\kappa(Y_{t_1}, Y_{t_2} Y_{t_3}, Y_{t_4}) &= \E{Y_{t_1} Y_{t_2} Y_{t_3} Y_{t_4}} - \E{Y_{t_1} Y_{t_2}}\E{Y_{t_3} Y_{t_4}} \\
		&\quad - \E{Y_{t_1} Y_{t_3}}\E{Y_{t_2} Y_{t_4}} - \E{Y_{t_1} Y_{t_4}}\E{Y_{t_2} Y_{t_3}}.
	\end{align*}
	The identity \eqn~\eqref{eq:moment-cumulant} is verified by direct substitution. From Proposition~\ref{thm:gamma-consistency}, $\E{S_n^2} = \mathcal{O}(n)$, thus $3 \, \E{S_n^2}^2 = \mathcal{O}(n^2)$. For the cumulant term, $\kappa(Y_{t_1}, Y_{t_2}, Y_{t_3}, Y_{t_4}) = c(h_1, h_2, h_3)$ by stationarity and depends only on the lag differences $h_k = t_{k+1} - t_1$, for $k \in \{1, 2, 3\}$. Hence
	\begin{equation*}
		\abs{\sum_{t_1, t_2, t_3, t_4 = 0}^{n-1} \kappa(Y_{t_1}, Y_{t_2}, Y_{t_3}, Y_{t_4})} \leq n \sum_{h_1, h_2, h_3 \in \Z} \abs{c(h_1, h_2, h_3)}.
	\end{equation*}
	The joint cumulant $c(h_1, h_2, h_3)$ involves expectations of the form $\E{Y_0 Y_{h_1} Y_{h_2} Y_{h_3}}$, each bounded by $\E{Y_0^4}$ via Cauchy-Schwarz. Then, $\E{Y_0^4} < \infty$ follows from $\E{\norm{x_0}^8} < \infty$, which holds by Lemma~\ref{thm:finite-p-moments} with $p = 8$ under the assumption $\E{\norm{\varepsilon_t}^8} < \infty$. By the same $u_h, v_h$ decomposition used in Lemma~\ref{thm:gamma-consistency}, the cumulant $c(h_1, h_2, h_3)$ decays geometrically in $\max(\abs{h_1}, \abs{h_2}, \abs{h_3})$, and thus
	\begin{equation*}
		\sum_{h_1, h_2, h_3 \in \Z} \abs{c(h_1, h_2, h_3)} < \infty.
	\end{equation*}
	Therefore 
	\begin{equation*}
		\sum_{t_1, t_2, t_3, t_4 = 0}^{n-1} \kappa(Y_{t_1}, Y_{t_2}, Y_{t_3}, Y_{t_4}) = \mathcal{O}(n) \qquad \text{as } n \to \infty,
	\end{equation*}
	and combining with \eqn~\eqref{eq:moment-cumulant},
	\begin{equation*}
		\E{(\Delta_n)_{ij}^4} = \frac{1}{n^4} \E{S_n^4} = \frac{1}{n^4} \left( \mathcal{O}(n^2) + \mathcal{O}(n) \right) = \mathcal{O}(n^{-2}). 
	\end{equation*}
\end{proof}

\unweightedbias*

\begin{proof}
	By Lemma~\ref{thm:dmd-estimator-conditional-bias} and the law of total expectation,
	\begin{equation}
		\label{eq:bias-decomposition}
		\E{\dmdop_n - \Theta} = \E{\E{\dmdop_n - \Theta \mid \mathcal{F}_{n-1}}} = \E{\Xi_n \Gamma_n^{-1}},
	\end{equation}
	where
	\begin{equation*}
		\Xi_n = \frac{1}{n} \sum_{t=1}^{n-1} \varepsilon_t \transp{x_{t-1}}, \qquad \Gamma_n = \frac{1}{n} \sum_{t=1}^n x_{t-1} \transp{x_{t-1}}.
	\end{equation*}
	Combining \eqn~\eqref{eq:bias-decomposition} with Lemma~\ref{thm:neumann-expansion},
	\begin{equation}
		\label{eq:bias-expansion}
		\E{\dmdop_n - \Theta} = \E{\Xi_n} \Gamma^{-1} - \E{\Xi_n \Gamma^{-1} \Delta_n} \Gamma^{-1} + \E{\Xi_n \Upsilon_n}.
	\end{equation}
	Since for each $t \in \{ 1, \dots, n-1 \}$, $x_{t-1}$ is $\mathcal{F}_{t-1}$-measurable and $\varepsilon_t$ is independent of $\mathcal{F}_{t-1}$,
	\begin{equation}
		\label{eq:unweighted-term1}
		\E{\Xi_n} = \frac{1}{n} \sum_{t=1}^{n-1} \E{\E{\varepsilon_t \transp{x_{t-1}} \mid \mathcal{F}_{t-1}}} = \frac{1}{n} \sum_{t=1}^{n-1} \E{\varepsilon_t \mid \mathcal{F}_{t-1}} \transp{x_{t-1}} = 0,
	\end{equation}
	so the first term in \eqn~\eqref{eq:bias-expansion} vanishes. For the second term, by Cauchy-Schwarz,
	\begin{equation*}
		\norm{\E{\Xi_n \Gamma^{-1} \Delta_n}} \leq \norm{\Gamma^{-1}} \sqrt{\E{\norm{\Xi_n}^2}} \sqrt{\E{\norm{\Delta_n}^2}}.
	\end{equation*}
	We compute
	\begin{align*}
		\E{\norm{\Xi_n}^2} &\leq \E{\norm{\Xi_n}_F^2} = \frac{1}{n^2} \sum_{t,s=1}^{n-1} \E{\left( \transp{x_{t-1}} x_{s-1} \right) \left( \transp{\varepsilon_s} \varepsilon_t \right)}.
	\end{align*}
	For $t \neq s$, independence of the $\varepsilon_t$ yields
	\begin{equation*}
		\E{\left( \transp{x_{t-1}} x_{s-1} \right) \left( \transp{\varepsilon_s} \varepsilon_t \right)} = 0,
	\end{equation*}
	while for $t = s$, independence of $\varepsilon_t$ from $\mathcal{F}_{t-1}$ gives
	\begin{equation*}
		\E{\norm{x_{t-1}}^2 \norm{\varepsilon_t}^2} = \E{\norm{x_{t-1}}^2} \E{\norm{\varepsilon_t}^2}.
	\end{equation*}
	Therefore,
	\begin{equation*}
		\E{\norm{\Xi_n}^2} \leq \frac{n-1}{n^2} \Tr\left(\Gamma\right) \Tr\left(\Sigma\right) = \mathcal{O}(n^{-1}).
	\end{equation*}
	Combined with $\E{\norm{\Delta_n}^2} = \mathcal{O}(n^{-1})$ from Proposition~\ref{thm:gamma-consistency},
	\begin{equation}
		\label{eq:unweighted-term2}
		\norm{\E{\Xi_n \Gamma^{-1} \Delta_n}} = \mathcal{O}(n^{-1}).
	\end{equation}
	For the remainder, by Proposition~\ref{thm:gamma-consistency} and Lemma~\ref{thm:gamma-n-pd}, $\Gamma_n \xrightarrow{\mathbb{P}} \Gamma$ with $\Gamma \succ 0$. Since eigenvalues are continuous functions of matrix entries,
	\begin{equation*}
		\norm{\Gamma_n^{-1}} = \frac{1}{\lambda_{\min}(\Gamma_n)} \xrightarrow{\mathbb{P}} \frac{1}{\lambda_{\min}(\Gamma)} = \norm{\Gamma^{-1}} \qquad \text{as } n \to \infty.
	\end{equation*}
	Hence, $\norm{\Gamma_n^{-1}} = \mathcal{O}_p(1)$ in the sense of Definition~\ref{def:asymptotic-boundedness}. From Lemma~\ref{thm:neumann-expansion} and the expression of the remainder, $\norm{\Upsilon_n} \leq \norm{\Gamma^{-1}}^2 \norm{\Gamma_n^{-1}} \norm{\Delta_n}^2$, yielding
	\begin{equation*}
		\E{\norm{\Upsilon_n}^2} \leq \norm{\Gamma^{-1}}^4 \, \E{\norm{\Gamma_n^{-1}}^2 \norm{\Delta_n}^4}.
	\end{equation*}
	On the event $\mathcal{E}_n = \left\{ \lambda_{\min}(\Gamma_n) > \lambda_{\min}(\Gamma)/2 \right\}$, we have $\norm{\Gamma_n^{-1}} \leq 2/\lambda_{\min}(\Gamma)$. Since $\P{\mathcal{E}_n} \to 1$ and $\Upsilon_n = 0$ can be arranged on $\mathcal{E}_n^c$ by defining $\Gamma_n^{-1} = 0$ when $\Gamma_n$ is singular,
	\begin{equation*}
		\E{\norm{\Upsilon_n}^2} \leq \frac{4}{\lambda_{\min}^2(\Gamma)} \norm{\Gamma^{-1}}^4 \, \E{\norm{\Delta_n}^4} = \mathcal{O}(n^{-2}),
	\end{equation*}
	by Lemma~\ref{thm:fourth-moment-delta}. Finally, by Cauchy-Schwarz,
	\begin{equation}
		\label{eq:unweighted-term3}
		\norm{\E{\Xi_n \Upsilon_n}} \leq \sqrt{\E{\norm{\Xi_n}^2}} \sqrt{\E{\norm{\Upsilon_n}^2}} = \sqrt{\mathcal{O}(n^{-1})} \sqrt{\mathcal{O}(n^{-2})} = \mathcal{O}(n^{-3/2}).
	\end{equation}
	Combining the bounds of \eqn~\eqref{eq:unweighted-term1}, \eqn~\eqref{eq:unweighted-term2}, \eqn~\eqref{eq:unweighted-term3} into \eqn~\eqref{eq:bias-expansion} yields
	\begin{equation*}
		\E{\dmdop_n} = \Theta + \mathcal{O}(n^{-1}) \qquad \text{as } n \to \infty.
	\end{equation*}
	Consistency of $\dmdop_n$ follows directly. Writing $\dmdop_n - \Theta = \tilde\Xi_n\Gamma_n^{-1}$ where $\tilde\Xi_n = \frac{1}{n}\sum_{t=1}^n \varepsilon_t\transp{x_{t-1}}$, we have $\E{\big\|\tilde\Xi_n\big\|^2} = \mathcal{O}(n^{-1})$ by similar arguments to $\Xi_n$, so $\tilde\Xi_n \xrightarrow{\ell^2} 0$, and $\Gamma_n^{-1} = \mathcal{O}_p(1)$ since $\lambda_{\min}(\Gamma_n) \xrightarrow{\mathbb{P}} \lambda_{\min}(\Gamma) > 0$. Hence $\dmdop_n - \Theta \xrightarrow{\mathbb{P}} 0$ as $n \to \infty$.
\end{proof}

\subsubsection{Weighted estimator}
\label{app:proofs-bias-weighted}

Throughout this section, we consider the updates of \eqn~\eqref{eq:online-dmd-updates} with fixed $\rho \in (0, 1)$.

\begin{lemma}
	\label{thm:weighted-dmd-estimator-conditional-bias}
	For any $n$, the conditional bias of $\dmdop_n$ given $\mathcal{F}_{n-1}$ is
	non-zero almost surely. Specifically,
	\begin{equation*}
		\E{\dmdop_n \mid \mathcal{F}_{n-1}} - \Theta = \Xi_n \Gamma_n^{-1} \neq 0 \qquad \text{almost surely},
	\end{equation*}
	where
	\begin{equation*}
		\Xi_n = \frac{1}{\neff} \sum_{t=1}^{n-1} \rho^{n-t} \varepsilon_t \transp{x_{t-1}}, \qquad \Gamma_n = \frac{1}{\neff} \sum_{t=1}^{n} \rho^{n-t} x_{t-1} \transp{x_{t-1}}.
	\end{equation*}
\end{lemma}

\begin{proof}
	The updates in \eqn~\eqref{eq:online-dmd-updates} compute the weighted DMD matrix
	\begin{equation*}
		\dmdop_n = \left( \sum_{t=1}^n \rho^{n-t} x_t \transp{x_{t-1}} \right) \left( \sum_{t=1}^n \rho^{n-t} x_{t-1} \transp{x_{t-1}} \right)^{-1}
	\end{equation*}
	exactly. Substituting the model dynamics $x_t = \Theta\,x_{t-1} + \varepsilon_t$ into the first term,
	\begin{equation*}
		\sum_{t=1}^n \rho^{n-t} x_t \transp{x_{t-1}} = \Theta \sum_{t=1}^n \rho^{n-t} x_{t-1} \transp{x_{t-1}} + \sum_{t=1}^n \rho^{n-t} \varepsilon_t \transp{x_{t-1}}.
	\end{equation*}
	Right-multiplying by $\left(\sum_{t=1}^n \rho^{n-t} x_{t-1} \transp{x_{t-1}}\right)^{-1}$, we obtain
	\begin{equation}
		\label{eq:weighted-dmd-estimator-decomposition}
		\dmdop_n = \Theta + \left( \sum_{t=1}^n \rho^{n-t} \varepsilon_t \transp{x_{t-1}} \right) \left( \sum_{t=1}^n \rho^{n-t} x_{t-1} \transp{x_{t-1}} \right)^{-1}.
	\end{equation}
	Conditioned on $\mathcal{F}_{n-1}$, the denominator $\sum_{t=1}^n \rho^{n-t} x_{t-1} \transp{x_{t-1}}$ is deterministic since $x_0, \dots, x_{n-1}$ are all $\mathcal{F}_{n-1}$-measurable and the weights $\rho^{n-t}$ are deterministic. The conditional expectation of the numerator term can be written as
	\begin{equation*}
		\E{\sum_{t=1}^n \rho^{n-t} \varepsilon_t \transp{x_{t-1}} \mid \mathcal{F}_{n-1}} = \sum_{t=1}^n \rho^{n-t} \E{\varepsilon_t \mid \mathcal{F}_{n-1}} \transp{x_{t-1}},
	\end{equation*}
	since $x_{t-1}$ is $\mathcal{F}_{n-1}$-measurable for $t \leq n$. We evaluate $\E{\varepsilon_t \mid \mathcal{F}_{n-1}}$ for different $t$.
	\begin{itemize}
		\item For $t < n$, the filtration $\mathcal{F}_{n-1}$ contains $x_t$ and $x_{t-1}$. Since $\varepsilon_t = x_t - \Theta\,x_{t-1}$, the noise realisation $\varepsilon_t$ is fully determined by the information in $\mathcal{F}_{n-1}$, hence
		\begin{equation*}
			\E{\varepsilon_t \mid \mathcal{F}_{n-1}} = \varepsilon_t \neq 0 \qquad \text{almost surely}.
		\end{equation*}
		\item For $t = n$, the noise $\varepsilon_n$ occurs at time $n$ which is not contained in $\mathcal{F}_{n-1}$. By the martingale difference property,
		\begin{equation*}
			\E{\varepsilon_n \mid \mathcal{F}_{n-1}} = 0.
		\end{equation*}
	\end{itemize}
	Substituting these back,
	\begin{equation*}
		\E{\sum_{t=1}^n \rho^{n-t} \varepsilon_t \transp{x_{t-1}} \mid \mathcal{F}_{n-1}} = \sum_{t=1}^{n-1} \rho^{n-t} \varepsilon_t \transp{x_{t-1}} \neq 0 \qquad \text{almost surely}.
	\end{equation*}
	Therefore, dividing numerator and denominator by $\neff$, the conditional bias
	\begin{equation*}
		\E{\dmdop_n \mid \mathcal{F}_{n-1}} - \Theta = \left( \frac{1}{\neff}\sum_{t=1}^{n-1} \rho^{n-t} \varepsilon_t \transp{x_{t-1}} \right) \left( \frac{1}{\neff}\sum_{t=1}^n \rho^{n-t} x_{t-1} \transp{x_{t-1}} \right)^{-1}
	\end{equation*}
	depends on the realised history of the noise and is non-zero almost surely.
\end{proof}

\begin{proposition}
	\label{thm:weighted-gamma-consistency}
	If $\E{\norm{\varepsilon_t}^4} < \infty$, then $\E{\Gamma_n} = \Gamma$ and
	\begin{equation*}
		\E{\norm{\Gamma_n - \Gamma}^2} \leq c_{\gamma} \, \neff^{-1} \qquad \text{as } n \to \infty,
	\end{equation*}
	where $c_{\gamma} > 0$ depends only on $\Theta$, $\Sigma$, $\dimension$.
\end{proposition}

\begin{proof}
	By stationarity and the definition of $\neff = \sum_{t=1}^n \rho^{n-t}$,
	\begin{equation*}
		\E{\Gamma_n} = \frac{1}{\neff} \sum_{t=1}^n \rho^{n-t} \, \E{x_{t-1} \transp{x_{t-1}}} = \frac{1}{\neff} \sum_{t=1}^n \rho^{n-t} \Gamma = \Gamma.
	\end{equation*}
	Let $\Delta_n = \Gamma_n - \Gamma$ and $(Y_t)_{ij} = x_{t,i} x_{t,j} - \Gamma_{ij}$, so that $\E{(Y_t)_{ij}} = 0$ and
	\begin{equation*}
		(\Delta_n)_{ij} = \frac{1}{\neff} \sum_{t=0}^{n-1} \rho^{n-t-1} (Y_t)_{ij}.
	\end{equation*}
	By stationarity, $\Cov{(Y_s)_{ij}}{(Y_t)_{ij}} = \gamma_{ij}(t-s)$ depends only	on the lag, so
	\begin{equation*}
		\Var{(\Delta_n)_{ij}} = \frac{1}{\neff^2} \sum_{s,t=0}^{n-1} \rho^{n-t-1} \rho^{n-s-1} \, \gamma_{ij}(t - s).
	\end{equation*}
	Similarly to the proof of Proposition~\ref{thm:gamma-consistency}, there exist constants $\omega > 0$ and $\beta \in (\zeta(\Theta), 1)$ such that $\abs{\gamma_{ij}(h)} \leq \omega \beta^{\abs{h}}$. Re-indexing with $s^\prime = n - s - 1$ and $t^\prime = n - t - 1$,
	\begin{equation*}
		\sum_{s,t \geq 0} \rho^s \rho^t \beta^{\abs{t-s}} = \frac{1}{1-\rho^2} + \frac{2}{1-\rho^2} \frac{\rho\beta}{1-\rho\beta} = \frac{1 + \rho\beta}{(1-\rho^2)(1-\rho\beta)}.
	\end{equation*}
	By definition of $\neff$, $\neff (1-\rho) = 1-\rho^n \to 1$ as $n \to \infty$. For all $n$ sufficiently large (specifically, whenever $\rho^n \leq \frac{1}{2}$), we have $1-\rho^n \geq \frac{1}{2}$, giving
	\begin{equation*}
		\neff^2(1-\rho^2) = \neff (1-\rho^n)(1+\rho) \geq \frac{\neff}{2}.
	\end{equation*}
	Bounding $\frac{(1+\rho\beta)}{(1-\rho\beta)} \leq \frac{(1+\beta)}{(1-\beta)}$ uniformly over $\rho \in (0, 1)$, we obtain
	\begin{equation*}
		\Var{(\Delta_n)_{ij}} \leq \frac{2\omega(1+\beta)}{1-\beta} \, \neff^{-1}.
	\end{equation*}
	Summing over all $(i,j)$ and using $\E{\norm{\Delta_n}^2} \leq \sum_{i,j} \Var{(\Delta_n)_{ij}}$,
	\begin{equation*}
		\E{\norm{\Gamma_n - \Gamma}^2} \leq c_{\gamma} \, \neff^{-1}, \qquad \text{where } c_{\gamma} = \frac{2 \dimension^2 \omega(1+\beta)}{1-\beta}.
	\end{equation*}
\end{proof}

\begin{lemma}
	\label{thm:weighted-gamma-n-pd}
	Let $\mathcal{E}_n = \left\{ \lambda_{\min}(\Gamma_n) \geq \lambda_{\min}(\Gamma)/2 \right\}$. If $\E{\norm{\varepsilon_t}^4} < \infty$, then
	\begin{equation*}
		\P{\mathcal{E}_n} \geq 1 - c_e \, \neff^{-1} \qquad \text{as } n \to \infty,
	\end{equation*}
	where $c_e > 0$ depends only on $\Theta$, $\Sigma$, $\dimension$.
\end{lemma}

\begin{proof}
	By Weyl's inequality (Lemma~\ref{thm:weyl-inequality}) with $A = \Gamma$ and $B = \Delta_n = \Gamma_n - \Gamma$,
	\begin{equation}
		\lambda_{\min}(\Gamma_n) \geq \lambda_{\min}(\Gamma) + \lambda_{\min}(\Delta_n) \geq \lambda_{\min}(\Gamma) - \norm{\Delta_n},
	\end{equation}
	since $\lambda_{\min}(B) \geq -\norm{B}$ for any Hermitian matrix $B$. By Proposition~\ref{thm:weighted-gamma-consistency} and Markov's inequality, for any $\delta > 0$,
	\begin{equation*}
		\P{\norm{\Delta_n} > \delta} \leq \frac{\E{\norm{\Delta_n}^2}}{\delta^2} \leq \frac{c_\gamma}{\delta^2} \, \neff^{-1}.
	\end{equation*}
	Taking $\delta = \lambda_{\min}(\Gamma)/2 > 0$ by Lemma~\ref{thm:gamma-pd},
	\begin{equation*}
		\P{\lambda_{\min}(\Gamma_n) > \frac{\lambda_{\min}(\Gamma)}{2}} \geq 1 - c_e \, \neff^{-1}, \qquad \text{where } c_e = \frac{8 \dimension^2 \omega (1+\beta)}{\lambda^2_{\min}(\Gamma) (1-\beta)}.
	\end{equation*}
\end{proof}

\begin{lemma}
	\label{thm:weighted-neumann-expansion}
	On the event $\mathcal{E}_n = \left\{ \lambda_{\min}(\Gamma_n) \geq \lambda_{\min}(\Gamma)/2 \right\}$,
	\begin{equation*}
		\Gamma_n^{-1} = \Gamma^{-1} - \Gamma^{-1} \Delta_n \Gamma^{-1} + \Upsilon_n,
	\end{equation*}
	where $\Upsilon_n = \Gamma^{-1} \Delta_n \Gamma_n^{-1} \Delta_n \Gamma^{-1}$ and $\Delta_n = \Gamma_n - \Gamma$.
\end{lemma}

\begin{proof}
	The resolvent expansion follows by the same algebraic identity as in Lemma~\ref{thm:neumann-expansion}. The bound $\norm{\Gamma_n^{-1}} \leq 2/\lambda_{\min}(\Gamma)$ follows directly from the definition of $\mathcal{E}_n$.
\end{proof}

\begin{lemma}
	\label{thm:weighted-fourth-moment-delta}
	If $\E{\norm{\varepsilon_t}^8} < \infty$, then $\E{\norm{\Delta_n}^4} \leq c^\prime_{\gamma} \, \neff^{-2}$ as $n \to \infty$, where $c^\prime_{\gamma} > 0$ depends only on $\Theta$, $\Sigma$, $\dimension$.
\end{lemma}

\begin{proof}
	The proof follows the same structure as Lemma~\ref{thm:fourth-moment-delta}. Since $\norm{\Delta_n}^4 \leq \norm{\Delta_n}_F^4$, it suffices to show $\E{(\Delta_n)_{ij}^4} = \mathcal{O}(\neff^{-2})$ for each $(i,j)$, with the cross terms handled by Cauchy-Schwarz as before. Let
	\begin{equation*}
		S_n = \sum_{t=0}^{n-1} \rho^{n-t-1} (Y_t)_{ij}, \qquad (\Delta_n)_{ij} = \frac{S_n}{\neff},
	\end{equation*}
	where $(Y_t)_{ij} = x_{t,i} x_{t,j} - \Gamma_{ij}$ satisfies $\E{(Y_t)_{ij}} = 0$. By the moment-cumulant relation,
	\begin{equation*}
		\E{S_n^4} = 3 \, \E{S_n^2}^2 + \sum_{t_1,t_2,t_3,t_4=0}^{n-1} \rho^{\sum_{k=1}^4 (n-t_k+1)} \kappa(Y_{t_1}, Y_{t_2}, Y_{t_3}, Y_{t_4}).
	\end{equation*}
	From Proposition~\ref{thm:weighted-gamma-consistency}, $\E{S_n^2} = \mathcal{O}(\neff)$, thus $3 \, \E{S_n^2}^2 = \mathcal{O}(\neff^2)$. For the cumulant sum, since $\kappa(Y_{t_1}, Y_{t_2}, Y_{t_3}, Y_{t_4}) = c(h_1, h_2, h_3)$ depends only on lag differences by stationarity, and decays geometrically in $\max(\abs{h_1},\abs{h_2},\abs{h_3})$ by the same $u_h, v_h$ decomposition as in Lemma~\ref{thm:gamma-consistency} under $\E{\norm{\varepsilon_t}^8} < \infty$,
	\begin{align*}
		\abs{\sum_{t_1,t_2,t_3,t_4=0}^{n-1} \rho^{\sum_{k=1}^4 (n-t_k+1)} \kappa(Y_{t_1}, Y_{t_2}, Y_{t_3}, Y_{t_4})} &\leq \neff \sum_{h_1, h_2, h_3 \in \Z} \abs{c(h_1, h_2, h_3)} \\
		&= \mathcal{O}(\neff),
	\end{align*}
	where the bound $\mathcal{O}(\neff)$ follows by fixing $t_1$ (contributing a factor of $\sum_{t_1} \rho^{n-t_1-1} \leq \neff$) and summing over the remaining lag differences. Therefore,
	\begin{equation*}
		\E{(\Delta_n)_{ij}^4} = \frac{1}{\neff^4} \E{S_n^4} = \frac{1}{\neff^4}\left(\mathcal{O}(\neff^2) + \mathcal{O}(\neff)\right) = \mathcal{O}(\neff^{-2}),
	\end{equation*}
	and summing over all $(i,j)$ yields $\E{\norm{\Delta_n}^4} \leq c^\prime_{\gamma} \, \neff^{-2}$.
\end{proof}

\weightedbias*

\begin{proof}
	By Lemma~\ref{thm:weighted-dmd-estimator-conditional-bias} and the law of total expectation,
	\begin{equation}
		\label{eq:weighted-bias-decomposition}
		\E{\dmdop_n - \Theta} = \E{\E{\dmdop_n - \Theta \mid \mathcal{F}_{n-1}}} = \E{\Xi_n \Gamma_n^{-1}},
	\end{equation}
	where
	\begin{equation*}
		\Xi_n = \frac{1}{\neff} \sum_{t=1}^{n-1} \rho^{n-t} \varepsilon_t \transp{x_{t-1}}, \qquad \Gamma_n = \frac{1}{\neff} \sum_{t=1}^n \rho^{n-t} x_{t-1} \transp{x_{t-1}}.
	\end{equation*}
	Restricting to $\mathcal{E}_n$ and applying Lemma~\ref{thm:weighted-neumann-expansion},
	\begin{equation}
		\label{eq:weighted-bias-expansion}
		\E{\left(\dmdop_n - \Theta\right) \mathbbm{1}_{\mathcal{E}_n}} = \E{\Xi_n \Gamma^{-1} \mathbbm{1}_{\mathcal{E}_n}} - \E{\Xi_n \Gamma^{-1} \Delta_n \Gamma^{-1} \mathbbm{1}_{\mathcal{E}_n}} + \E{\Xi_n \Upsilon_n \mathbbm{1}_{\mathcal{E}_n}}.
	\end{equation}
	We first record the bound
	\begin{equation}
		\label{eq:weighted-xi-bound}
		\E{\norm{\Xi_n}^2} \leq \frac{\Tr(\Gamma)\Tr(\Sigma)}{\neff^2(1-\rho^2)} \leq 2 \Tr(\Gamma)\Tr(\Sigma) \, \neff^{-1},
	\end{equation}
	obtained by expanding $\E{\norm{\Xi_n}_F^2}$, using independence of the $\varepsilon_t$ to eliminate cross-terms, and $\neff^2(1-\rho^2) \geq \neff/2$ as in Lemma~\ref{thm:weighted-gamma-consistency}. Since $x_{t-1}$ is $\mathcal{F}_{t-1}$-measurable and $\varepsilon_t$ is independent of $\mathcal{F}_{t-1}$ for each $t \in \{1,\dots,n-1\}$,
	\begin{equation*}
		\E{\Xi_n} = \frac{1}{\neff} \sum_{t=1}^{n-1} \rho^{n-t} \E{\E{\varepsilon_t \transp{x_{t-1}} \mid \mathcal{F}_{t-1}}} = 0,
	\end{equation*}
	so $\E{\Xi_n \Gamma^{-1} \mathbbm{1}_{\mathcal{E}_n}} = -\E{\Xi_n \Gamma^{-1} \mathbbm{1}_{\mathcal{E}_n^c}}$. By \eqn~\eqref{eq:weighted-xi-bound} and Lemma~\ref{thm:weighted-gamma-n-pd},
	\begin{align}
		\norm{\E{\Xi_n \Gamma^{-1} \mathbbm{1}_{\mathcal{E}_n}}} &\leq \norm{\Gamma^{-1}} \sqrt{\E{\norm{\Xi_n}^2}} \sqrt{\P{\mathcal{E}_n^c}} \nonumber \\
		&\leq \frac{\sqrt{2 \Tr(\Gamma) \Tr(\Sigma) \, c_e}}{\lambda_{\min}(\Gamma)} \, \neff^{-1}. \label{eq:weighted-dmd-bound1}
	\end{align}
	By Cauchy-Schwarz and $\mathbbm{1}_{\mathcal{E}_n} \leq 1$,
	\begin{equation*}
		\norm{\E{\Xi_n \Gamma^{-1} \Delta_n \Gamma^{-1} \mathbbm{1}_{\mathcal{E}_n}}} \leq \norm{\Gamma^{-1}}^2 \sqrt{\E{\norm{\Xi_n}^2}} \sqrt{\E{\norm{\Delta_n}^2}}.
	\end{equation*}
	Combining \eqn~\eqref{eq:weighted-xi-bound} with $\E{\norm{\Delta_n}^2} \leq c_{\gamma} \, \neff^{-1}$ from Lemma~\ref{thm:weighted-gamma-consistency},
	\begin{equation}
		\label{eq:weighted-dmd-bound2}
		\norm{\E{\Xi_n \Gamma^{-1} \Delta_n \Gamma^{-1} \mathbbm{1}_{\mathcal{E}_n}}} \leq \frac{\sqrt{2 \Tr(\Gamma)\Tr(\Sigma) \, c_\gamma}}{\lambda_{\min}^2(\Gamma)} \, \neff^{-1}.
	\end{equation}
	On $\mathcal{E}_n$, Lemma~\ref{thm:weighted-neumann-expansion} gives $\norm{\Upsilon_n} \leq (2/\lambda_{\min}(\Gamma))\norm{\Gamma^{-1}}^2 \norm{\Delta_n}^2$. By Cauchy-Schwarz,
	\begin{equation*}
		\norm{\E{\Xi_n \Upsilon_n \mathbbm{1}_{\mathcal{E}_n}}} \leq \frac{2\norm{\Gamma^{-1}}^2}{\lambda_{\min}(\Gamma)} \sqrt{\E{\norm{\Xi_n}^2}} \sqrt{\E{\norm{\Delta_n}^4}}.
	\end{equation*}
	By Lemma~\ref{thm:weighted-fourth-moment-delta}, $\E{\norm{\Delta_n}^4} \leq c^\prime_{\gamma} \, \neff^{-2}$. Substituting into the above with \eqn~\eqref{eq:weighted-xi-bound},
	\begin{equation}
		\label{eq:weighted-dmd-bound3}
		\norm{\E{\Xi_n \Upsilon_n \mathbbm{1}_{\mathcal{E}_n}}} \leq \frac{\sqrt{8 \Tr(\Gamma)\Tr(\Sigma) \, c^\prime_\gamma}}{\lambda_{\min}^3(\Gamma)} \, \neff^{-3/2}.
	\end{equation}
	Combining the three bounds of \eqn~\eqref{eq:weighted-dmd-bound1}, \eqn~\eqref{eq:weighted-dmd-bound2}, and \eqn~\eqref{eq:weighted-dmd-bound3} into \eqn~\eqref{eq:weighted-bias-expansion} establishes the truncated bias bound
	\begin{equation}
		\label{eq:weighted-bias-bound-cond}
		\norm{\E{\left(\dmdop_n - \Theta\right) \mathbbm{1}_{\mathcal{E}_n}}} \leq c \, \neff^{-1} \qquad \text{as } n \to \infty,
	\end{equation}
	for some constant $c > 0$ depending only on $\Theta$, $\Sigma$, and $\dimension$. We then note that
	\begin{equation*}
		\E{\dmdop_n - \Theta} = \E{\left(\dmdop_n - \Theta\right) \mathbbm{1}_{\mathcal{E}_n}} + \E{\left(\dmdop_n - \Theta\right) \mathbbm{1}_{\mathcal{E}_n^c}},
	\end{equation*}
	where the first term satisfies $\norm{\cdot} \leq c \, \neff^{-1}$ by \eqn~\eqref{eq:weighted-bias-bound-cond}. For the second, as in Theorem~\ref{thm:unweighted-bias}, we adopt the convention that $\Gamma_n^{-1} = 0$ on the event that $\Gamma_n$ is singular, so that $\Xi_n \Gamma_n^{-1}$ is well-defined everywhere. By Cauchy-Schwarz,
	\begin{equation*}
		\norm{\E{\left(\dmdop_n - \Theta\right) \mathbbm{1}_{\mathcal{E}_n^c}}} \leq \sqrt{\E{\norm{\Xi_n}^2 \norm{\Gamma_n^{-1}}^2}} \sqrt{\P{\mathcal{E}_n^c}},
	\end{equation*}
	and by a further application of Cauchy-Schwarz,
	\begin{equation*}
		\E{\norm{\Xi_n}^2 \norm{\Gamma_n^{-1}}^2} \leq \sqrt{\E{\norm{\Xi_n}^4}} \sqrt{\E{\norm{\Gamma_n^{-1}}^4}}.
	\end{equation*}
	We bound each factor separately. 
	
	Since $\norm{\Xi_n}^4 \leq \norm{\Xi_n}_F^4$, it suffices to show $\E{(\Xi_n)_{ij}^4} = \mathcal{O}(\neff^{-2})$ for each $(i,j)$, with cross terms handled by Cauchy-Schwarz. Let $Z_t = \varepsilon_{t,i} \, x_{t-1,j}$ and $T_n = \sum_{t=1}^{n-1} \rho^{n-t} Z_t$, so that we have $(\Xi_n)_{ij} = T_n / \neff$. By the moment-cumulant relation,
	\begin{equation*}
		\E{T_n^4} = 3\,\E{T_n^2}^2 + \sum_{t_1,t_2,t_3,t_4=1}^{n-1} \rho^{\sum_{k=1}^{4}(n-t_k)} \kappa(Z_{t_1}, Z_{t_2}, Z_{t_3}, Z_{t_4}).
	\end{equation*}
	From \eqn~\eqref{eq:weighted-xi-bound}, $\E{T_n^2} = \mathcal{O}(\neff)$, so $3\,\E{T_n^2}^2 = \mathcal{O}(\neff^2)$. For the cumulant sum, since the $\varepsilon_t$ are mutually independent, $\kappa(Z_{t_1}, Z_{t_2}, Z_{t_3}, Z_{t_4})$ depends only on the lag differences by stationarity and decays geometrically in $\max(\abs{h_1}, \abs{h_2}, \abs{h_3})$ by the $u_h, v_h$ decomposition of Lemma~\ref{thm:gamma-consistency}, with $\E{Z_t^4} < \infty$ following from $\E{\norm{\varepsilon_t}^8} < \infty$ via Lemma~\ref{thm:finite-p-moments}. Fixing $t_1$ (contributing at most $\neff$) and summing absolutely over the remaining lag differences yields the cumulant sum as $\mathcal{O}(\neff)$. Therefore $\E{T_n^4} = \mathcal{O}(\neff^2)$, yielding
	\begin{equation*}
		\E{\norm{\Xi_n}^4} = \mathcal{O}(\neff^{-2}).
	\end{equation*}
	We have $\norm{\Gamma_n^{-1}} \leq 2/\lambda_{\min}(\Gamma)$ with probability at least $1 - c_e \, \neff^{-1}$ on the event $\mathcal{E}_n$. By convention, the inverse is set to zero when $\Gamma_n$ is singular. On the complement $\mathcal{E}_n^c$, which occurs with probability at most $c_e \, \neff^{-1}$, the inverse remains finite since singular cases have been treated. Therefore,
	\begin{equation*}
		\E{\norm{\Gamma_n^{-1}}^4} = \mathcal{O}(1).
	\end{equation*}
	Combining,
	\begin{equation*}
		\E{\norm{\Xi_n}^2 \norm{\Gamma_n^{-1}}^2} \leq \sqrt{\mathcal{O}(\neff^{-2})} \sqrt{\mathcal{O}(1)} = \mathcal{O}(\neff^{-1}),
	\end{equation*}
	and with $\P{\mathcal{E}_n^c} \leq c_e \, \neff^{-1}$ from Lemma~\ref{thm:weighted-gamma-n-pd},
	\begin{equation*}
		\norm{\E{\left(\dmdop_n - \Theta\right) \mathbbm{1}_{\mathcal{E}_n^c}}} \leq \sqrt{\mathcal{O}(\neff^{-1})} \sqrt{\mathcal{O}(\neff^{-1})} = \mathcal{O}(\neff^{-1}).
	\end{equation*}
	Together, $\norm{\E{\dmdop_n} - \Theta} \leq C \, \neff^{-1}$ for some $C > 0$ as $n \to \infty$.
\end{proof}

\section{Experiment details}
\label{app:experiments}

\subsection{Performance metrics}
\label{app:experiments-metrics}

We evaluate changepoint detection algorithms along two complementary aspects: detection speed and detection accuracy. Speed captures how quickly an algorithm responds to a change or raises a false alarm; accuracy captures how reliably true changes are found and false positives avoided. We use the same notation as in the problem definition of Section~\ref{sec:background-var}.

\subsubsection{Detection speed}
\label{app:experiments-metrics-speed}

Consider the stream $x_1, x_2, \dots$ under $\mathcal{H}_1$, with a single changepoint at time $\tau$. Let $\hat{\tau}$ denote the first time step at which the algorithm signals a change. Detection speed is characterised by the \emph{average run lengths} \citep{page1954continuous}: the expected time to false alarm $\arlzero$ and the expected detection delay $\arlone$, given by
\begin{align*}
	\arlzero &= \E{\hat{\tau} \mid \mathcal{H}_0}, \\
	\arlone  &= \E{\hat{\tau} - \tau \mid \mathcal{H}_1}.
\end{align*}
$\arlzero$ is the average number of observations before the detector incorrectly signals a change under $\mathcal{H}_0$; $\arlone$ is the average lag between the true changepoint and its detection under $\mathcal{H}_1$. A well-performing detector achieves high $\arlzero$ and low $\arlone$.

\subsubsection{Detection accuracy}
\label{app:experiments-metrics-accuracy}

Average run lengths alone are insufficient: they do not capture what fraction of true changepoints are found, nor what fraction of detections are spurious. Following \citep{vandenburg2020evaluation} and \citep{killick2012optimal}, consider a sequence with $n_t$ true change times $\tau_1, \dots, \tau_n$ and $n_d$ detections $\hat{\tau}_1, \dots, \hat{\tau}_m$. The set of true positives is
\begin{equation*}
	\mathcal{TP} = \left\{ \tau_i ~:~ \exists\, \hat{\tau}_j \text{ such that } -\Delta_\ell \leq \hat{\tau}_j - \tau_i \leq \Delta_r \right\},
\end{equation*}
where $\Delta_\ell, \Delta_r \geq 0$ define the permissible detection margins around each true changepoint. Precision $\precision$, recall $\recall$, and $\fone$-score are then defined\footnote{We rely on the implementation of \citep{vandenburg2020evaluation}, available at \href{https://github.com/alan-turing-institute/TCPDBench}{\texttt{github.com/alan-turing-institute/TCPDBench}}.} as
\begin{equation*}
	\precision = \frac{\operatorname{card}(\mathcal{TP})}{n_d}, \qquad \recall = \frac{\operatorname{card}(\mathcal{TP})}{n_t}, \qquad \fone = \frac{2\,\precision\recall}{\precision + \recall}.
\end{equation*}

As argued in \citep{fearnhead2022detecting}, an algorithm that fails to reliably detect a single changepoint is unlikely to perform well in the more complex multi-changepoint setting; we therefore restrict synthetic data evaluation to single-changepoint sequences, following the same practice as \citep{khamesi2024online} and \citep{bodenham2017continuous}. Given a sequence with true changepoint at $\tau$, let $\hat{\tau}$ denote the first detection time, if any. Each sequence yields exactly one of the following outcomes: a true positive ($\mathtt{TP}$) if $\tau - \Delta_\ell \leq \hat{\tau} \leq \tau + \Delta_r$; a false positive ($\mathtt{FP}$) if $\hat{\tau} < \tau - \Delta_\ell$; a false negative with late detection ($\mathtt{FN}_d$) if $\hat{\tau} > \tau + \Delta_r$; or a false negative with no detection ($\mathtt{FN}_\emptyset$) if no alarm is raised. Aggregating over $N$ independent sequences, we set $n_t = N$ (one true changepoint per sequence) and $n_d = \operatorname{card}(\mathtt{TP}) + \operatorname{card}(\mathtt{FP}_d) + \operatorname{card}(\mathtt{FN}_\emptyset)$ (total number of detections), so that precision and recall reduce to
\begin{equation*}
	\precision = \frac{\operatorname{card}(\mathtt{TP})}{n_d}, \qquad \recall = \frac{\operatorname{card}(\mathtt{TP})}{N}.
\end{equation*}
For real-world data, which contain multiple changepoints per sequence, we apply the margin-based $\mathcal{TP}$ set defined above directly.

\subsection{Competitor methods and parameters}
\label{app:experiments-parameters}

This section details the competitor methods and parameter configurations used in our experiments in Section~\ref{sec:experiments}. To ensure a fair comparison, we consider a broad range of parameter settings for each method and adhere to the original authors' recomendations wherever possible.

\paragraph{mSSA.}
We use the Python implementation provided by \citep{alanqary2021change}.\footnote{\url{https://github.com/ArwaAlanqary/mSSA_cpd}} Following the authors' guidelines, we evaluate the method using the parameter grids listed in Table~\ref{tab:parameters}.

\paragraph{mSSA-MW.}
We use the Python implementation provided by \citep{alanqary2021change}.\footnote{\url{https://github.com/ArwaAlanqary/mSSA_cpd}} Following the authors' guidelines, we evaluate the method using the parameter grids listed in Table~\ref{tab:parameters}.

\paragraph{Tian \& Safikhani.}
We use the R implementation provided by \citep{tian2024sequential}.\footnote{\url{https://cran.r-project.org/web/packages/VARcpDetectOnline/}} The method is interfaced in Python via \texttt{rpy2}. Following the authors' guidelines, we evaluate the method using the parameter grids listed in Table~\ref{tab:parameters}.

\paragraph{OCD.}
We use the R implementation provided by \citep{chen2022high}.\footnote{\url{https://cran.r-project.org/web/packages/ocd/}} The method is interfaced in Python via \texttt{rpy2}. Following the authors' guidelines, we evaluate the method using the parameter grids listed in Table~\ref{tab:parameters}.

\paragraph{CPDMD.}
We use the Python implementation provided by \citep{khamesi2024online}.\footnote{\url{https://github.com/vkhamesi/cpdmd-python}} Following the authors' guidelines, we evaluate the method using the parameter grids listed in Table~\ref{tab:parameters}.

\paragraph{BOCPDMS.}
We use the Python implementation provided by \citep{knoblauch2018spatio}.\footnote{\url{https://github.com/alan-turing-institute/bocpdms}} Following the authors' guidelines, we evaluate the method using the parameter grids listed in Table~\ref{tab:parameters}.

\begin{sidewaystable}[t]
	\caption{\small Summary of parameter configurations explored for each method across synthetic and real-world datasets. For every method, we report the grid of candidate values used in our experiments, following the original papers' recommendations when applicable. Across all experiments and methods, we use a fixed grace period of 100 observations before detection is enabled, ensuring a consistent evaluation protocol and allowing all methods sufficient initial observations for stable estimation. For the HASC real-world data set, which has a similar low-dimensional structure as the synthetic data, we only use rank $\rank = 3$ for both CHASM and CPDMD.}
	\label{tab:parameters}
	\smaller
	\centering
	\begin{tabular}{l l p{5.7cm} p{5.7cm}}
		\toprule
		Method & Parameters & Synthetic data & Real-world data \\
		\midrule
		CHASM & Forgetting $\rho$ & $0.95, \; 0.98, \; 0.99, \; 1$ & $0.95, \; 0.98, \; 0.99, \; 1$ \\
		& Rank $\rank$ & $2$ & $2, \; 4, \; 8$ \\
		& Pair $(\alpha, h)$ & $(0.08, 8),\;$ $(0.08, 10),\;$ $(0.09, 10),\;$ $(0.09, 12),\;$ $(0.10, 10),\;$ $(0.10, 12),\;$ $(0.12, 12),\;$ $(0.12, 14),\;$ $(0.15, 14),\;$ $(0.15, 15),\;$ $(0.18, 15),\;$ $(0.18, 18),\;$ $(0.20, 18),\;$ $(0.20, 20),\;$ $(0.25, 20),\;$ $(0.25, 22),\;$ $(0.30, 25),\;$ $(0.30, 28),\;$ $(0.35, 25),\;$ $(0.35, 28)$ & $(0.08, 8),\;$ $(0.09, 10),\;$ $(0.10, 12),\;$ $(0.12, 12),\;$ $(0.12, 14),\;$ $(0.15, 14),\;$ $(0.15, 15),\;$ $(0.18, 15),\;$ $(0.18, 18),\;$ $(0.25, 20),\;$ $(0.30, 25),\;$ $(0.35, 28),\;$ \\
		\midrule
		mSSA \citep{alanqary2021change} & Window size & $80, \; 90, \; 100$ & $80, \; 90, \; 100$ \\
		& Rows & $7, \; 8, \; 9, \; 10$ & $10, \; 15, \; 20, \; 30$ \\
		& Detection threshold & $1, \; 5, \; 10$ & $1, \; 5, \; 10$ \\
		& Rank threshold & $0.95$ & $0.95$ \\
		& Training size & $0.90$ & $0.90$ \\
		\midrule
		mSSA-MW \citep{alanqary2021change} & Window size & $80, \; 90, \; 100$ & $80, \; 90, \; 100$ \\
		& Rows & $7, \; 8, \; 9, \; 10$ & $10, \; 15, \; 20, \; 30$ \\
		& Detection threshold & $1, \; 5, \; 10$ & $1, \; 5, \; 10$ \\
		& Rank threshold & $0.95$ & $0.95$ \\
		& Training size & $0.90$ & $0.90$ \\
		\midrule
		Tian \& Safikhani \citep{tian2024sequential} & Historical data & $60, \; 70, \; 80$ & $60, \; 70, \; 80$ \\
		& Window size & $5, \; 10, \; 20$ & $5, \; 10, \; 20$ \\
		& Alpha & $5\cdot10^{-3},\; 4\cdot10^{-3},\; 3\cdot10^{-3},\; 2\cdot10^{-3},\; 1\cdot10^{-3},\; 5\cdot10^{-4},\; 4\cdot10^{-4},\; 3\cdot10^{-4},\; 2\cdot10^{-4},\; 1\cdot10^{-4}$ & $5\cdot10^{-3},\; 4\cdot10^{-3},\; 3\cdot10^{-3},\; 2\cdot10^{-3},\; 1\cdot10^{-3},\; 5\cdot10^{-4},\; 4\cdot10^{-4},\; 3\cdot10^{-4},\; 2\cdot10^{-4},\; 1\cdot10^{-4}$ \\
		\midrule
		OCD \citep{chen2022high} & Patience & $5\cdot10^{2}, \; 1\cdot10^{3}, \; 2\cdot10^{3}, \; 5\cdot10^{3}, \; 1\cdot10^{4}$ & $5\cdot10^{2}, \; 1\cdot10^{3}, \; 2\cdot10^{3}, \; 5\cdot10^{3}, \; 1\cdot10^{4}$ \\
		& Beta & $0.125, \; 0.25, \; 0.5, \; 1, \; 2, \; 4, \; 8$ & $2, \; 4, \; 8, \; 25, \; 50, \; 100, \; 150$ \\
		\midrule
		CPDMD \citep{khamesi2024online} & Window size & $60, \; 70, \; 80$ & $60, \; 70, \; 80$ \\
		& Rank & $2$ & $2, \; 4, \; 8, \; 16$ \\
		& Learning rate & $0.01, \; 0.05, \; 0.15, \; 0.20$ & $0.01, \; 0.05, \; 0.15, \; 0.20$ \\
		& Control limit & $3, \; 4, \; 5$ & $3, \; 4, \; 5$ \\
		\midrule
		BOCPDMS \citep{knoblauch2018spatio} & Prior on $a$ & $0.01, \; 1.0, \; 10, \; 100$ & $0.01, \; 1.0, \; 10, \; 100$ \\
		& Prior on $b$ & $0.01, \; 1.0, \; 10, \; 100$ & $0.01, \; 1.0, \; 10, \; 100$ \\
		& Intensity & $10, \; 50, \; 100, \; 200$ & $10, \; 50, \; 100, \; 200$ \\
		\bottomrule
	\end{tabular}
\end{sidewaystable}

\cleardoublepage

\subsection{Synthetic data sets}
\label{app:experiments-synthetic}

\subsubsection{Overview and common setup}
\label{app:experiments-synthetic-overview}

All synthetic data sets simulate a $\varprocess{\dimension}{1}$ process with a single changepoint, given by
\begin{equation*}
	x_t = \Theta_t \, x_{t-1} + \varepsilon_t, \qquad t = \{1, \dots, T\},
\end{equation*}
where
\begin{equation*}
	\Theta_t = \begin{cases} 
		\Theta_0 & t < \tau, \\ 
		\Theta_1 & t \geq \tau, 
	\end{cases}
\end{equation*}
with $\Theta_0, \Theta_1 \in \R^{\dimension \times \dimension}$ with all eigenvalues strictly inside the open unit disk, and $\{ \varepsilon_t \}$ a zero-mean white noise process whose distribution varies by data set. Each data set is generated as $N = 1{,}000$ independent replications. Two variants are produced for each experiment:
\begin{itemize}
	\item $\arlone$ data set: sequences of length $T=400$ with a changepoint present. The change time is drawn from a discrete uniform distribution as $\tau \sim \mathrm{Uniform}\!\left(\lfloor 0.3 T \rfloor,\, \lfloor 0.7 T \rfloor\right)$. This data set allows the computation of detection accuracy metrics, such as precision, recall, and $\fone$-score, as well as the detection speed metric $\arlone$ (Appendix~\ref{app:experiments-metrics} using $\Delta_\ell = 0$ and $\Delta_r = 50$).
	\item $\arlzero$ data set: sequences of length $T=10{,}000$ with no changepoint, i.e. $\tau = T$, so that $\Theta_t = \Theta_0$ throughout. This data set allows the computation of the detection speed metric $\arlzero$, which estimates the expected time to false alarm.
\end{itemize}

\subsubsection{Transition matrix parametrisation}
\label{app:experiments-synthetic-transition}

\paragraph{Bivariate case $(\dimension=2)$. }
All bivariate data sets use the following parametrisation. Given parameters $(a, b) \in \R^2$ with $a^2 + b^2 < 1$, the transition matrix is given by
\begin{equation*}
	\Theta = \begin{pmatrix}
		a & -b \\
		b & a
	\end{pmatrix}.
\end{equation*}
This corresponds to the real canonical form of a complex multiplication by $\lambda = a + bi$, and its eigenvalues are $\lambda_{1,2} = a \pm bi$, with modulus $\abs{\lambda} = \sqrt{a^2 + b^2} < 1$. Therefore, $\Theta$ is stable and $\{ x_t \}$ is stationary.

\paragraph{Uniform sampling from the unit disk.}
Each pair $(a,b)$ is sampled uniformly from the closed unit disk $\{z \in \C : \abs{z} < 1\}$ via the following procedure. We independently let $u \sim \mathrm{Uniform}(0, 1)$ and $\varphi \sim \mathrm{Uniform}(0, 2\pi)$ and then set
\begin{equation*}
	r = \sqrt{u}, \quad a = r\cos\varphi, \quad b = r\sin\varphi.
\end{equation*}
Since the area element in polar coordinates is $r \, dr \, d\varphi$, and the density of $r = \sqrt{u}$ is $f_r(r) = 2r$ for $r \in (0, 1)$, the joint density $f(r, \varphi) = 2r / (2\pi) = r/\pi$ is the uniform density on the unit disk (whose area is $\pi$). This ensures $a^2 + b^2 = r^2 = u < 1$ almost surely.

\paragraph{Noise covariance generation.}
Unless specified otherwise, the noise covariance matrix is generated as follows. Draw $S \in \R^{\dimension \times \dimension}$ with entries i.i.d. from $\mathrm{Uniform}(-1,1)$, and set $\Sigma = \transp{S} S$. This procedure yields a random semi-definite matrix almost surely. 

\subsubsection{Data set descriptions}
\label{app:experiments-synthetic-datasets}

\paragraph{\phantom{}}

\begin{datasetbox}[Bivariate VAR with Gaussian noise]
	\begin{datasetitems}
		\item \textbf{Purpose.} 
		Baseline evaluation on the standard $\varprocess{2}{1}$ setting with Gaussian noise.
		
		\item \textbf{Generative process.}
		For each replication, draw $(a_0, b_0)$ and $(a_1, b_1)$ independently and uniformly from the unit disk (Appendix~\ref{app:experiments-synthetic-transition}), and form the transition matrices
		\begin{equation*}
			\Theta_0 =
			\begin{pmatrix}
				a_0 & -b_0 \\
				b_0 & a_0
			\end{pmatrix},
			\qquad
			\Theta_1 =
			\begin{pmatrix}
				a_1 & -b_1 \\
				b_1 & a_1
			\end{pmatrix}.
		\end{equation*}
		Draw the covariance matrix $\Sigma$ as described in Appendix~\ref{app:experiments-synthetic-transition}. The noise process is independent and identically distributed multivariate Gaussian,
		\begin{equation*}
			\varepsilon_t \sim \mathcal{N}(0, \Sigma).
		\end{equation*}
		The two transition matrices $\Theta_0$ and $\Theta_1$ are sampled independently with no constraint on the distance between their eigenvalues; in particular, the change magnitude is unconstrained.
	\end{datasetitems}
\end{datasetbox}

\paragraph{\phantom{}}

\begin{datasetbox}[Bivariate VAR with Laplace noise]
	\begin{datasetitems}
		\item \textbf{Purpose.}
		Evaluate robustness to heavier-tailed, non-Gaussian noise. 
		
		\item \textbf{Generative process.}
		The transition matrices $\Theta_0$, $\Theta_1$, and the noise covariance $\Sigma$ are generated identically to the Gaussian case (Appendix~\ref{app:experiments-synthetic-transition}). The noise is drawn from a multivariate Laplace distribution with covariance $\Sigma$ via the following Gaussian copula construction:
		\begin{enumerate}
			\item Compute the correlation matrix $R$ as $R_{ij} = \Sigma_{ij} / \sqrt{\Sigma_{ii}\Sigma_{jj}}$.
			\item Draw $z \sim \mathcal{N}(0, R)$, a correlated standard normal vector.
			\item Transform to uniform marginals as $u = \Phi(z)$ element-wise.
			\item Apply the inverse Laplace $x_i = F_{\mathrm{Lap}}^{-1}(u_i \mid 0, s_i)$, with $s_i = \sqrt{\Sigma_{ii}}/\sqrt{2}$.
		\end{enumerate}
		The result $\varepsilon_t = X$ has marginal Laplace distributions with variance $\Sigma_{ii}$ and cross-correlation structure induced by $R$. Note that the marginal variance of a $\mathrm{Laplace}(0, s)$ distribution is $2 s^2 = \Sigma_{ii}$, so the covariance matrix of $\varepsilon_t$ is $\Sigma$ by construction.
	\end{datasetitems}
\end{datasetbox}

\paragraph{\phantom{}}

\begin{datasetbox}[Bivariate VAR with Student's $\boldsymbol t$ noise]
	\begin{datasetitems}
		\item \textbf{Purpose.} 
		Evaluate robustness across a range of noise tail weights, from heavy tails ($\nu=3$) to near-Gaussian ($\nu=30$).
		
		\item \textbf{Generative process.} 
		The transition matrices and noise covariances are generated as in the Gaussian and Laplace noise data sets (Appendix~\ref{app:experiments-synthetic-transition}). The noise at each time step is drawn from a multivariate Student's $t$ distribution
		\begin{equation*}
			\varepsilon_t \sim t_\nu(0, \Sigma),
		\end{equation*}
		generated via a scale-mixture of Gaussians: draw $y \sim \mathcal{N}(0, \Sigma)$ and $u \sim \chi_\nu^2$ independently, and set $\varepsilon_t = y \sqrt{\nu / u}$. Note that for $\nu > 2$, $\operatorname{Cov}(\varepsilon_t) = \frac{\nu}{\nu-2}\Sigma$.
		
		\item \textbf{Binning.} 
		The $N=1{,}000$ replications are divided equally into 10 bins of 100 samples each, corresponding to degrees of freedom 
		\begin{equation*}
			\nu \in \{ 3, 4, 5, 6, 8, 10, 12, 15, 20, 30 \},
		\end{equation*}
		sorted from heavy-tailed to near-Gaussian. Within each bin, $\Theta_0$, $\Theta_1$, and $\Sigma$ are redrawn independently per replication. 
	\end{datasetitems}
\end{datasetbox}

\paragraph{\phantom{}}

\begin{datasetbox}[Bivariate VAR with Huber-contaminated noise]
	\begin{datasetitems}
		\item \textbf{Purpose.} 
		Evaluate robustness to outlier contamination. The noise follows Huber's $\epsilon$-contamination model: most observations are drawn from the nominal Gaussian, but with probability $\epsilon$ an outlier is drawn instead. Changepoint detection in contaminated or adversarial noise settings has attracted interest, including the algorithm proposed by \citep{li2021adversarially}.
		
		\item \textbf{Generative process.} 
		The transition matrices are generated as in Appendix~\ref{app:experiments-synthetic-transition}. At each time step, the noise is drawn from the mixture
		\begin{equation*}
			\varepsilon_t \sim (1-\epsilon) \, \mathcal{N}(0, I_2) + \epsilon \, \mathcal{N}(0, 9 I_2),
		\end{equation*}
		where $9 I_2 = \transp{(3 I_2)} 3 I_2$ corresponds to the contaminating observations with standard deviation three times the nominal. This is implemented as a Bernoulli indicator: $\varepsilon_t \sim \mathcal{N}(0, I_2)$ w.p. $1-\epsilon$ and $\varepsilon_t \sim \mathcal{N}(0, 9 I_2)$ w.p. $\epsilon$.
		
		\item \textbf{Binning.} 
		The 10 contamination levels are 
		\begin{equation*}
			\epsilon \in \{0, 0.01, 0.02, 0.05, 0.10, 0.15, 0.20, 0.25, 0.30, 0.40\},
		\end{equation*}
		each applied to $N/10=100$ replications. At $\epsilon = 0$, the model reduces to the standard Gaussian case.
	\end{datasetitems}
\end{datasetbox}

\paragraph{\phantom{}}

\begin{datasetbox}[High-dimensional VAR with sparse dynamics]
	\begin{datasetitems}
		\item \textbf{Purpose.} 
		Evaluate scalability and robustness to ambient noise dimensions. A $\varprocess{\dimension}{1}$ process is generated by embedding a bivariate dynamical system into a $\dimension$-dimensional observation space, where the remaining $\dimension-2$ components evolve as pure noise with no cross-interaction. This simulates a realistic setting in which a low-dimensional signal is embedded in a high-dimensional noisy observation space.
		
		\item \textbf{Generative process.} 
		A random permutation $\pi$ of $\{1, \ldots, \dimension\}$ is drawn, and two coordinates $i = \pi(1)$, $j = \pi(2)$ are selected. The $\dimension \times \dimension$ transition matrix is
		\begin{equation*}
			(\Theta)_{ab} = \begin{cases} a_{\mathrm{val}} & a = b \in \{i, j\}, \\ -b_{\mathrm{val}} & a = i, b = j, \\ b_{\mathrm{val}} & a = j, b = i, \\ 0 & \text{otherwise,} \end{cases}
		\end{equation*}
		where $(a_{\mathrm{val}}, b_{\mathrm{val}})$ parametrises the $2 \times 2$ rotation block. The $\dimension$-dimensional noise is $\varepsilon_t \sim \mathcal{N}(0, \Sigma)$ where $\Sigma \in \R^{\dimension \times \dimension}$ is generated as $S^\top S$ with $S \sim \mathrm{Uniform}(-1, 1)^{\dimension \times \dimension}$.
		
		\item \textbf{Minimum change size.} 
		To ensure the embedded dynamics change is non-trivial and detectable, the pair of eigenvalues $(z_0, z_1)$ for $(\Theta_0, \Theta_1)$ are drawn with a minimum Euclidean distance constraint:
		\begin{equation*}
			\abs{z_1 - z_0} \geq d_{\mathrm{low}},
		\end{equation*}
		where $d_{\mathrm{low}}$ is the 90th percentile of the Euclidean distance $\abs{z_1 - z_0}$ between two points drawn independently and uniformly from the unit disk. This 90th percentile is estimated with reproducibility from $10^6$ Monte Carlo samples. Note that here the Euclidean distance $\abs{z_1 - z_0}$ is used directly (not symmetrised as in the controlled distance case), since both points are embedded at the same positions $i, j$.
		
		\item \textbf{Noise.}
		Gaussian $\varepsilon_t \sim \mathcal{N}(0, \Sigma)$ where $\Sigma = \transp{S} S$ with $S \sim \mathrm{Uniform}(-1,1)^{\dimension \times \dimension}$, independently drawn per replication.
		
		\item \textbf{Binning.} 
		Dimension is varied within 10 bins as 
		\begin{equation*}
			\dimension \in \{ 2, 4, 6, 10, 15, 20, 25, 30, 35, 40 \},
		\end{equation*}
		each receiving $N=100$ replications.
	\end{datasetitems}
\end{datasetbox}

\paragraph{\phantom{}}

\begin{datasetbox}[High-dimensional VAR with full-rank transition matrix]
	\begin{datasetitems}
		\item \textbf{Purpose.} 
		Evaluate performance on high-dimensional VAR processes where all dimensions participate in the dynamics, without sparsity assumptions. This contrasts with the previous data set by removing the sparse embedding structure and generating dense, non-diagonally dominant transition matrices. 
		
		\item \textbf{Eigenvalue sampling.}
		To ensure conjugate closure (required for the transition matrix to be real-valued), eigenvalues are drawn as follows. First, sample $n_{\mathrm{pairs}} \sim \mathrm{Uniform}(0, 1, \dots, \lfloor\dimension/2\rfloor)$, giving $n_{\mathrm{real}} = \dimension - 2 n_{\mathrm{pairs}}$. Then,
		\begin{itemize}
			\item For each real eigenvalue, draw $(a, \cdot)$ from the unit disk; add $\lambda = a \in \R$.
			\item For each complex conjugate pair, draw $(a, b)$ from the unit disk with $b>0$; add the pair $\lambda = a \pm bi \in \C$.
		\end{itemize}
		
		\item \textbf{Real canonical form.}
		Construct the block-diagonal matrix $D_c \in \R^{\dimension \times \dimension}$ as
		\begin{equation*}
			D_c = \mathrm{blockdiag}\!\left(\lambda_1, \ldots, \lambda_{n_{\mathrm{real}}}, \begin{pmatrix} a_1 & -b_1 \\ b_1 & a_1 \end{pmatrix}, \ldots, \begin{pmatrix} a_{n_{\mathrm{pairs}}} & -b_{n_{\mathrm{pairs}}} \\ b_{n_{\mathrm{pairs}}} & a_{n_{\mathrm{pairs}}} \end{pmatrix}\right).
		\end{equation*}
		
		\item \textbf{Similarity transform.}
		A random invertible matrix $P \in \R^{\dimension \times \dimension}$ with bounded condition number is constructed as follows. Draw $G \in \R^{\dimension \times \dimension}$ with i.i.d. standard Gaussian entries and compute the singular value decomposition $G = U \diag(s) \transp{V}$. Clip the singular values as
		\begin{equation*}
			s_{\mathrm{clipped}, i} = \max\!\left(s_i, \frac{s_1}{\kappa_{\max}}\right), \qquad \kappa_{\max} = 15,
		\end{equation*}
		and set $P = U \diag(s_{\mathrm{clipped}}) \transp{V}$. This preserves the random orientation of $U$ and $V$ (uniformly distributed on the orthogonal group) while bounding the 2-norm condition number $\kappa(P) \leq \kappa_{\max}$, controlling the spread of entries in the resulting transition matrix.
		
		\item \textbf{Transition matrix.}
		The final transition matrix is given by
		\begin{equation*}
			\Theta = P D_c P^{-1},
		\end{equation*}
		which as the prescribed spectrum $\sigma(\Theta) = \sigma(D_c)$ and is real-valued by construction.
		
		\item \textbf{Noise.}
		Gaussian $\varepsilon_t \sim \mathcal{N}(0, \Sigma)$ where $\Sigma = \transp{S} S$ with $S \sim \mathrm{Uniform}(-1,1)^{\dimension \times \dimension}$, independently drawn per replication.
		
		\item \textbf{Binning.} 
		Dimension is varied within 10 bins as 
		\begin{equation*}
			\dimension \in \{ 2, 4, 6, 10, 15, 20, 25, 30, 35, 40 \},
		\end{equation*}
		each receiving $N=100$ replications.
	\end{datasetitems}
\end{datasetbox}

\subsubsection{Computational complexity}
\label{app:experiments-synthetic-complexity}

\begin{figure}[t]
	\centering
	\includegraphics[width=\linewidth]{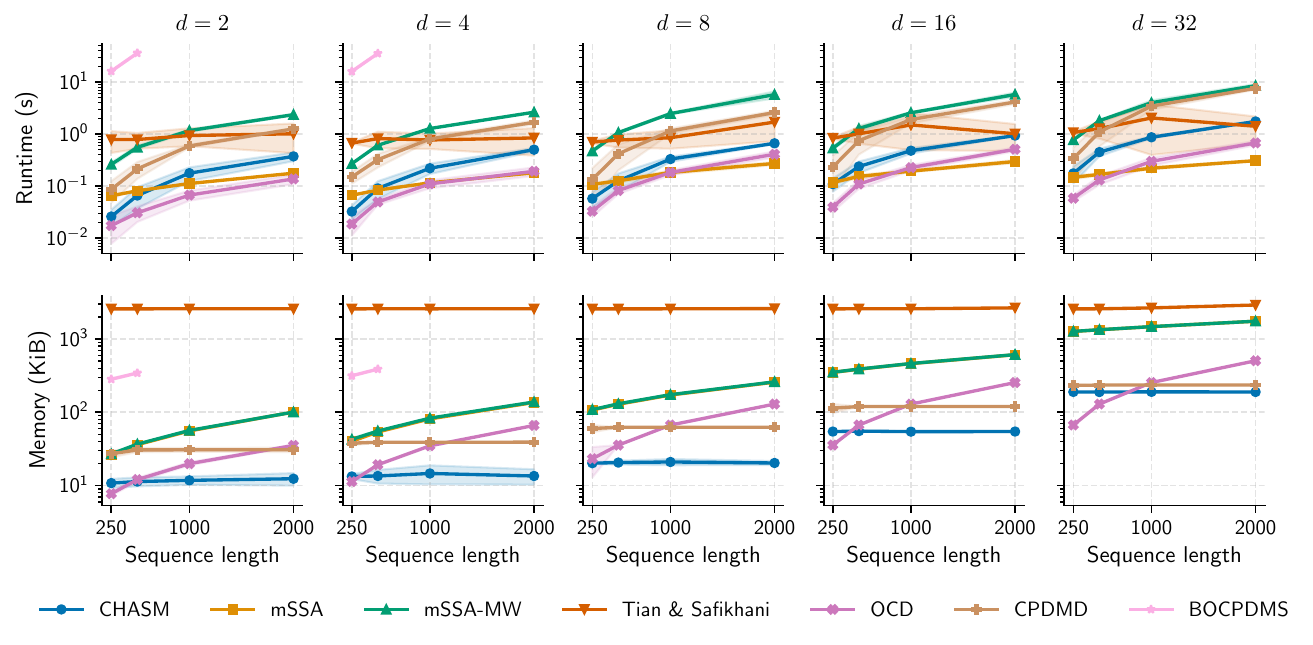}
	\caption{\small Computational scaling of online changepoint detection methods across sequence length and dimension. Runtime (top) and peak memory usage (bottom) are reported for sequence lengths $T \in \{ 250, 500, 1000, 2000 \}$ and dimensions $\dimension \in \{2, 4, 8, 16, 32\}$, on $\varprocess{\dimension}{1}$ process $x_t = \Theta \, x_{t-1} + \varepsilon_t$, $\varepsilon_t \sim \mathcal{N}(0, I_\dimension)$, with transition $\Theta_{ij} \overset{\text{i.i.d.}}{\sim} \mathcal{N}(0, 0.09/\dimension)$. Lines denote the mean and shaded bands denote $\pm 1$ standard deviation estimated over 10 Monte Carlo replications, which is sufficiently large given the relatively low variability of the results. All methods are run with default parameters, as specified in the respective publications. Memory is measured via \texttt{tracemalloc}; note that OCD and Tian \& Safikhani are called via \texttt{rpy2} R bindings, their reported memory figures reflect only the Python-side allocation and are likely underestimates of true memory consumption. BOCPDMS is evaluated only for $\dimension \in \{2, 4\}$ and $T \in \{ 250, 500 \}$ due to its substantially higher computational cost; truncated results are sufficient to establish it as the most expensive method.}
	\label{fig:complexity}
\end{figure}

Figure~\ref{fig:complexity} illustrates two distinct behaviours with respect to memory. Methods whose memory footprint is flat in $T$ (CHASM, CPDMD, and Tian \& Safikhani) are truly streaming, retaining only a fixed-size summary of the data regardless of sequence length. In contrast, mSSA, mSSA-MW, OCD, and BOCPDMS exhibit memory that grows with $T$, reflecting window buffers or accumulated history inherent to their implementations. CHASM achieves the lowest overall memory consumption across nearly all $(\dimension, T)$ configurations, and its flat profile is consistent with the $\mathcal{O}(\dimension^2)$ state maintained by the online DMD update: no observations are stored beyond the current sufficient statistics.

In terms of runtime, BOCPDMS is the most computationally demanding method by a large margin, already exhibiting the highest latency at $\dimension=2$ and $T=250$. CHASM is consistently among the fastest methods across all dimensions, with only mSSA and OCD achieving lower or comparable runtimes in some regimes. Notably, CHASM runtime scaling closely mirrors that of OCD, which has reported complexity $\mathcal{O}(\dimension^2 \log\dimension)$; this is consistent with CHASM's theoretical complexity of $\mathcal{O}(\dimension^2 k \rank)$, where $k$ is the number of Arnoldi iterations, placing it below the cubic cost of a full eigendecomposition. This is particularly notable given that the dense transition matrix $\Theta$ affords no low-rank structure to exploit, making this a worst-case scenario for CHASM. Taken together, CHASM dominates mSSA-MW and CPDMD in both runtime and memory across all considered settings, suggests a favourable trade-off compared to Tian \& Safikhani, and scales comparably to OCD in time while using lower memory.

\subsection{Real-world data sets}
\label{app:experiments-real}

\subsubsection{Overview}
\label{app:experiments-real-overview}

Real-world data sets are used to evaluate CHASM on naturally occurring changepoints across diverse modalities (motion, image, text, video). For each data set, the time series $\{x_t\}$ is a sequence of feature vectors (pre-trained model embeddings or hand-crafted features) at consecutive time steps, and the ground-truth changepoints are derived from provided annotations. All data sets are used in online changepoint detection protocol supporting multiple sequential changepoints. 

\subsubsection{Data set descriptions}
\label{app:experiments-real-datasets}

\paragraph{\phantom{}}

\begin{datasetbox}[HASC]
	\begin{datasetitems}
		\item \textbf{Source.} 
		The HASC \citep{kawaguchi2011hasc} data set provides three-axis accelerometer recordings from wearable devices, collected from seven subjects performing different labelled activities (stay, walk, jogging, skip, stair-up, and stair-down). Each recording corresponds to 120 seconds of measurements, sampled at 100 Hz.
		
		\item \textbf{Preprocessing.} 
		All $N=18$ recordings contain timestamped three-axis accelerometer readings, with corresponding annotations specifying time intervals of labelled activity. A binary label sequence is constructed: observations within a labelled interval receive label 0 (active), and all other observations receive label 1 (unlabelled). Changepoints are identified as time locations where consecutive binary labels differ.
		
		\item \textbf{Changepoint extraction.} 
		The following procedure is applied to the detected transitions:
		\begin{enumerate}
			\item The first and last transitions are discarded (start/end boundary artefacts).
			\item The other transitions come in pairs (labelled $\to$ unlabelled and unlabelled $\to$ labelled); the first element of each pair is retained, corresponding to the onset of each labelled activity region.
		\end{enumerate}
		
		\item \textbf{Output.} 
		Each recording yields one time series of dimension $\dimension=3$ and length $T_i$, where $T_i$ varies by recording, along with a list of changepoint indices. We compute accuracy metrics using $\Delta_\ell = 200$ and $\Delta_r = 200$, corresponding to 2 seconds of data each (Appendix~\ref{app:experiments-metrics-accuracy}).
	\end{datasetitems}
\end{datasetbox}

\paragraph{\phantom{}}

\begin{datasetbox}[CIFAR-100]
	\begin{datasetitems}
		\item \textbf{Source.} 
		CIFAR-100 \citep{krizhevsky2009learning} consists of $60{,}000$ colour images ($32 \times 32 \times 3$ pixels) across $100$ fine-grained classes grouped into $20$ coarse categories.
		
		\item \textbf{Embedding.} 
		Each image is encoded with CLIP ViT-B/32\footnote{Available on Hugging Face via \href{https://huggingface.co/openai/clip-vit-base-patch32}{\texttt{openai/clip-vit-base-patch32}}.}, a vision-language model trained on a contrastive learning objective that jointly learns an image encoder and a text encoder, denoted $f_{\mathrm{image}}$ and $f_{\mathrm{text}}$, respectively. Both encoders produce $512$-dimensional vector representations in a shared latent space \citep{radford2021learning}. The embedding of an image $I$ is defined as the $\ell^2$-normalised output of the visual encoder
		\begin{equation*}
			\varphi(I) = \frac{f_{\mathrm{image}}(I)}{\norm{f_{\mathrm{image}}(I)}} \in \R^{512}.
		\end{equation*}
		
		\item \textbf{PCA projection.}
		Prior to dimensionality reduction, the embeddings are standardised feature-wise to have zero mean and unit variance. The standardisation parameters and the principal component basis with target dimension $\dimension=32$ are estimated exclusively from the embeddings of the test split ($10{,}000$ images). The embeddings from the training split are then transformed using these fixed parameters and projected onto the first principal components via
		\begin{equation*}
			\varphi_{\mathrm{pca}}(I) = \transp{W_{\mathrm{image}}} \frac{\varphi(I) - \mu}{\sigma} \in \R^{32},
		\end{equation*}
		where $\mu \in \R^{512}$ and $\sigma \in \R^{512}$ denote the feature-wise mean and standard deviation estimated from the test split, and $W_{\mathrm{image}} \in \R^{512 \times 32}$ contains the leading $32$ principal component directions.
		
		This procedure ensures that the dimensionality reduction mapping is estimated on a held-out split that is entirely disjoint from the sequences used in the experiment, thereby preventing information leakage into the time series construction.
		
		\item \textbf{Time series construction.} 
		Each of the $N=100$ samples is a time series of length $T=1{,}000$ observations with $K=4$ changepoints, yielding $K+1=5$ segments. For each sample:
		\begin{enumerate}
			\item Draw $K+1=5$ distinct coarse classes uniformly without replacement.
			\item For each chosen superclass, randomly select one of its 5 fine classes.
			\item Compute nominal changepoint positions $\tau_k^{\mathrm{nom}} = k \, T / (K + 1)$ for $k = 1, \ldots, K$, and add i.i.d. jitter $j_k \sim \mathrm{Uniform}\{-30, \ldots, +30\}$: $\tau_k = \tau_k^{\mathrm{nom}} + j_k$.
			\item Segment lengths are $\ell_k = \tau_{k+1} - \tau_k$ (with $\tau_0 = 0$, $\tau_{K+1} = T$).
			\item Each segment is filled by sampling $\ell_k$ images without replacement from the train-split images of the selected fine class ($50{,}000$ images).
		\end{enumerate}
		All sequences are drawn from the train split; the test split is used only for PCA fitting. A sample time series from the created data set is shown in Figure~\ref{fig:cifar-sample}. We compute accuracy metrics using $\Delta_\ell = 0$ and $\Delta_r = 50$.
		
		\item \textbf{Motivation.}
		This construction yields piecewise i.i.d. segments, where each segment corresponds to a fixed image class. Although the resulting changes do not reflect changepoints in dynamics but rather shifts in the distribution of high-dimensional observations, this setting is relevant in practice. In remote sensing and satellite imagery, consecutive images of the same scene are observed under varying conditions such as sensor noise, illumination, or weather, and changepoints correspond to abrupt scene changes such as floods, wildfires, or new constructions \cite{borsoi2021online, liu2019review}. A similar formulation arises in visual quality control, where repeated observations of manufactured objects are monitored to detect defects or production faults from visual appearance alone \cite{fridman2021changechip, megahed2012spatiotemporal}.
	\end{datasetitems}
\end{datasetbox}

\begin{figure}[t]
	\centering
	\includegraphics[width=\linewidth]{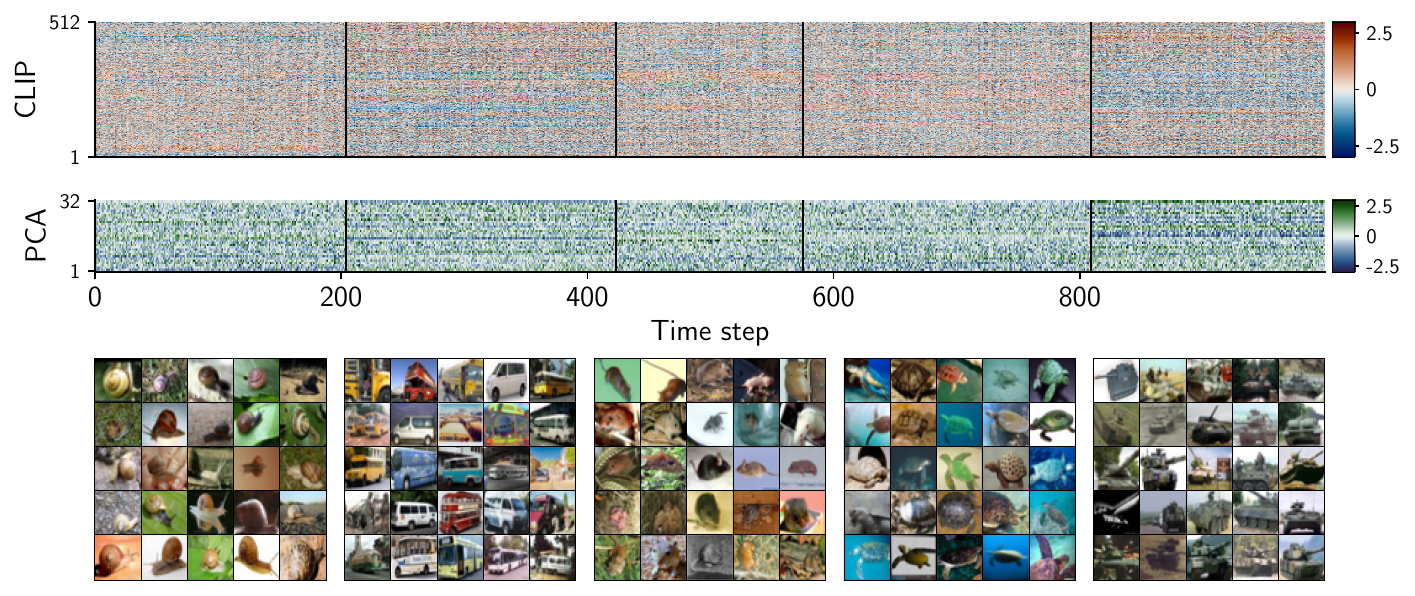}
	\caption{\small Sample time series from the CIFAR-100 benchmark dataset. \emph{Top:} Heatmap of the raw $512$-dimensional CLIP ViT-B/32 embeddings $\varphi(I_t) \in \R^{512}$ across $T = 1{,}000$ time steps. \emph{Middle:} Heatmap of the corresponding PCA-projected embeddings $\varphi_{\mathrm{pca}}(I_t) \in \R^{32}$, obtained by standardising the CLIP embeddings and projecting onto the leading $32$ principal components estimated from the held-out test split. \emph{Bottom:} Representative image grids for each of the $K+1 = 5$ segments, corresponding to the fine classes \texttt{snail}, \texttt{bus}, \texttt{mouse}, \texttt{turtle}, and \texttt{tank}, drawn sequentially from five distinct CIFAR-100 coarse categories. Vertical black lines mark the four changepoint locations $\tau_1, \ldots, \tau_4$, placed at approximately equal spacing with $\pm 30$-step uniform jitter. Distribution shifts between segments are clearly visible in both embedding spaces, particularly in the PCA projection where the per-segment colour profile changes abruptly at each changepoint.}
	\label{fig:cifar-sample}
\end{figure}

\paragraph{\phantom{}}

\begin{datasetbox}[20 Newsgroups]
	\begin{datasetitems}
		\item \textbf{Source.} 
		The 20 Newsgroups corpus \citep{lang1995newsweeder} consists of approximately $18{,}000$ Usenet posts distributed across $20$ thematic newsgroups. These groups are commonly organised into $7$ broader topical categories, including computer technology, recreation, science, politics, religion, and miscellaneous discussions.
		
		\item \textbf{Embedding.} 
		Each document is encoded using the text encoder of CLIP ViT-B/32\footnote{Available on Hugging Face via \href{https://huggingface.co/openai/clip-vit-base-patch32}{\texttt{openai/clip-vit-base-patch32}}.}, a vision-language model trained with a contrastive learning objective that jointly learns image and text representations in a shared latent space \citep{radford2021learning}. Prior to encoding, document headers, footers, and quoted reply text are removed, retaining only the main body content of each post. The embedding of a document $D$ is defined as the $\ell^2$-normalised output of the text encoder
		\begin{equation*}
			\varphi(t) = \frac{f_{\mathrm{text}}(D)}{\norm{f_{\mathrm{text}}(D)}} \in \R^{512}.
		\end{equation*}
		
		\item \textbf{PCA projection.}
		Prior to dimensionality reduction, the embeddings are standardised feature-wise to have zero mean and unit variance. The standardisation parameters and the principal component basis with target dimension $\dimension=32$ are estimated exclusively from the embeddings of the test split. The embeddings from the training split are then transformed using these fixed parameters and projected onto the first principal components via
		\begin{equation*}
			\varphi_{\mathrm{pca}}(D) = \transp{W_{\mathrm{text}}} \frac{\varphi(D) - \mu}{\sigma} \in \R^{32},
		\end{equation*}
		where $\mu \in \R^{512}$ and $\sigma \in \R^{512}$ denote the feature-wise mean and standard deviation estimated from the test split, and $W_{\mathrm{text}} \in \R^{512 \times 32}$ contains the leading $32$ principal component directions.
		
		This procedure ensures that the dimensionality reduction mapping is estimated on a held-out split that is entirely disjoint from the sequences used in the experiment, thereby preventing information leakage into the time series construction.
		
		\item \textbf{Time series construction.} 
		Each of the $N=100$ samples is a time series of length $T=1{,}000$ observations with $K=4$ changepoints, constructed as in the CIFAR-100 data set described above. Coarse categories are drawn uniformly without replacement from the set of topical categories. For each selected category, one constituent newsgroup is chosen uniformly at random. Each segment is filled by sampling documents without replacement from the training split of the selected newsgroup. All sequences are drawn from the training split; the test split is used only for PCA fitting. A sample time series from the created data set is shown in Figure~\ref{fig:newsgroups-sample}. We compute accuracy metrics using $\Delta_\ell = 0$ and $\Delta_r = 50$ (Appendix~\ref{app:experiments-metrics-accuracy}).
		
		\item \textbf{Motivation.}
		Text streams such as news articles, social media posts, and customer communications often exhibit abrupt changes in topic or semantic content, and the 20 Newsgroups data set has been used to model this setting \citep{garcia2025concept}. In our construction, each segment corresponds to a fixed topic, and documents within a segment are sampled independently from the same class, yielding a piecewise i.i.d. structure in the embedding space. Although this does not correspond to changes in temporal dynamics, it captures distributional shifts that are central in applications such as news and media monitoring, online product reviews, and customer support systems, where changepoints correspond to topic or sentiment changes in text streams \citep{jia2026unsupervised, wang2018real}.
	\end{datasetitems}
\end{datasetbox}

\begin{figure}[t]
	\centering
	\includegraphics[width=\linewidth]{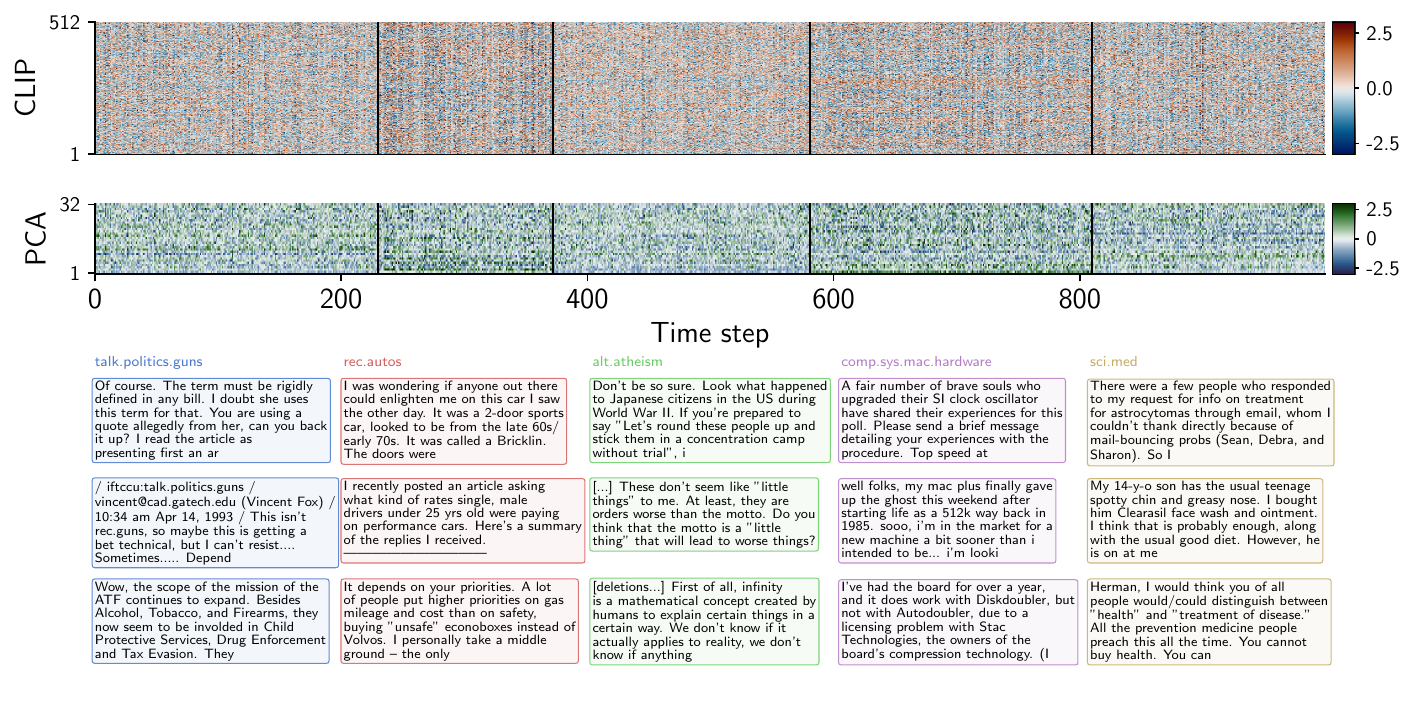}
	\caption{\small Sample time series from the 20 Newsgroups benchmark dataset. \emph{Top:} Heatmap of the raw $512$-dimensional CLIP ViT-B/32 embeddings $\varphi(D_t) \in \R^{512}$ across $T = 1{,}000$ time steps. \emph{Middle:} Heatmap of the corresponding PCA-projected embeddings $\varphi_{\mathrm{pca}}(D_t) \in \R^{32}$, obtained by standardising the CLIP embeddings and projecting onto the leading $32$ principal components estimated from the held-out test split. \emph{Bottom:} Representative image grids for each of the $K+1 = 5$ segments, corresponding to the fine classes \texttt{guns}, \texttt{autos}, \texttt{atheism}, \texttt{hardware}, and \texttt{med}, drawn sequentially from five distinct 20 Newsgroups coarse categories. Vertical black lines mark the four changepoint locations $\tau_1, \ldots, \tau_4$, placed at approximately equal spacing with $\pm 30$-step uniform jitter. Distribution shifts between segments are clearly visible in both embedding spaces, particularly in the PCA projection where the per-segment colour profile changes abruptly at each changepoint.}
	\label{fig:newsgroups-sample}
\end{figure}

\paragraph{\phantom{}}

\begin{datasetbox}[UCF-Crime]
	\begin{datasetitems}
		\item \textbf{Source.}
		The UCF-Crime data set \citep{sultani2018real} consists of real-world surveillance videos with frame-level temporal annotations of anomalous events. We focus on the \emph{Explosion} category, where annotations specify the start and end of each explosion interval. Because a video can contain up to two explosion events, each sample may contain up to four annotated changepoints.
		
		\item \textbf{Embedding.}
		We represent each video as $V = \begin{bmatrix} F_1 & \cdots & F_T \end{bmatrix} \in \R^{H \times W \times C \times T}$, with height $H=240$, width $W=320$, channels $C=3$, and $T$ denotes the number of frames. Frames are decoded at the native frame rate of 30 frames per second, without downsampling, and processed independently with CLIP ViT-B/32\footnote{Available on Hugging Face via \href{https://huggingface.co/openai/clip-vit-base-patch32}{\texttt{openai/clip-vit-base-patch32}}.}, a vision-language model trained with a contrastive objective \citep{radford2021learning}. For each frame $F_t \in \R^{H \times W \times C}$, we compute the $\ell^2$-normalised image embedding
		\begin{equation*}
			\varphi(F_t) = \frac{f_{\mathrm{image}}(F_t)}{\norm{f_{\mathrm{image}}(F_t)}} \in \R^{512}.
		\end{equation*}
		The resulting video embedding is the temporally ordered sequence
		\begin{equation*}
			\varphi(V)= \begin{bmatrix}
				\varphi(F_1) & \cdots & \varphi(F_T)
			\end{bmatrix} \in \R^{512 \times T}.
		\end{equation*}
		This representation uses no explicit video model; temporal structure is left entirely to the downstream changepoint detection algorithm.
		
		\item \textbf{PCA projection.}
		Because UCF-Crime does not provide a designated split for fitting PCA independently, we instead transfer the projection learned in the CIFAR-100 experiment. For each frame embedding $\varphi(F_t) \in \R^{512}$, we compute
		\begin{equation*}
			\varphi_{\mathrm{pca}}(F_t) = \transp{W_{\mathrm{image}}} \frac{\varphi(F_t)-\mu}{\sigma} \in \R^{32},
		\end{equation*}
		where $\mu \in \R^{512}$, $\sigma \in \R^{512}$, and $W_{\mathrm{image}} \in \R^{512 \times 32}$ are the feature-wise mean, standard deviation, and leading principal directions estimated from the CIFAR-100 test embeddings. The resulting PCA-reduced video representation is
		\begin{equation*}
			\varphi_{\mathrm{pca}}(V) = \begin{bmatrix} \varphi_{\mathrm{pca}}(F_1) & \cdots & \varphi_{\mathrm{pca}}(F_T) \end{bmatrix} \in \R^{32 \times T}.
		\end{equation*}
		Since both data sets are encoded with the same CLIP ViT-B/32 model, this transfer is well-defined in the shared $512$-dimensional embedding space. From the UCF-Crime perspective, it defines a fixed linear map that is independent of UCF-Crime data. While not optimal under distribution shift, we find it preferable to a random projection since it remains aligned with the geometry of CLIP-based image features.
		
		\item \textbf{Ground truth.}
		For each available video, ground-truth changepoints are extracted from the annotated frame indices. We exclude videos whose first annotated changepoint occurs within the first $100$ frames, since they do not provide sufficient pre-change context. After preprocessing, this results in a collection of $N=18$ time series. A sample time series from the created data set is shown in Figure~\ref{fig:ucf-sample}. We compute accuracy metrics using $\Delta_\ell = 100$ and $\Delta_r = 100$, corresponding to $\approx$ 1 second of data each (Appendix~\ref{app:experiments-metrics-accuracy})
	\end{datasetitems}
\end{datasetbox}

\begin{figure}[t]
	\centering
	\includegraphics[width=\linewidth]{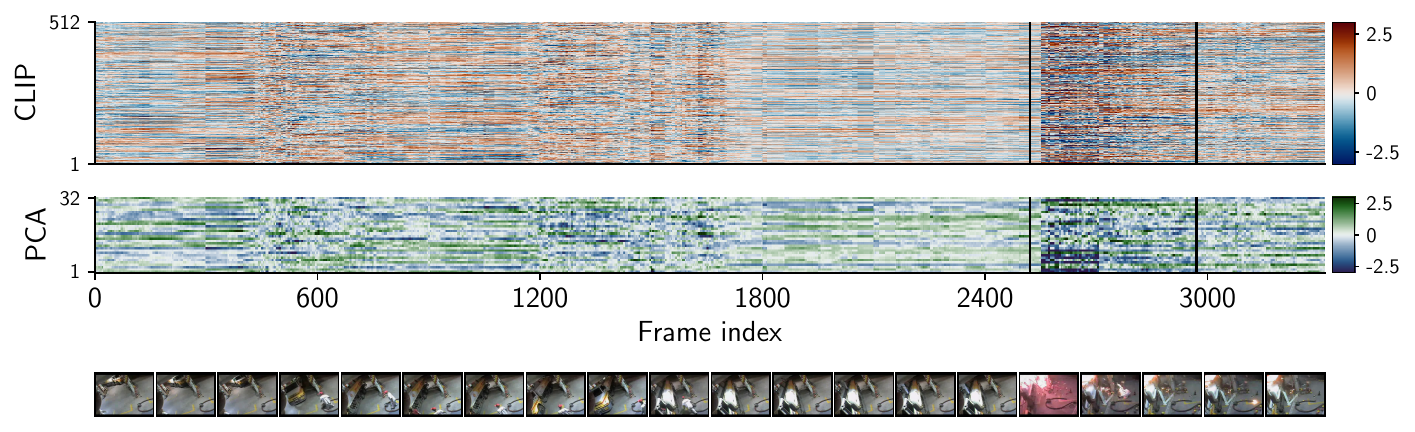}
	\caption{\small Sample time series from the UCF-Crime benchmark dataset (Explosion category). \emph{Top:} Heatmap of the raw $512$-dimensional CLIP ViT-B/32 frame embeddings $\varphi(F_t) \in \R^{512}$ across $T \approx 3{,}200$ frames, decoded at 30\,fps. \emph{Middle:} Heatmap of the corresponding PCA-projected embeddings $\varphi_{\mathrm{pca}}(F_t) \in \R^{32}$, obtained by applying the fixed linear projection transferred from the CIFAR-100 experiment. \emph{Bottom:} A strip of representative frames sampled uniformly across the video, illustrating the visual transition from a normal scene to an explosion event and its aftermath. Vertical black lines mark the annotated changepoint locations delimiting the explosion interval. Distributional shifts here arise from genuine changes in scene dynamics and visual content within a continuous video stream, and are clearly visible as abrupt colour discontinuities in both embedding spaces near the annotated changepoints.}
	\label{fig:ucf-sample}
\end{figure}

\paragraph{\phantom{}}

\begin{datasetbox}[WikiSection]
	\begin{datasetitems}
		\item \textbf{Source.}
		WikiSection \citep{arnold2019sector} provides Wikipedia articles annotated with section boundaries, where each section transition defines a topic changepoint. We use the \emph{English city} subset from the test split.
		
		\item \textbf{Filtering.}
		We retain only articles satisfying the following criteria: the number of sections lies in $[4,8]$, in order to control the number of changepoints per sequence; the total number of tokens lies in $[1000,4000]$, in order to control sequence length; no section label is repeated within the same article; and the first section boundary occurs after at least $200$ tokens ($\approx$ 150 words), to ensure sufficient pre-change context.
		
		\item \textbf{Embedding.}
		Let $x_{1:T}$ denote the tokenised document, where $T$ is the number of tokens. Each document is encoded with Qwen3-Embedding-4B\footnote{Available via \href{https://huggingface.co/Qwen/Qwen3-Embedding-4B}{\texttt{Qwen/Qwen3-Embedding-4B}}.}, a causal language model with embedding dimension $2560$ \citep{zhang2025qwen3}. The model is applied in a single causal forward pass to the full token sequence, without special tokens. For each token position $t \in \{1,\dots,T\}$, we extract the last hidden state and define the $\ell^2$-normalised token embedding
		\begin{equation*}
			h_t = \operatorname{LastHiddenState}(x_{1:t}) \in \R^{2560}, \qquad \varphi(x_t) = \frac{h_t}{\norm{h_t}} \in \R^{2560}.
		\end{equation*}
		Stacking these embeddings in sequential order yields the document representation
		\begin{equation*}
			\varphi(x_{1:T}) =
			\begin{bmatrix}
				\varphi(x_1) & \cdots & \varphi(x_T)
			\end{bmatrix} \in \R^{2560 \times T}.
		\end{equation*}
		
		\item \textbf{MRL truncation.}
		Qwen3-Embedding-4B is trained with the Matryoshka Representation Learning objective \citep{kusupati2022matryoshka}, which makes the leading coordinates of the embedding space informative at multiple target dimensions. We therefore truncate each token embedding to its first $\dimension=64$ coordinates and re-normalise:
		\begin{equation*}
			\varphi_{\mathrm{mrl}}(x_t) = \frac{\varphi(x_t)_{1:\dimension}}{\norm{\varphi(x_t)_{1:\dimension}}} \in \R^{64}.
		\end{equation*}
		The resulting truncated document representation is
		\begin{equation*}
			\varphi_{\mathrm{mrl}}(x_{1:T}) = \begin{bmatrix}
				\varphi_{\mathrm{mrl}}(x_1) & \cdots & \varphi_{\mathrm{mrl}}(x_T)
			\end{bmatrix} \in \R^{64 \times T}.
		\end{equation*}
		
		\item \textbf{Ground truth.}
		For each article, ground-truth changepoints are defined by the annotated section boundaries, at the token level. From the filtered candidate pool, we sample $N=50$ articles uniformly at random without replacement, using a fixed random seed for reproducibility. A sample time series from the created data set is shown in Figure~\ref{fig:wikisection-sample}. We compute accuracy metrics using $\Delta_\ell = 50$ and $\Delta_r = 50$ tokens, corresponding to $\approx$ 37 words (Appendix~\ref{app:experiments-metrics-accuracy}).
		
		\item \textbf{Novelty.}
		Text and topic segmentation are predominantly studied in an offline setting. For example, \citep{jia2026unsupervised} assume access to the complete document and perform batch inference to identify segment boundaries at the sentence level. In contrast, online changepoint detection methods for text data such as \citep{wang2018real} are comparatively less explored, and typically operate at the coarser document level, as in the 20 Newsgroups experiment. Online changepoint detection at finer granularities, such as the word or token level, remains largely unexplored. To the best of our knowledge, there is limited work on causal, real-time detection of topic boundaries directly from token streams. We argue that this formulation has several advantages: it enables fine-grained, online detection of changes directly at the token level, while remaining computationally efficient and thus allow applications such as authorship change detection, which can arguably be more challenging than document-level segmentation. In this experiment, we rely on autoregressive large language models, which generate tokens sequentially. We interpret the resulting embedding sequence as a dynamical system and detect changes in the effective dynamics of the token-generation process induced by the model’s hidden representations.
	\end{datasetitems}
\end{datasetbox}

\begin{figure}[t]
	\centering
	\includegraphics[width=\linewidth]{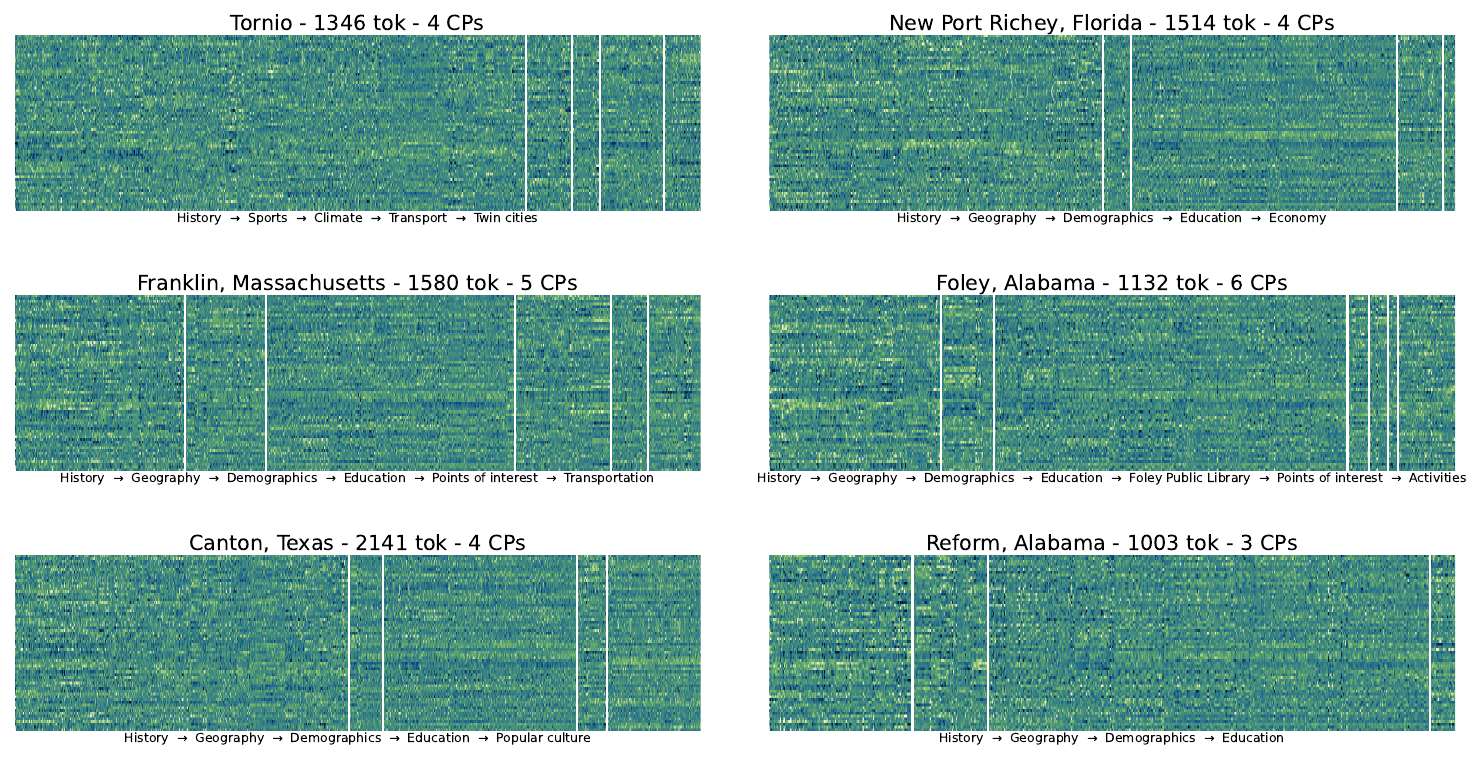}
	\caption{\small Six samples from the WikiSection benchmark dataset (English city subset). Each panel displays the heatmap of the $64$-dimensional MRL-truncated token embeddings $\varphi_{\mathrm{mrl}}(x_t) \in \R^{64}$ for a single Wikipedia article, with the article title, total token count, and number of changepoints indicated above, and the sequence of section labels shown below. Vertical white lines mark the ground-truth changepoint locations, defined by annotated section boundaries at the token level. Embeddings are produced by a single causal forward pass of Qwen3-Embedding-4B, so each $\varphi_{\mathrm{mrl}}(x_t)$ summarises the left context $x_{1:t}$ up to and including token $t$. Subtle shifts in the embedding colour profile are visible across section boundaries.}
	\label{fig:wikisection-sample}
\end{figure}

\paragraph{Dimensionality reduction.}
For data sets whose ambient dimension exceeds a practical threshold, we reduce to $\dimension$ principal components via PCA prior to applying all methods. As discussed in Appendix~\ref{app:experiments-synthetic-complexity}, competing approaches can rely on matrix inversions or SVD, making high-dimensional inputs such as flattened image pixels or raw CLIP embeddings computationally prohibitive; this work does not specifically target high-dimensional regimes. Reducing dimension is standard practice in changepoint detection applied to high-dimensional data; for instance,  \citep{ryan2023detecting} apply pixel subset selection when processing image streams. Among the alternatives (subset selection, random projection, and offline fine-tuning of a smaller projection head) PCA requires no offline training, is computationally lightweight, and empirically generalises well across modalities. Crucially, the same projected sequence is fed to all competing methods, preserving a fair comparison.

\subsection{Additional results}
\label{app:experiments-results}

\subsubsection{Synthetic bivariate time series}

Figures~\ref{fig:additional-gaussian}, \ref{fig:additional-laplace}, \ref{fig:additional-student}, and \ref{fig:additional-huber} report accuracy and speed metrics across all parameter configurations for four noise distributions. CHASM ($\rho = 1$ and $\rho < 1$) consistently dominates competing methods on precision and $\fone$-score across all settings, with relatively concentrated distributions indicating robustness to parameter choice. Jointly, CHASM achieves among the smallest $\arlone$ and the largest $\arlzero$, reflecting the best overall trade-off between detection speed and false-alarm control.

\begin{figure}[p]
	\centering
	\includegraphics[width=0.8\linewidth]{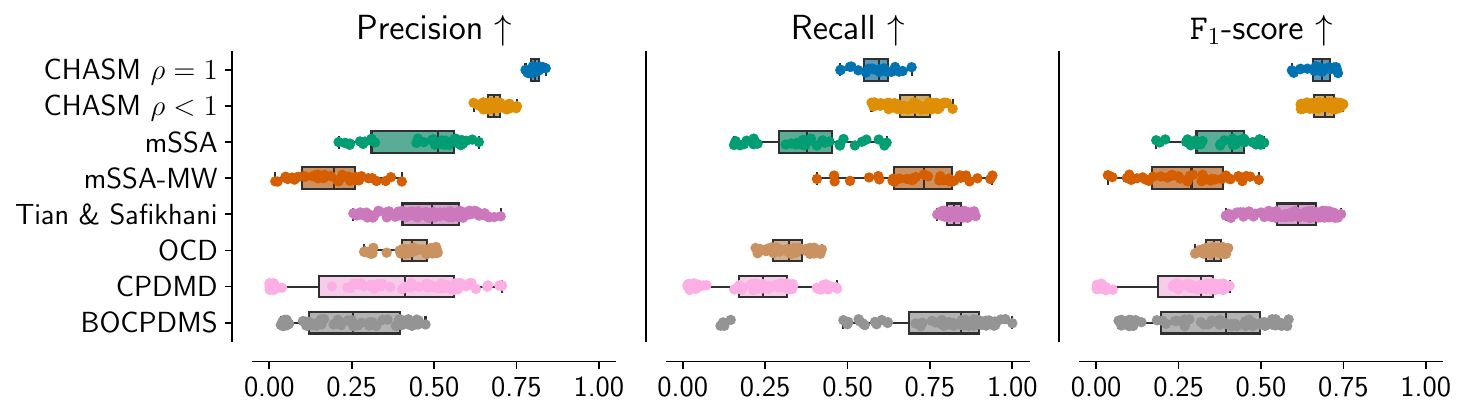}
	\includegraphics[width=0.6\linewidth]{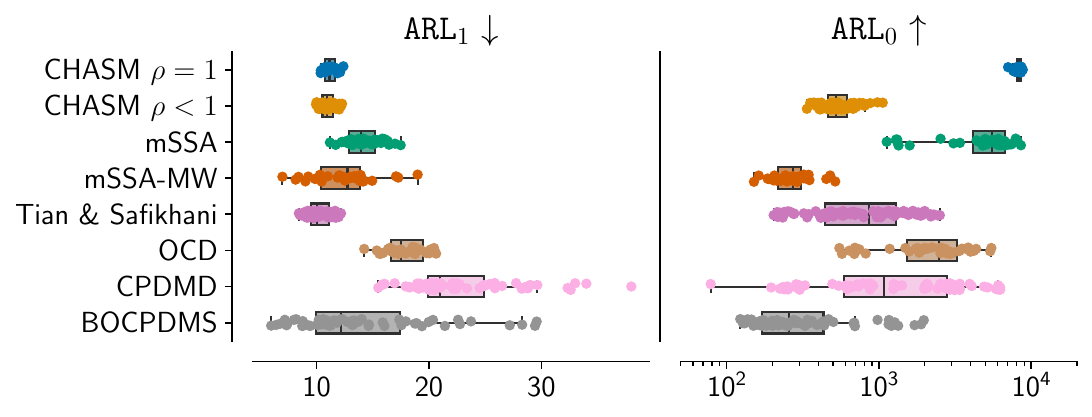}
	\caption{\small Accuracy (top) and speed (bottom) metrics (see Appendix~\ref{app:experiments-metrics}) across all parameter configurations on the synthetic bivariate $\varprocess{2}{1}$ data set with Gaussian noise (see Appendix~\ref{app:experiments-synthetic-datasets}). Each dot corresponds to a single parameter set evaluated over the grid (see Table~\ref{tab:parameters}).}
	\label{fig:additional-gaussian}
\end{figure}

\begin{figure}[p]
	\centering
	\includegraphics[width=0.8\linewidth]{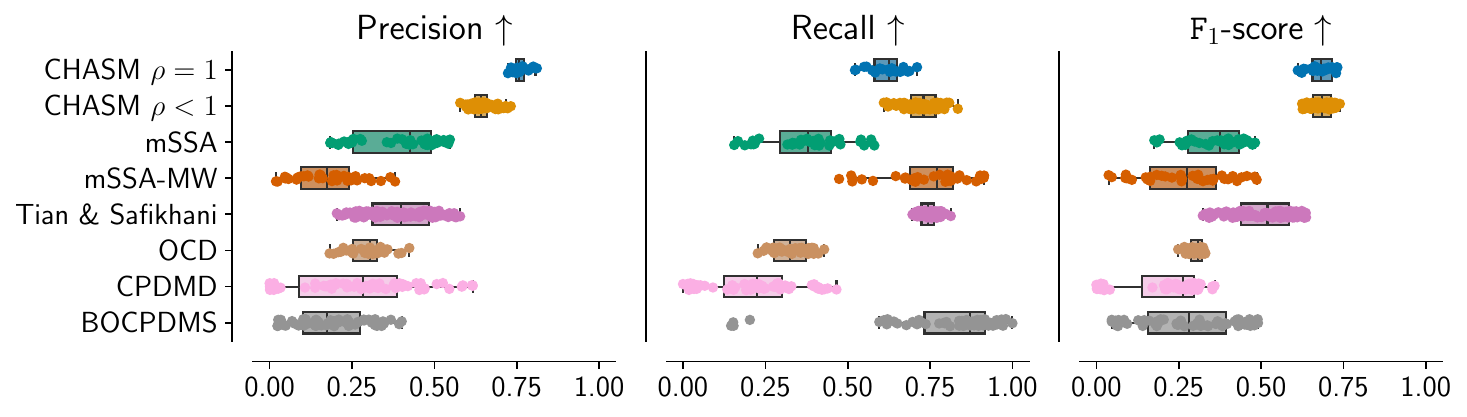}
	\includegraphics[width=0.6\linewidth]{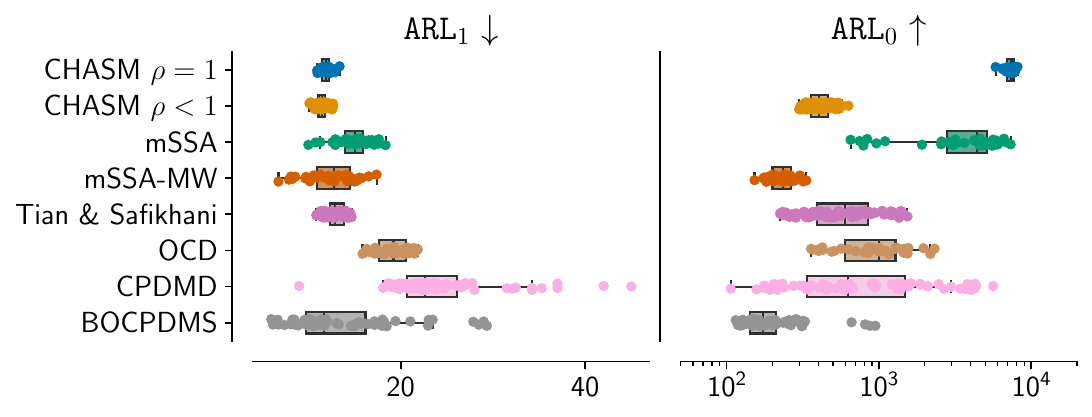}
	\caption{\small Accuracy (top) and speed (bottom) metrics (see Appendix~\ref{app:experiments-metrics}) across all parameter configurations on the synthetic bivariate $\varprocess{2}{1}$ data set with Laplace noise (see Appendix~\ref{app:experiments-synthetic-datasets}). Each dot corresponds to a single parameter set evaluated over the grid (see Table~\ref{tab:parameters}).}
	\label{fig:additional-laplace}
\end{figure}

\begin{figure}[p]
	\centering
	\includegraphics[width=0.8\linewidth]{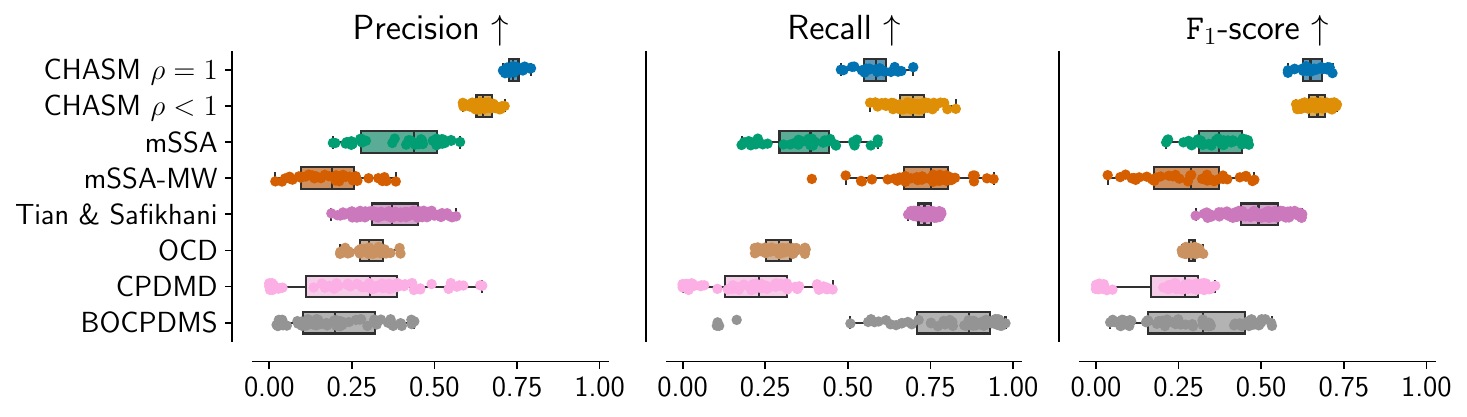}
	\includegraphics[width=0.6\linewidth]{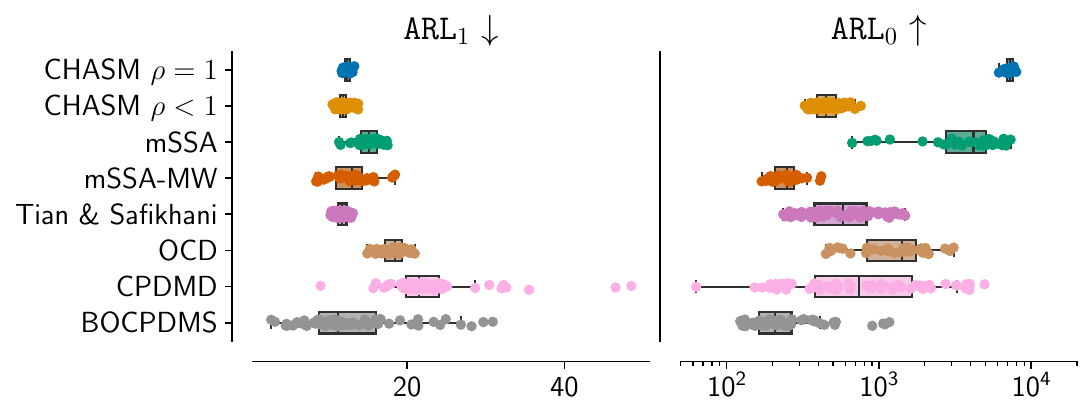}
	\caption{\small Accuracy (top) and speed (bottom) metrics (see Appendix~\ref{app:experiments-metrics}) across all parameter configurations on the synthetic bivariate $\varprocess{2}{1}$ data set with Student's $t_\nu$ noise (see Appendix~\ref{app:experiments-synthetic-datasets}). Each dot corresponds to a single parameter set evaluated over the grid (see Table~\ref{tab:parameters}).}
	\label{fig:additional-student}
\end{figure}

\begin{figure}[p]
	\centering
	\includegraphics[width=0.8\linewidth]{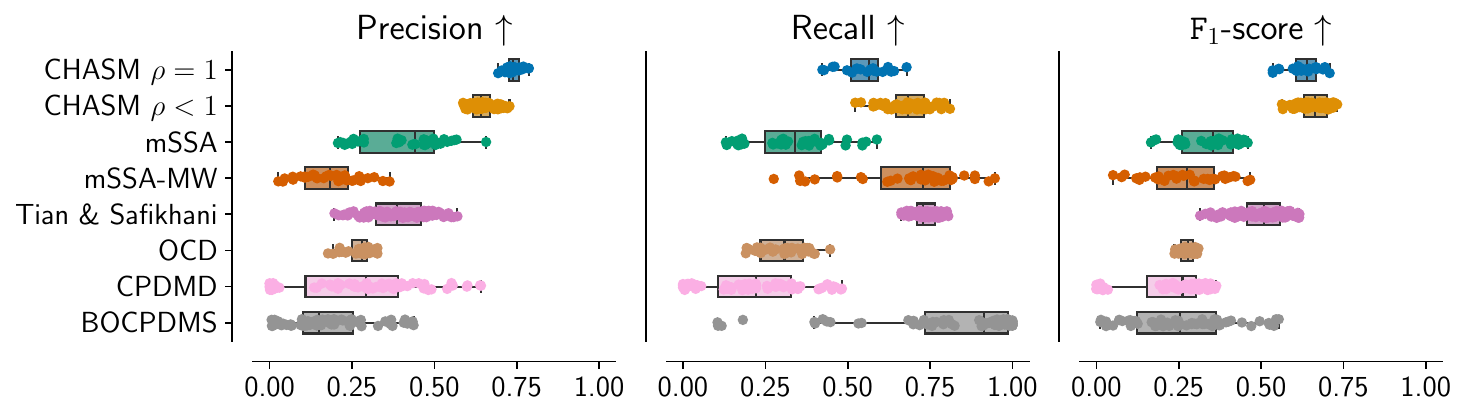}
	\includegraphics[width=0.6\linewidth]{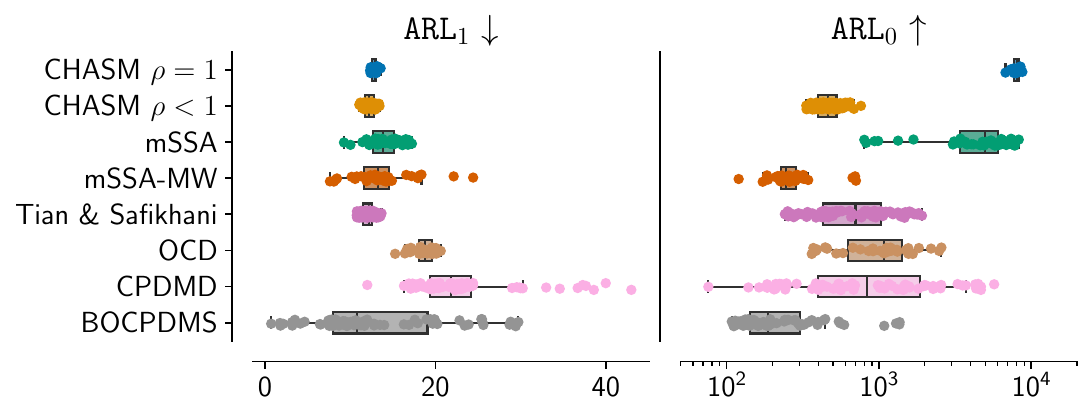}
	\caption{\small Accuracy (top) and speed (bottom) metrics (see Appendix~\ref{app:experiments-metrics}) across all parameter configurations on the synthetic bivariate $\varprocess{2}{1}$ data set with Huber-$\epsilon$ contaminated noise (see Appendix~\ref{app:experiments-synthetic-datasets}). Each dot corresponds to a single parameter set evaluated over the grid (see Table~\ref{tab:parameters}).}
	\label{fig:additional-huber}
\end{figure}

\subsubsection{Huber-$\epsilon$ contamination noise}

Figure~\ref{fig:additional-huber-eps} illustrates CHASM's overall superior performance overal all contamination levels $\epsilon$, especially for $\rho=1$, as it achieves the largest $\fone$-score and $\arlzero$ along with the lowest $\arlone$ values.

\begin{figure}[p]
	\centering
	\includegraphics[width=0.9\linewidth]{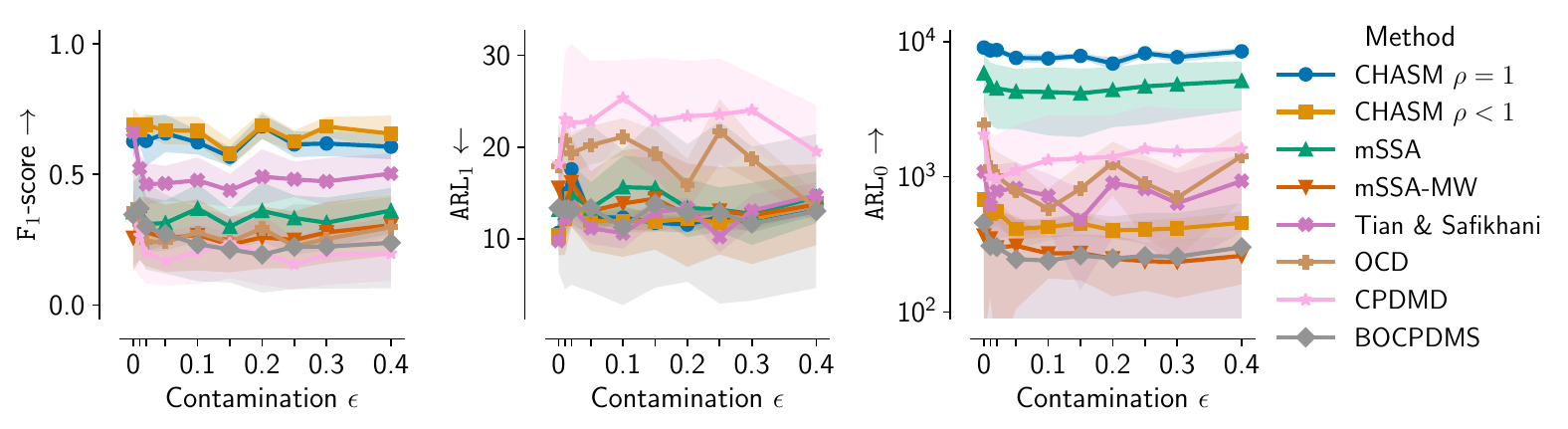}
	\caption{\small $\fone$, $\arlone$, and $\arlzero$ vs. contamination $\epsilon$ for the $\varprocess{2}{1}$ data set with Huber-$\epsilon$ contaminated noise. Shaded bands show mean $\pm$ std over the parameter grid.}
	\label{fig:additional-huber-eps}
\end{figure}

\subsubsection{Heavy tailed noise distribution}

Figure~\ref{fig:additional-student-df} reports accuracy and speed metrics as the degrees of freedom $\nu$ of the Student's $t_\nu$ noise distribution increases. As $\zeta \to \nu$, the noise approaches Gaussian, and performance of all methods tends to improve. CHASM maintains a superior joint trade-off between detection accuracy and speed throughout the entire range, with the advantage most pronounced at low $\nu$ where heavy tails are most severe. We note that $\nu < 8$ falls outside the finite eighth-moment assumption of Section~\ref{sec:theory}, yet CHASM remains competitive empirically in this regime.

\begin{figure}[p]
	\centering
	\includegraphics[width=0.9\linewidth]{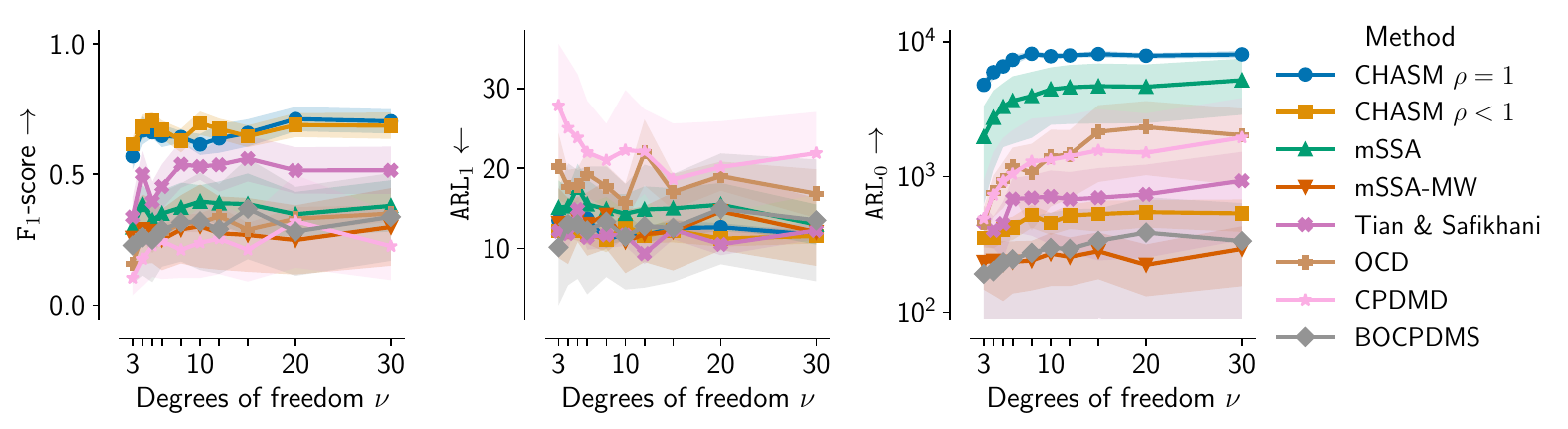}
	\caption{\small $\fone$, $\arlone$, and $\arlzero$ vs. degrees of freedom $\nu$ for the $\varprocess{2}{1}$ data set with Student's $t_\nu$ noise. Shaded bands show mean $\pm$ std over the parameter grid.}
	\label{fig:additional-student-df}
\end{figure}

\subsubsection{Synthetic high-dimensional time series}

Figures~\ref{fig:additional-sparse} and~\ref{fig:additional-fullrank} report accuracy and speed metrics across all parameter configurations for sparse and full-rank dynamics respectively. Under sparse dynamics, which correspond to the optimal setting for Tian \& Safikhani, CHASM achieves competitive performance, confirming that no accuracy is sacrificed relative to a method specifically designed for this regime. Under full-rank dynamics, where the low-rank assumption of Tian \& Safikhani is violated, their method degrades substantially, whereas CHASM maintains strong performance and dominates all competing methods across both accuracy and speed metrics.

\begin{figure}[p]
	\centering
	\includegraphics[width=0.8\linewidth]{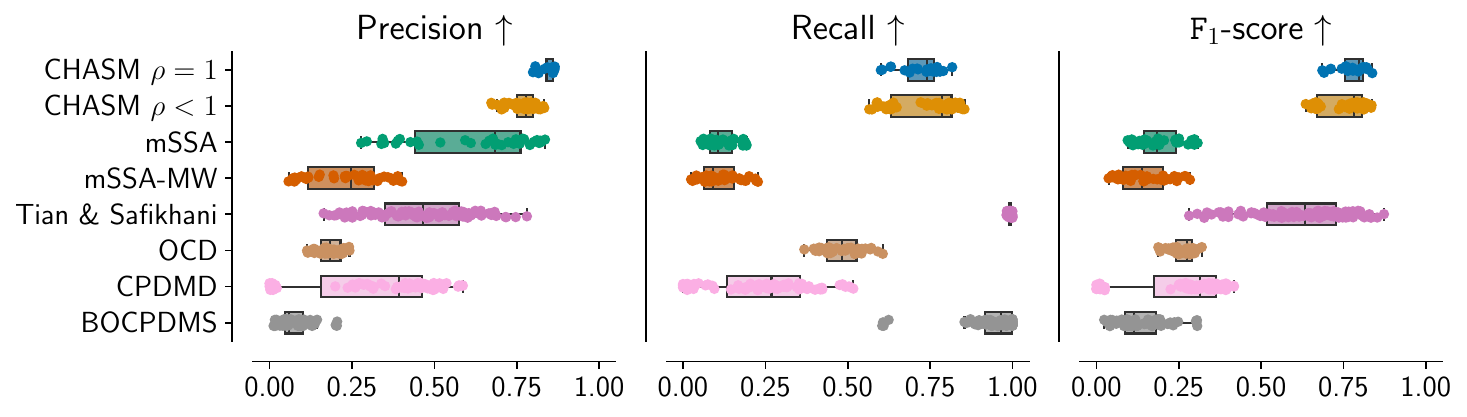}
	\includegraphics[width=0.6\linewidth]{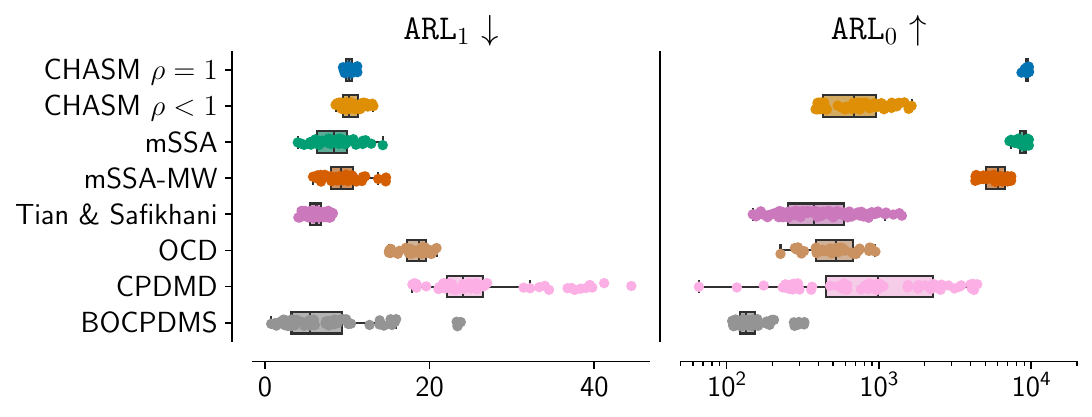}
	\caption{\small Accuracy (top) and speed (bottom) metrics (see Appendix~\ref{app:experiments-metrics}) across all parameter configurations on the synthetic bivariate $\varprocess{\dimension}{1}$ data set with sparse dynamics, $\dimension \in \{ 2, \dots, 40\}$ (see Appendix~\ref{app:experiments-synthetic-datasets}). Each dot corresponds to a single parameter set evaluated over the grid (see Table~\ref{tab:parameters}).}
	\label{fig:additional-sparse}
\end{figure}

\begin{figure}[t]
	\centering
	\includegraphics[width=0.8\linewidth]{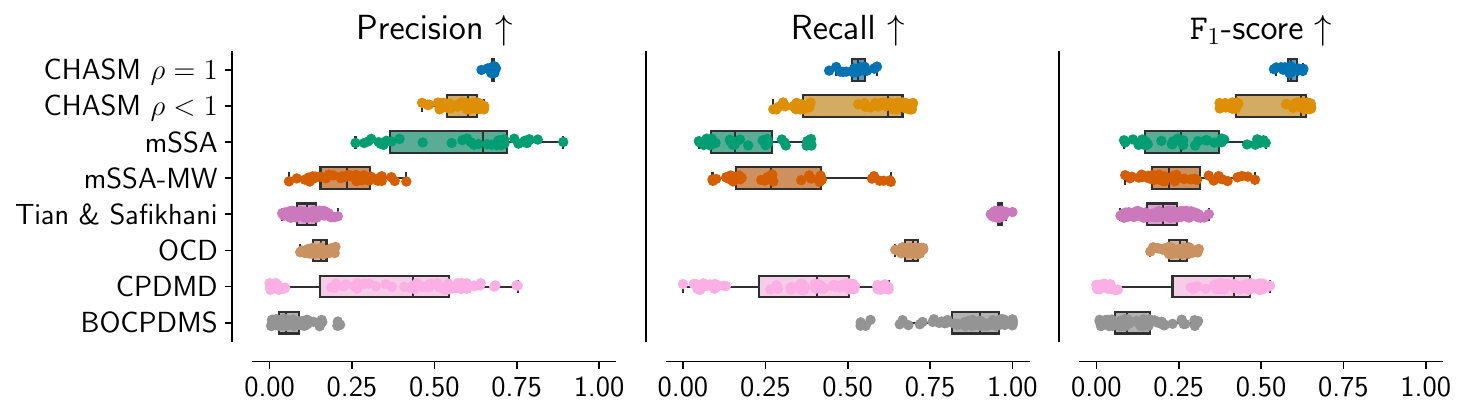}
	\includegraphics[width=0.6\linewidth]{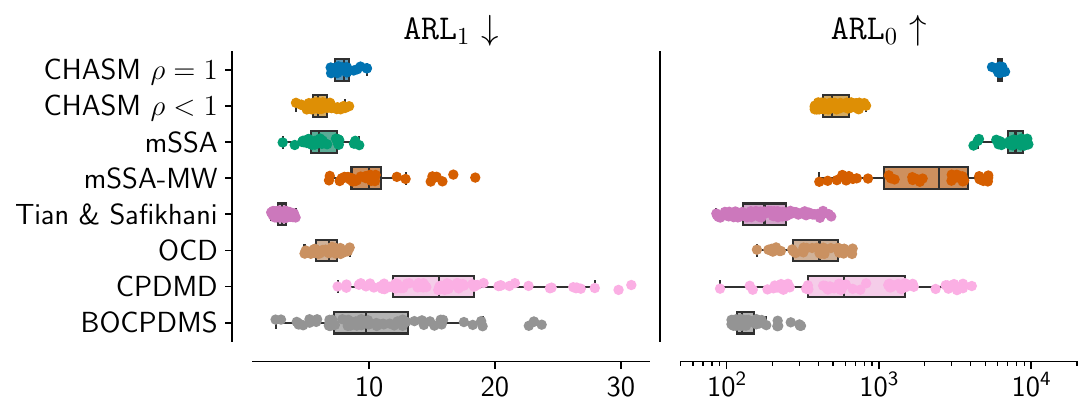}
	\caption{\small Accuracy (top) and speed (bottom) metrics (see Appendix~\ref{app:experiments-metrics}) across all parameter configurations on the synthetic bivariate $\varprocess{\dimension}{1}$ data set with full-rank dynamics, $\dimension \in \{ 2, \dots, 40\}$ (see Appendix~\ref{app:experiments-synthetic-datasets}). Each dot corresponds to a single parameter set evaluated over the grid (see Table~\ref{tab:parameters}).}
	\label{fig:additional-fullrank}
\end{figure}

\subsubsection{Performance with dimension}

Figure~\ref{fig:additional-sparse-dim} reports performance across all dimensions $d \in \{2, \dots, 40\}$ under sparse dynamics, complementing Figure~\ref{fig:dimension} in the main paper by including \textsc{mSSA}, \textsc{mSSA-MW}, and \textsc{Bocpdms} in the $\mathrm{ARL}_1$ panel. Figure~\ref{fig:additional-fullrank-dim} shows the analogous results under full-rank dynamics, where Tian \& Safikhani degrades sharply beyond $\dimension > 15$, while CHASM remains robust across the full dimensional range.

\begin{figure}[p]
	\centering
	\includegraphics[width=0.9\linewidth]{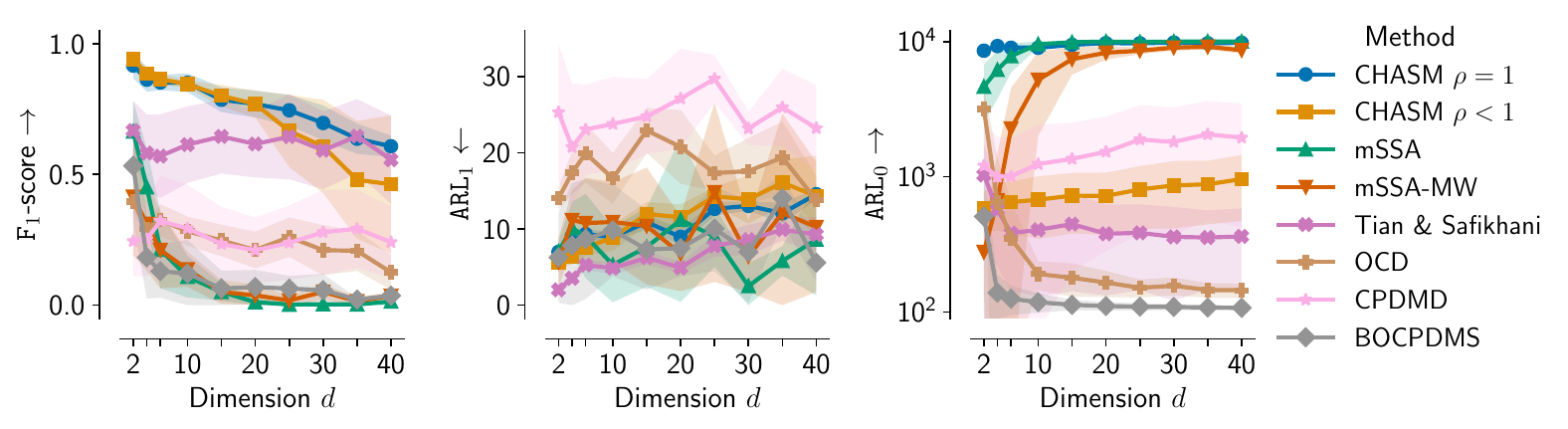}
	\caption{\small $\fone$, $\arlone$, and $\arlzero$ vs. dimension $\dimension$ for the sparse $\varprocess{\dimension}{1}$ data set. Shaded bands show mean $\pm$ std over the parameter grid.}
	\label{fig:additional-sparse-dim}
\end{figure}

\begin{figure}[p]
	\centering
	\includegraphics[width=0.9\linewidth]{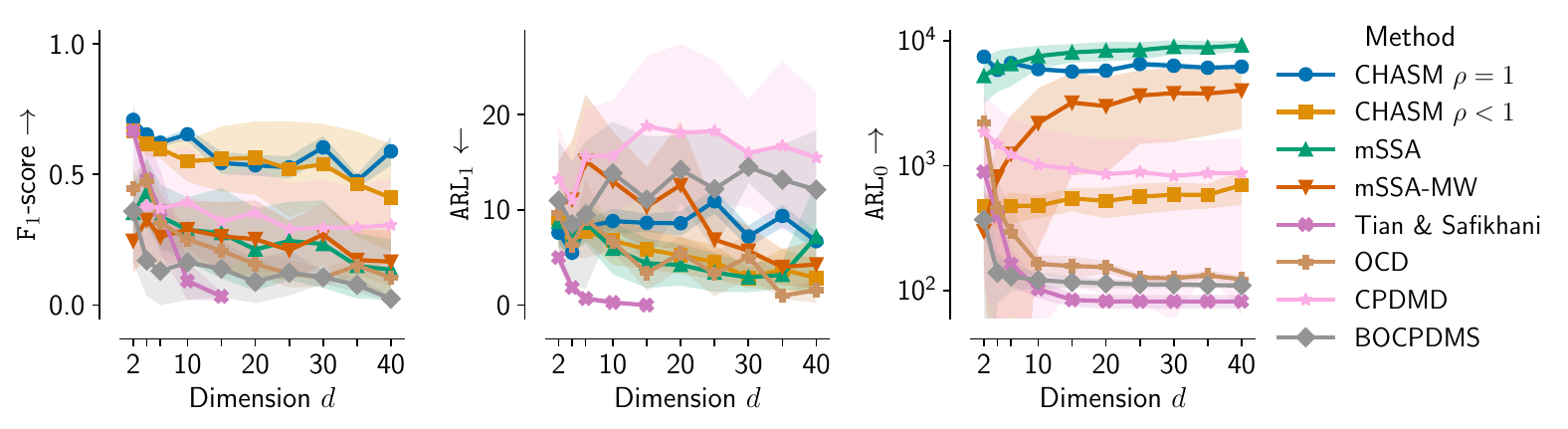}
	\caption{\small $\fone$, $\arlone$, and $\arlzero$ vs. dimension $\dimension$ for the full-rank $\varprocess{\dimension}{1}$ data set. Shaded bands show mean $\pm$ std over the parameter grid.}
	\label{fig:additional-fullrank-dim}
\end{figure}

\subsubsection{Real-world performance per data set}

Figures~\ref{fig:additional-hasc}, \ref{fig:additional-cifar}, \ref{fig:additional-newsgroups}, \ref{fig:additional-ucf} and~\ref{fig:additional-wikisection} report accuracy metrics across all parameter configurations on the HASC, CIFAR-100, 20 Newsgroups, UCF-Crime, and WikiSection real-world data sets respectively. Across these benchmarks, CHASM ($\rho < 1$) demonstrates practical advantages that were less apparent in the synthetic experiments, suggesting that adaptive downweighting of older observations is beneficial when the data-generating process exhibits more complex or gradual dynamics. Both regimes $\rho = 1$ and $\rho < 1$ remain competitive and are worth exploring in practice. Representative detection outputs are shown in Figures~\ref{fig:additional-cifar-det}, \ref{fig:additional-newsgroups-det}, \ref{fig:additional-ucf-det} and~\ref{fig:additional-wikisection-det}.

\begin{figure}[p]
	\centering
	\includegraphics[width=0.8\linewidth]{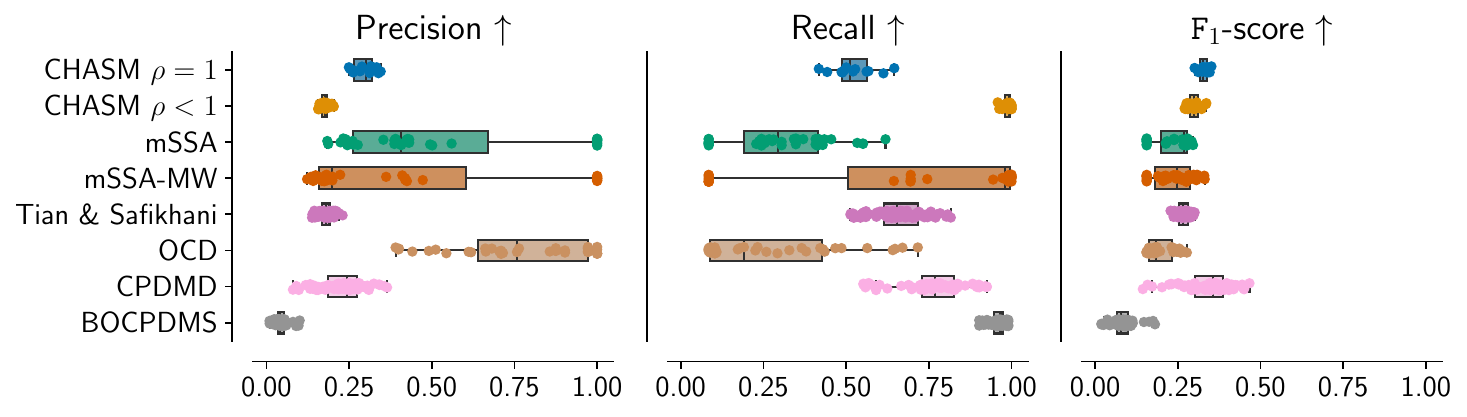}
	\caption{\small Accuracy metrics (see Appendix~\ref{app:experiments-metrics}) across all parameter configurations on the HASC real-world data set. Each dot corresponds to a single parameter set evaluated over the grid (see Table~\ref{tab:parameters}).}
	\label{fig:additional-hasc}
\end{figure}

\begin{figure}[p]
	\centering
	\includegraphics[width=0.8\linewidth]{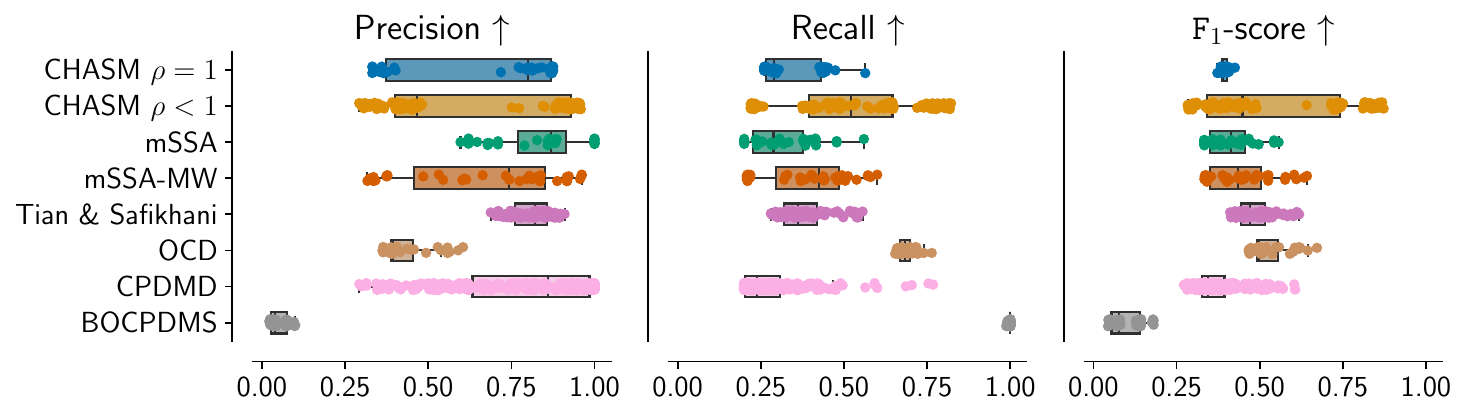}
	\caption{\small Accuracy metrics (see Appendix~\ref{app:experiments-metrics}) across all parameter configurations on the CIFAR-100 real-world data set. Each dot corresponds to a single parameter set evaluated over the grid (see Table~\ref{tab:parameters}).}
	\label{fig:additional-cifar}
\end{figure}

\begin{figure}[p]
	\centering
	\includegraphics[width=0.8\linewidth]{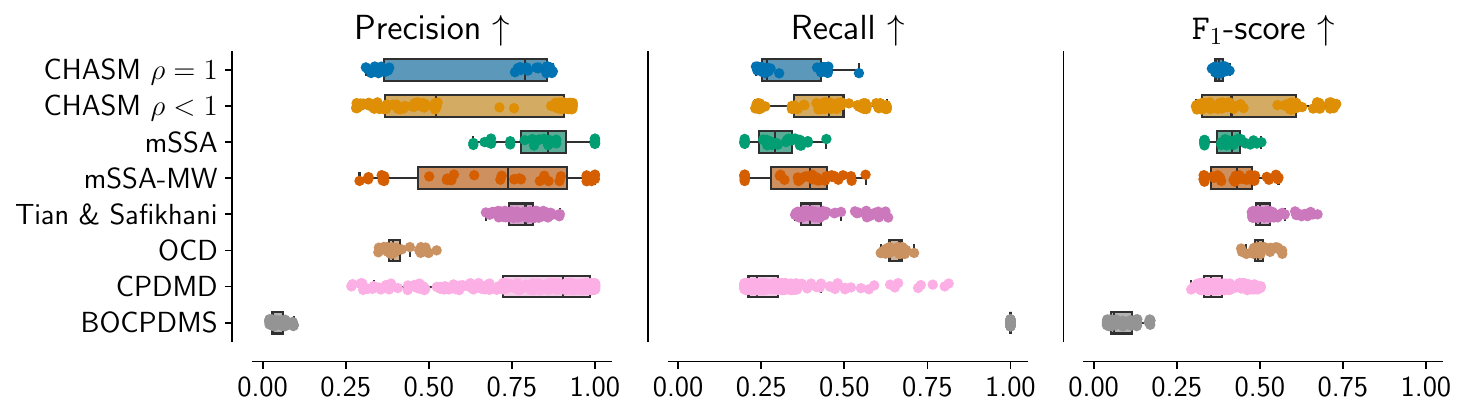}
	\caption{\small Accuracy metrics (see Appendix~\ref{app:experiments-metrics}) across all parameter configurations on the 20 Newsgroups real-world data set. Each dot corresponds to a single parameter set evaluated over the grid (see Table~\ref{tab:parameters}).}
	\label{fig:additional-newsgroups}
\end{figure}

\begin{figure}[p]
	\centering
	\includegraphics[width=0.8\linewidth]{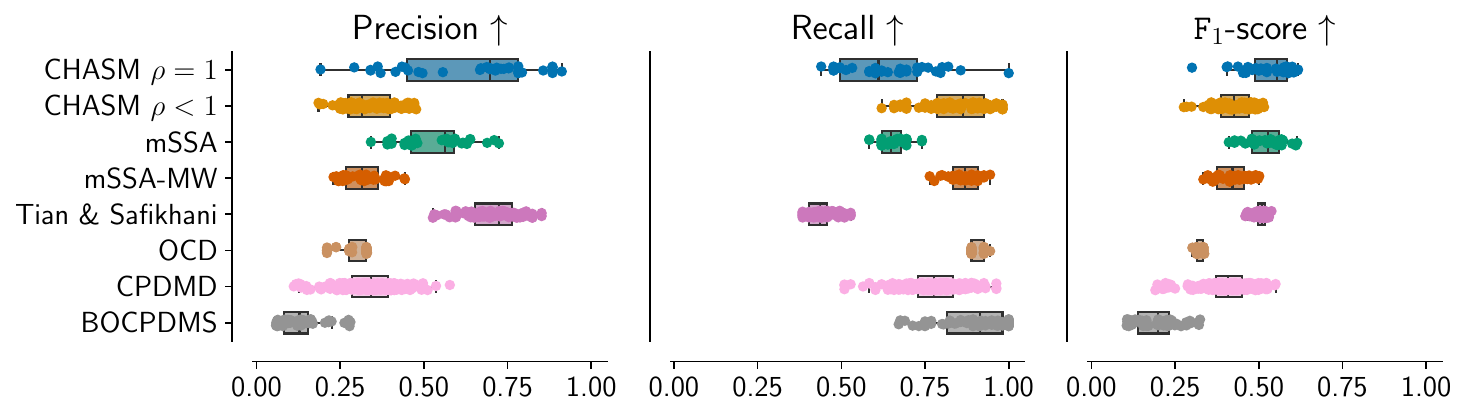}
	\caption{\small Accuracy metrics (see Appendix~\ref{app:experiments-metrics}) across all parameter configurations on the UCF-Crime real-world data set. Each dot corresponds to a single parameter set evaluated over the grid (see Table~\ref{tab:parameters}).}
	\label{fig:additional-ucf}
\end{figure}

\begin{figure}[p]
	\centering
	\includegraphics[width=0.8\linewidth]{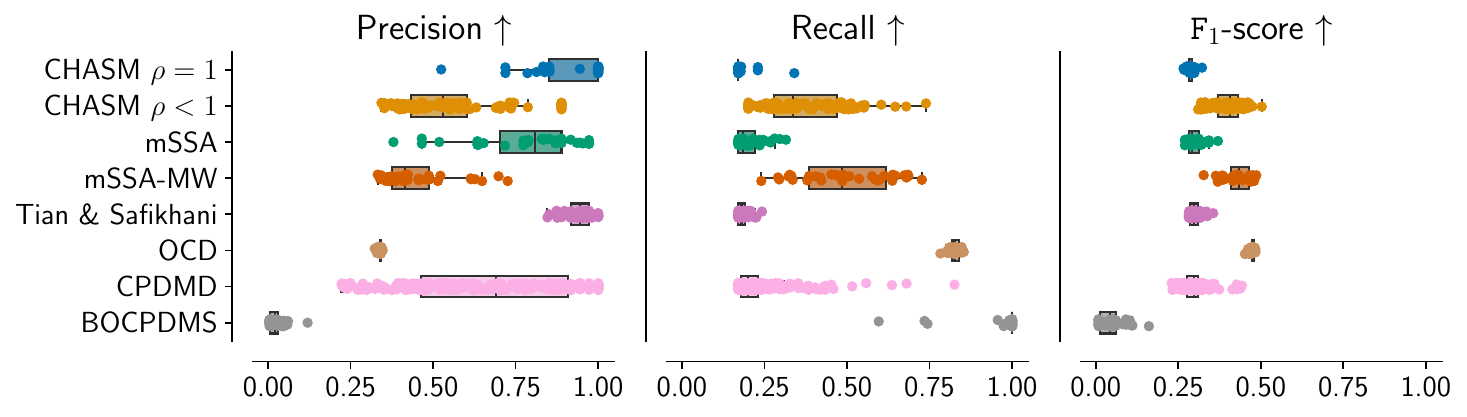}
	\caption{\small Accuracy metrics (see Appendix~\ref{app:experiments-metrics}) across all parameter configurations on the WikiSection real-world data set. Each dot corresponds to a single parameter set evaluated over the grid (see Table~\ref{tab:parameters}).}
	\label{fig:additional-wikisection}
\end{figure}

\begin{figure}[p]
	\centering
	\includegraphics[width=0.8\linewidth]{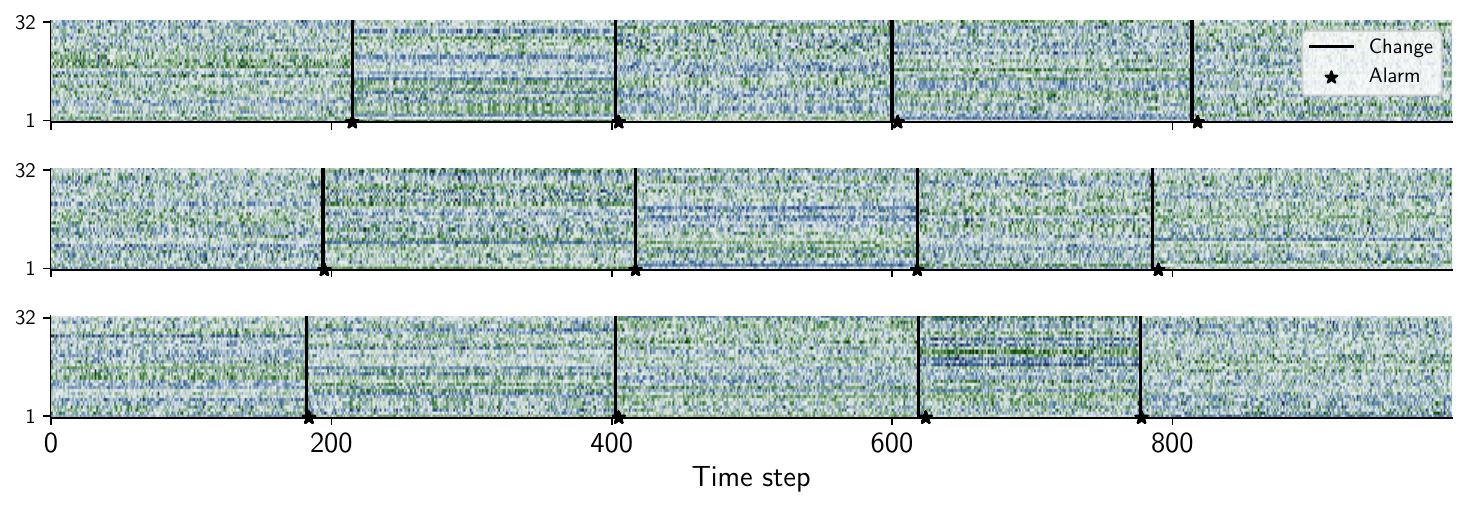}
	\caption{\small Three representative samples from the CIFAR-100 data set, visualised as heatmaps over the $\dimension = 32$ feature dimensions across time. Vertical black lines indicate true changepoints and star markers ($\star$) indicate CHASM alarms with parameters $\rho=0.98$, $\rank=4$, and $(\alpha, h)=(0.35, 28)$.}
	\label{fig:additional-cifar-det}
\end{figure}

\begin{figure}[p]
	\centering
	\includegraphics[width=0.8\linewidth]{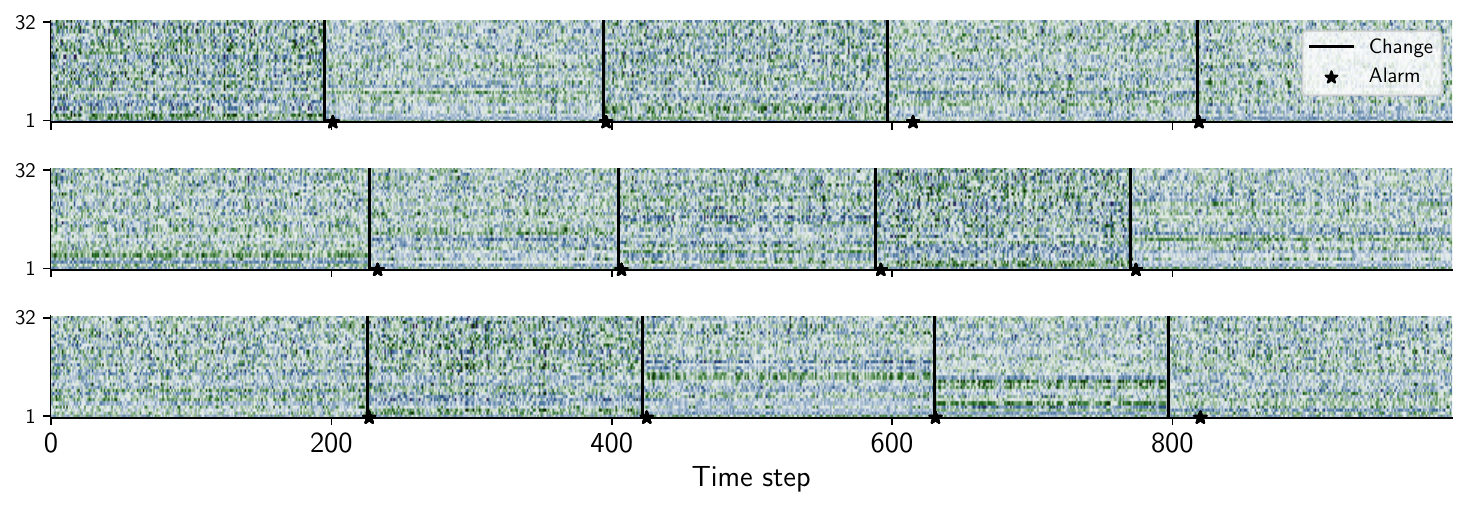}
	\caption{\small Three representative samples from the 20 Newsgroups data set, visualised as heatmaps over the $\dimension = 32$ feature dimensions across time. Vertical black lines indicate true changepoints and star markers ($\star$) indicate CHASM alarms with parameters $\rho=0.98$, $\rank=4$, and $(\alpha, h)=(0.35, 28)$.}
	\label{fig:additional-newsgroups-det}
\end{figure}

\begin{figure}[p]
	\centering
	\includegraphics[width=0.8\linewidth]{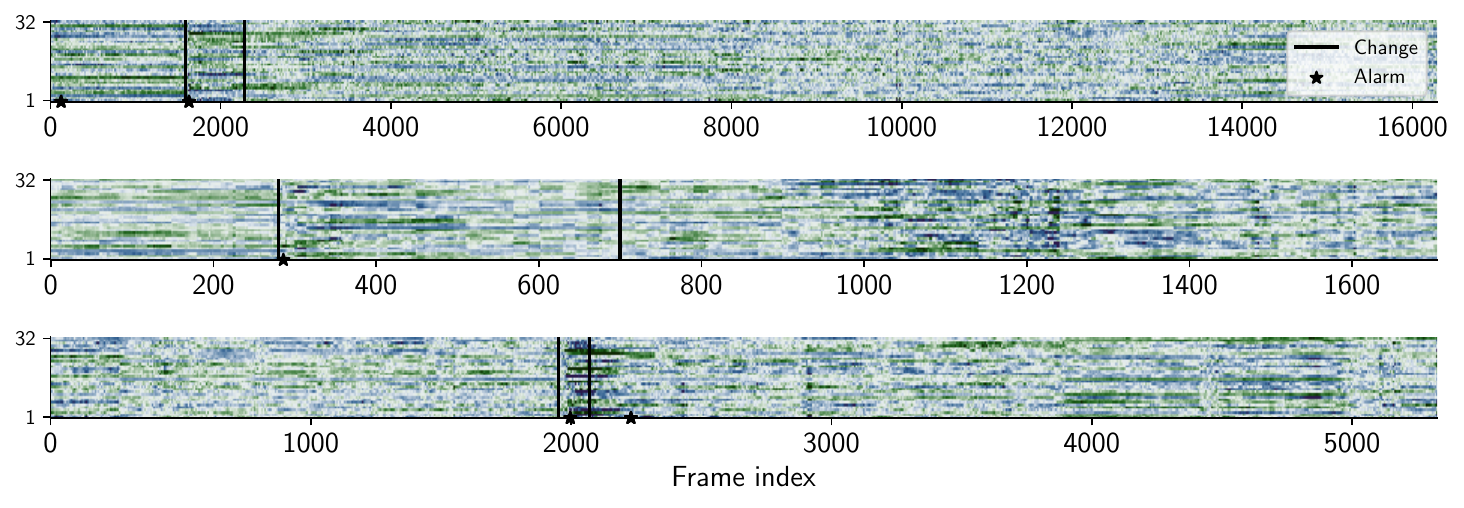}
	\caption{\small Three representative samples from the UCF-Crime data set, visualised as heatmaps over the $\dimension = 32$ feature dimensions across time. Vertical black lines indicate true changepoints and star markers ($\star$) indicate CHASM alarms with parameters $\rho=1$, $\rank=4$, and $(\alpha, h)=(0.18, 15)$. The data exhibits an asymmetry between onset and offset transitions: onset changepoints, corresponding to the beginning of anomalous activity, produce more pronounced shifts and are reliably detected, whereas offset changepoints are subtler and more frequently missed. We also note that, even for long sequences, CHASM is robust to false alarms.}
	\label{fig:additional-ucf-det}
\end{figure}

\begin{figure}[p]
	\centering
	\includegraphics[width=0.8\linewidth]{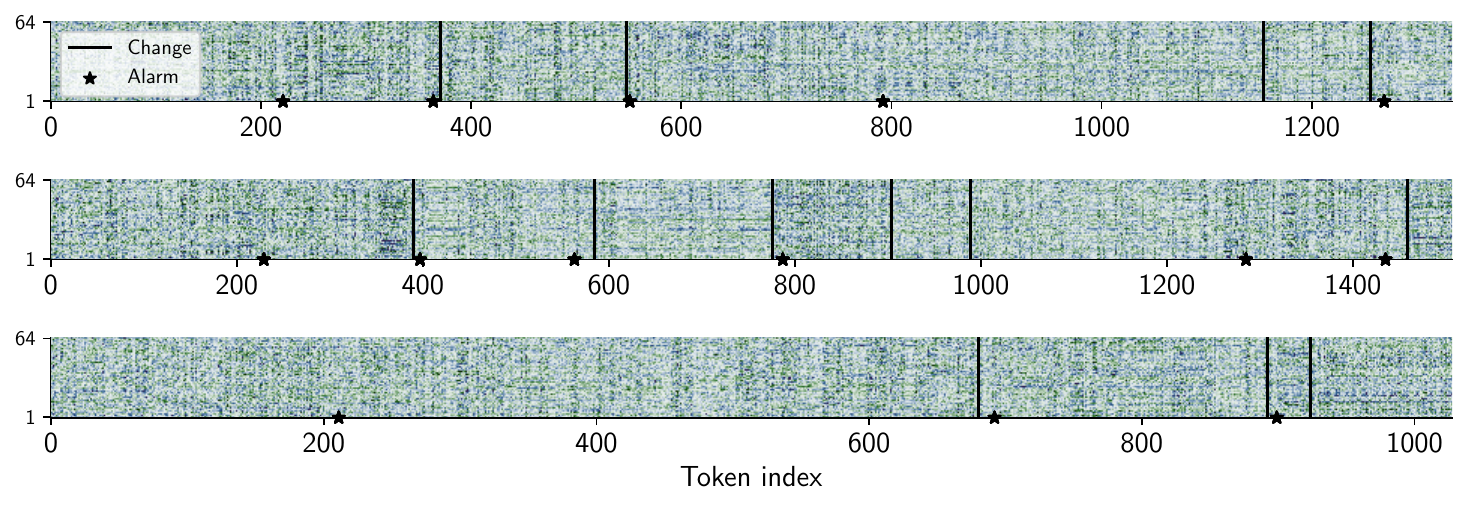}
	\caption{\small Three representative samples from the WikiSection data set, visualised as heatmaps over the $\dimension = 32$ feature dimensions across time. Vertical black lines indicate true changepoints and star markers ($\star$) indicate CHASM alarms with parameters $\rho=0.95$, $\rank=8$, and $(\alpha, h)=(0.25, 25)$. Changes in this data set are subtle, with little visible structure in the heatmaps, yet CHASM successfully detects a subset of changepoints with reasonable detection delay, demonstrating its ability to identify shifts in topic structure from sequential text embeddings.}
	\label{fig:additional-wikisection-det}
\end{figure}

\subsection{Computing resources}
\label{app:experiments-resources}

Experiments were conducted on a high-performance computing cluster comprising 325 compute nodes, each with two AMD EPYC 7742 processors (128 cores, 1\,TB RAM per node). Jobs used up to 100 cores in parallel with up to 200\,GB of memory, varying with the evaluated method; see Appendix~\ref{app:experiments-synthetic-complexity} for per-method time and memory complexity. Most jobs complete within a few hours. The exception is $\arlzero$ simulation for BOCPDMS, which requires on the order of days due to its increased computational cost.

\end{document}